%% file: main.tex
\documentclass[12pt]{article}
\usepackage{amsmath}
\usepackage{graphicx}
\usepackage{enumerate}
\usepackage{natbib}
\usepackage{url} 
\usepackage[flushleft]{threeparttable}
\usepackage{multirow}
\usepackage{verbatim}
\newcommand{\blind}{1}

\usepackage{varwidth}
\usepackage{amsmath}
\usepackage{numprint}
\usepackage{algorithm}
\usepackage{algpseudocode}
\usepackage{multirow}
\usepackage[pdfborder={0 0 0}, colorlinks=false]{hyperref}
\usepackage{mdframed}

\algrenewcommand\algorithmicrequire{\textbf{Input:}}
\algrenewcommand\algorithmicensure{\textbf{Output:}}
\input{definition}

\theoremstyle{definition}

\makeatletter

\newcommand{\Rmnum}[1]{\expandafter\@slowromancap\romannumeral #1@}
\makeatother

\usepackage{chngcntr}
\usepackage{apptools}
\AtAppendix{\counterwithin{lemma}{section}}

\usepackage{xr}

\begin{document}

\def\spacingset#1{\renewcommand{\baselinestretch}%
{#1}\small\normalsize} \spacingset{1}
\if1\blind
{
  \title{\bf \Large Model Checking for Regressions Based on Weighted Residual Processes with Diverging Number of Predictors}
  \author{Yue Hu\\
    Center for Data Science, Zhejiang University\\
    and \\
    Haiqi Li\\
    College of Finance and Statistics, Hunan University\\
    and \\
    Xintao Xia\\
    Center for Data Science, Zhejiang University}
  \maketitle
} \fi

\if0\blind
{
  \bigskip
  \bigskip
  \bigskip
  \begin{center}
    {\LARGE\bf Model Checking for Regressions Based on Weighted Residual Processes with Diverging Number of Predictors}
\end{center}
  \medskip
} \fi

\bigskip

\begin{abstract}
The integrated conditional moment (ICM) test is a classical and widely used method for assessing the adequacy of regression models. Although it performs well in fixed-dimension settings, its behavior changes dramatically when the predictor dimension diverges: in such regimes, the limiting null and alternative distributions of the ICM statistic degenerate to fixed constants. Moreover, when the number of predictors diverges, the commonly used wild bootstrap no longer approximates the null distribution of the ICM statistic well, leading to size distortion and substantial power loss. To address these challenges, we propose a new specification test based on weighted residual processes for evaluating the parametric form of the regression mean function in high-dimensional settings where the number of predictors increases with the sample size. We establish the asymptotic properties of the test statistic under the null hypothesis and under global and local alternatives. The proposed test maintains the nominal significance level and can detect local alternatives that deviate from the null hypothesis at the parametric rate $1/\sqrt{n}$. Furthermore, we propose a smooth residual bootstrap to approximate the limiting null distribution and establish its validity in high-dimensional settings. Two simulation studies and a real-data example are conducted to evaluate the finite-sample performance of the proposed test. 
\end{abstract}

\noindent%
\vfill

\newpage
\spacingset{1.9} 

\section{Introduction}
Regression models are fundamental tools for characterizing the relationship between a response variable and its predictors \citep{fan2014challenges}. Parametric regression models are particularly appealing due to their interpretability, computational efficiency, and well-established theoretical properties. However, when the parametric form is misspecified, subsequent statistical analysis and inference may be invalid, leading to misleading scientific conclusions. This highlights the importance of rigorously evaluating whether the specified parameter form is consistent with the data.

To formally examine model adequacy, we consider the following regression framework:
\begin{equation}\label{eq:model}
Y=m(X)+\varepsilon
\end{equation}
where $Y\in\mathbb{R}^{1}$ denotes the response variable associated with a $d$-dimensional predictor vector $X \in \mathbb{R}^{d}$, the regression function $m(\cdot)$ satisfies $m(x)=E(Y|X=x)$, and $\varepsilon$ is the error term satisfying $E(\varepsilon\mid X)=0$. We further assume that $\varepsilon$ is independent of $X$. Our goal is to test whether the unknown regression function $m(x)$ belongs to a prespecified parametric family $\mathcal{M} = \{m(\cdot, \beta) : \beta \in \Theta \subset \mathbb{R}^p  \}$, where $\Theta$ is the parameter space. Formally, we consider the following hypothesis testing problem:
\begin{equation}
\label{eq:test}
    \begin{split}
        H_0 &: \ \mathbb{P}\{ m(X) = m(X,\beta_0)\} = 1 \quad \text{for some } \beta_0 \in \Theta, \\
H_1 &: \ \mathbb{P}\{ m(X) = m(X,\beta)\} < 1 \quad \text{for all } \beta \in \Theta.
    \end{split}
\end{equation}
This formulation states that, under $H_0$, the regression mean function $m(\cdot)$ is correctly specified, in the sense that there exists $\beta_0\in\Theta$ such that $m(\cdot,\beta_0)$ coincides with the conditional expectation almost surely. Under $H_1$, no such parameter value exists, so the model is misspecified. Importantly, the hypothesis concerns only the specification of the conditional mean and not the full data-generating mechanism; the distribution of the error term $\varepsilon$ remains unrestricted. 

Moreover, modern statistical analyses increasingly involve settings where the predictor dimension $d$ and the parameter dimension $p$ grow with the sample size $n$ \citep{fan2010selective,hastie2015statistical}. The corresponding model specification problem is particularly important and comparatively less well studied. Consequently, developing rigorous and reliable goodness-of-fit tests in high-dimensional parametric regression has become a central problem in modern statistical methodology. In this paper, we consider a regime where both the predictor dimension $d$ and the parameter dimension $p$ diverge with the sample size $n$.

\subsection{Related literature}

There is a large literature on model specification tests when the predictor dimension is fixed. Existing methods can be broadly classified into two main categories. The first category contains \emph{locally smoothed tests} \citep{hardle1993comparing,horowitz1994testing,dette1999consistent}. For example, \cite{zheng1996consistent} utilized conditional moment restrictions with nonparametric kernel estimation to construct moment-based diagnostics. \citet{koul2004minimum} proposed a kernel-based test that used the minimized $L_2$ distance between a nonparametric estimator of the regression function and the fitted parametric model. \citet{guerre2005rate} proposed a data-driven smoothing-parameter selection criterion under which the local smoothing test is adaptively rate-optimal and consistent against Pitman local alternatives. \cite{van2008goodness} proposed to quantify the discrepancy between the empirical distribution functions of the parametric and nonparametric residuals using kernel regression. \citet{lavergne2012one} proposed a kernel-based test against a sequence of directional nonparametric alternatives to enhance power. \citet{guo2016model} proposed a dimension-reduction, model-adaptive local smoothing test for generalized linear models, which alleviates the curse of dimensionality.

The second category consists of \textit{global smoothing tests}, which replace the conditional mean independence condition $E(\varepsilon|X)=0$ with a family of parametric unconditional orthogonality conditions:
\begin{equation*}
 \mathbb{E}(\varepsilon|X)=0 \quad \Leftrightarrow \quad  \mathbb{E}\{\varepsilon w(X,t)\}=0,
\quad {\rm \forall \ t \in \Omega },
\end{equation*}
where $\Omega\subseteq\mathbb{R}$ is an index set and $w(\cdot,t)$ is a parametric family of weight functions. \cite{stinchcombe1998consistent} and \cite{escanciano2006consistent} provided conditions ensuring that the class $w(\cdot,t)$ is rich enough to guarantee the above equivalence. Common choices include the indicator weights $I(X\leq t)$ and the characteristic-function weight $\exp(it^{\top}X)$, where $i=\sqrt{-1}$ denotes the imaginary unit. \cite{bierens1982consistent,bierens1990consistent} proposed the \textit{integrated conditional moment} (ICM) test, which constructs a scalar statistic by integrating the resulting unconditional moment restrictions over $t$. See \cite{stute1997nonparametric}, \cite{bierens1997asymptotic}, \cite{stute1998model}, \cite{khmaladze2004martingale} for further developments on global smoothing tests.

Although all the above methods perform well when the predictor dimension is fixed, their power can deteriorate markedly as the dimension increases, reflecting the well-known \textit{curse of dimensionality}. Local smoothing–based tests typically rely on kernel smoothing or kernel regression, and it is well known that the performance of kernel-based methods can deteriorate even in moderate dimensions. For global smoothing tests, the test statistics often admit a pairwise-sum representation with kernel-type weights depending on $\|X_j - X_k\|^2$, as illustrated in \eqref{eq:test_icm}. In high dimensions, interpoint distances tend to concentrate, leading to severe data sparsity, reduced discrimination in the kernel weights, and substantial loss of testing power. 

In high-dimensional settings, the model specification test is comparatively less well studied, and only a limited number of test statistics have been developed to assess the adequacy of parametric regression models. \citet{tan2019adaptive, tan2025weighted} proposed tests built on sufficient dimension reduction (SDR) for single-index and multi-index models, respectively. These approaches are most effective when the regression function admits a low-dimensional index representation. However, because such constructions are intrinsically tied to dimension-reduction structures, they do not extend naturally to general parametric models. \cite{tan2025weighted} proposed weighted residual empirical process-based tests for general parametric regression models using indicator weights. In this work, we propose a novel test statistic that controls type I error and achieves reasonable power against alternatives in high-dimensional settings.

\subsection{Our contributions}

To address the degeneracy of classical ICM tests and the failure of the standard wild bootstrap in high dimensions \citep{bierens1982consistent,tan2022integrated}, we develop a new goodness-of-fit procedure that remains valid as the dimension diverges. The proposed test avoids degeneracy under both the null and alternative hypotheses and remains effective in high-dimensional settings. We further show that the proposed test remains consistent against fixed alternatives and can detect local alternatives approaching the null at the parametric rate $1/\sqrt{n}$ in the diverging dimension setting. Since the limiting null distribution is not asymptotically distribution-free, we further introduce a smooth residual bootstrap and establish its validity for consistently approximating the null distribution in high-dimensional settings. Extensive simulations verify the theoretical findings, and an application to a real dataset illustrates the practical utility of our method.

\subsection{Organization}
The rest of the paper is organized as follows. In Section~\ref{sec:prelim}, we review the background of the ICM test and explain why it encounters difficulties in high-dimensional settings. In Section~\ref{sec:method}, we introduce the proposed test based on weighted residual processes and develop a smooth residual bootstrap procedure to approximate the null distribution of the test statistic. Section~\ref{sec:theory} establishes the asymptotic properties of both the test statistic and the bootstrap distribution under the null and alternative hypotheses. We also propose a data-driven strategy for selecting the weight function under directional and nonparametric alternatives. In Section~\ref{sec:num}, we report simulation results and a real-data example to evaluate the finite-sample performance of the proposed method. Section~\ref{sec:con} concludes with a discussion of possible future directions. All technical proofs are deferred to the appendix.

\section{Preliminary}
\label{sec:prelim}
In this section, we briefly review the integrated conditional moment (ICM) approach for testing the adequacy of the regression model. Let $(Y,X)\in \mathbb{R}\times\mathbb{R}^d$ be a random vector that satisfies the regression model \eqref{eq:model}. We observe \textit{independent and identically distributed} (i.i.d.) samples $\{(Y_i,X_i)\}_{i=1}^n$ from the joint distribution of $(Y,X)$. Let $\tilde{\beta}_0$ denote the population least-squares estimator of the regression function onto the parametric model class $\Theta$,
\begin{equation}\label{eq:population_ols}
\widetilde{\beta}_0=\operatorname{argmin}_{\beta\in \Theta} \mathbb{E}\{Y-m(X,\beta)\}^2 = \operatorname{argmin}_{\beta\in \Theta}\mathbb{E}\{m(X)-m(X,\beta)\}^2
\end{equation}
and we define the approximation residual by $e=Y-m(X,\tilde{\beta}_0)$. 

Under the null hypothesis, we have $m(X,\widetilde{\beta}_0)=m(X)$, and therefore $e=\varepsilon$, which satisfies $E(e|X)=0$. The central idea of ICM \citep{bierens1982consistent} is to convert the conditional moment restriction $E(e|X)=0$ into a continuum of unconditional moment restrictions through an appropriate class of weight functions. Specifically, \citet{bierens1982consistent} proposed using the characteristic function as the weight function, $w(X,t)=\exp(i t^\top X)$, where $i=\sqrt{-1}$ denotes the imaginary unit. Let $\widehat{\beta}_n$ denote the least squares estimator of $\widetilde{\beta}_0$, defined by
\begin{equation*}
\widehat{\beta}_n =\operatorname{argmin}_{\beta\in \Theta} \sum_{i=1}^{n}\{Y_i-m(X_i,\beta)\}^2.
\end{equation*}
and we define the fitted residuals by $\widehat{e}_i=Y_i-m(X_i,\widehat{\beta}_n)$. The ICM statistic is defined as
\begin{equation}
\label{eq:icm}
\text{ICM}_n = \int_{\Gamma} \left| \frac{1}{\sqrt{n}} \sum_{j=1}^{n} \widehat{e}_j \exp(i t^\top \Phi(X_j)) \right|^2 d\mu(t),
\end{equation}
where $\widehat{e}_j$ denotes the fitted residual, $\Gamma=\prod_{i=1}^{d}[-\epsilon_i,\epsilon_i]$ for some $\epsilon_i>0$, $\mu(\cdot)$ is the Lebesgue measure on $\Gamma$, and $\Phi(\cdot)$ is a bounded smoothing transformation from $\mathbb{R}^d$ to $\mathbb{R}^d$. The limiting distribution of $\text{ICM}_n$ differs under the null and alternative hypotheses, which allows the test to discriminate between them.

Although effective, implementing the ICM statistic involves numerical integration over a $d$-dimensional region, leading to a rapidly increasing computational burden as $d$ increases. To overcome such difficulty, \cite{escanciano2006consistent, lavergne2008breaking} proposed replacing the Lebesgue measure on a compact set with the standard normal measure over $\mathbb{R}^d$. Specifically, when $\mu$ is taken to be the standard normal distribution on $\Gamma=\mathbb{R}^d$, and $\Phi(\cdot)$ is the identity map, the ICM statistic admits the simplified form
\begin{equation}
\begin{split}
\label{eq:test_icm}
\text{ICM}_n &= \int_{\mathbb{R}^d} \left| \frac{1}{\sqrt{n}} \sum_{j=1}^{n} \widehat{e}_j \exp(i t^\top X_j) \right|^2 \phi(t) dt \\
&= \frac{1}{n} \sum_{j,k=1}^{n} \widehat{e}_j \widehat{e}_k \exp \Big(-\frac{1}{2} \|X_j - X_k\|^2 \Big),
\end{split}
\end{equation}
where $\phi(t)$ denotes the standard normal density on $\mathbb{R}^d$.

When the dimension is fixed, the asymptotic properties of the ICM test are well-established \citep{bierens1997asymptotic}. Under the null hypothesis, the test statistic converges in distribution to a nondegenerate distribution, while under the fixed alternative, it diverges to infinity in probability. However, the asymptotic null distribution depends on the unknown joint distribution of $(X,Y)$ and is thus not asymptotically pivotal. Nevertheless, it can be consistently approximated by the wild bootstrap; see, for example, \citet{stute1998bootstrap} and \citet{dominguez2005power}.

When the dimension diverges as the sample size increases, the asymptotic behavior of the ICM statistic changes substantially. To see this, note that the pairwise-sum representation form of ICM in \eqref{eq:test_icm} contains exponential weights of the form, $\exp (-\|X_j - X_k\|^2/2)$, whose expectation decays exponentially fast in $p$ under mild regularity conditions. Consequently, the random fluctuation of the test statistic vanishes as the dimension increases, and the asymptotic distribution of the ICM test statistics becomes degenerate. Formally, \cite{tan2025weighted} showed that, under the conditions $p\to\infty$, $p^2/n\to 0$ and $\log(n)/p\to 0$, the ICM statistic converges in probability to a finite constant under both the null and alternative hypotheses. Similarly, when the dimension diverges from the sample size, the wild bootstrap version of the ICM statistic converges to the same finite constant as the original statistic. Thus, the bootstrap critical values are no longer valid, and the test fails to maintain its nominal size.

To address the degeneracy of the classical ICM statistic and the invalidity of the wild bootstrap in high-dimensional settings, we develop a new test based on weighted residual processes. Unlike the classical ICM approach, our construction is built on a one-dimensional function of the residuals rather than on high-dimensional functions of the covariates. With a suitable choice of weight function, the resulting statistic depends asymptotically only on one-dimensional quantities, thereby substantially avoiding the curse of dimensionality and ensuring a nondegenerate null limiting distribution even as the dimension diverges with the sample size. Since the null limiting distribution is not asymptotically distribution-free, we propose a smoothed residual bootstrap procedure for approximating critical values and establish its validity under standard regularity conditions.

\section{Construction of the Test Statistic and Bootstrap Approximation}
 \label{sec:method}
In this section, we introduce the proposed test statistic, explain its construction, and develop a smooth residual bootstrap procedure to approximate the critical values. By the definition of the population least-squares estimator $\tilde{\beta}_0$ in \eqref{eq:population_ols}, the null hypothesis in \eqref{eq:test} can be equivalently written as
\begin{eqnarray*}
H_0:\ \mathbb{P}\{m(X)=m(X,\tilde{\beta}_0)\}=1.
\end{eqnarray*}
Under the null hypothesis, we have $\beta_0=\widetilde{\beta}_0$ and the regression residual $e:=Y-m(X,\tilde{\beta}_0)=Y-m(X,\beta_0)=\varepsilon$ is independent of $X$. The classical ICM test statistic in \eqref{eq:icm} is motivated by the moment condition
\begin{equation*}
    \mathbb{E}\{\varepsilon \exp(it^{\top}X)\}
    =\mathbb{E}\bigl\{\mathbb{E}(\varepsilon\mid X)\exp(it^{\top}X)\bigr\}
    =0
\end{equation*}
under the null hypothesis. However, when $X$ is high-dimensional, the resulting procedure requires integration over a $p$-dimensional domain and suffers from the curse of dimensionality.

To overcome the difficulties of the classical ICM test in high-dimensional settings, we consider a real-valued weight function $g(X):\mathbb{R}^{d}\to\mathbb{R}$ and write $g_0(X)=g(X)-\mathbb{E}\{g(X)\}$. We consider the random variable $g_0(X)\exp(it e).$ Since the covariate enters only through the real-valued weight function $g_0(X)$, rather than through the high-dimensional transform $\exp(it^{\top}X)$, the resulting procedure avoids the curse of dimensionality and the integration over a $p$-dimensional space. We now study the behavior of this quantity under the null and alternative hypotheses. Under $H_0$, we have $e=Y-m(X,\widetilde{\beta}_0)=\varepsilon$ is independent of $X$, by \eqref{eq:model}. Hence,
\begin{eqnarray*}
\mathbb{E}\{g_0(X)\exp(it e)\}=\mathbb{E}\{g_0(X)\}\mathbb{E}\{\exp(it\varepsilon)\}=0,\quad \text{for all } t\in\mathbb{R}.
\end{eqnarray*}
Under the alternative hypothesis $H_1$, we have $e=m(X)-m(X,\tilde{\beta}_0)+\varepsilon $. We can choose a proper weight function $g(X)$ such that
\begin{eqnarray*}
\mathbb{E}\{g_0(X) \exp(ite)\} = \mathbb{E}\big[g_0(X) \exp[it\{ m(X)-m(X,\tilde{\beta}_0)+\varepsilon\}]\big] \neq 0,
\end{eqnarray*}
for some $t \in \mathbb{R}$. These different behaviors under the null and alternative hypotheses provide the key motivation for the proposed test statistic.

We then provide details of the proposed test statistics. Let $\{(X_i,Y_i)\}_{i=1}^n$ be independent and identically distributed observations from $(X,Y)$. 
Note that the constructions based on $\exp(it\varepsilon)$ and $\cos(t\varepsilon)+\sin(t\varepsilon)$ yield the same test statistic in \eqref{eq:pro_test} whenever $\varphi(t)$ is even. To avoid working with the complex-valued function $\exp(it\varepsilon)$, we therefore assume that $\varphi(t)$ is even and use the real-valued transformation $\cos(t\varepsilon)+\sin(t\varepsilon)$. This leads to the residual empirical process
\begin{equation}\label{eq:hatu}
    \widehat{U}_n(t)=\frac{1}{\sqrt{n}}\sum_{i=1}^{n}\{g(X_i)-\bar{g}\}\{\cos(t\widehat{e}_i)+\sin(t\widehat{e}_i)\},
\end{equation}
where
\begin{equation*}
\bar{g}=n^{-1}\sum_{i=1}^{n} g(X_i), \qquad
\widehat{e}_i=Y_i-m(X_i,\widehat{\beta}_n),
\end{equation*}
with $\widehat{\beta}_n$ being the least squares estimator of $\tilde{\beta}_0$. However, $\widehat{U}_n(t)$ cannot be used directly as a test statistic, because it is a process indexed by $t\in\mathbb{R}$. We therefore aggregate the information in $\widehat{U}_n(t)$ over $t$ using an even weight function $\varphi(t)$, and define the test statistic as
\begin{equation}\label{eq:pro_test}
\text{WICM}_n =  \int_{\mathbb{R}}	|\widehat{U}_{n}(t)|^2 \varphi(t)dt.
\end{equation}

As established in Theorem \ref{theorem:1}, the limiting null distribution of the proposed test statistic $\text{WICM}_n$ is not asymptotically pivotal because it depends on the unknown residual distribution. We therefore propose a smooth residual bootstrap procedure to approximate the null distribution of $\text{WICM}_n$ and compute the corresponding critical values; see \citet{dette2007new}, \citet{neumeyer2009smooth}, and \citet{tan2025weighted}. The procedure is summarized below.

\begin{enumerate}
\item In the $j$th bootstrap replication, generate the bootstrap errors 
\begin{equation*}    \varepsilon_{i,j}^{*}=\tilde{\varepsilon}_{i,j}^{*}+v_n z_{i,j}, \qquad 1\le i\le n,
\end{equation*}
where $\tilde{\varepsilon}_{i,j}^*$ are sampled with replacement from the centered residuals $\tilde{\varepsilon}_i=\widehat{e}_i-\sum_{i=1}^{n}\widehat{e}_i/n$, $v_n$ is a smoothing parameter, and $z_{1,j}, \cdots, z_{n,j}$ are independent and identical centered random variables with continuous density $l(\cdot)$.

\item Generate the corresponding bootstrap responses by $Y_{i,j}^* = \widehat{m}(X_i,\widehat{\beta}_n) + \varepsilon_{i,j}^*$. Let $\widehat{\beta}_{n,j}^{*}$ be the least squares estimator based on the bootstrap sample $\{(Y^*_{i,j},X_i)\}_{i=1}^n$.

\item Define the bootstrap version of the test statistic by
$$\text{WICM}_{n,j}^*=\int_{\mathbb{R}}	|U^*_{n,j}(t)|^2 \varphi(t)dt,$$
where 
$$U^*_{n,j}(t)=\frac{1}{\sqrt{n}}\sum_{i=1}^{n}(g(X_i)-\bar{g})\exp(it\widehat{\varepsilon}^*_{i,j})$$ with $\widehat{\varepsilon}^{*}_{i,j}=Y_{i,j}^* - \widehat{m}(X_i,\widehat{\beta}^*_{n,j})$.

\item Repeat Steps 1--3 independently for $B$ bootstrap replications, where $B$ is chosen to be sufficiently large. Let $\text{WICM}_{n}^*$ denote the resulting bootstrap distribution of the test statistic. For a given significance level $\alpha$, define the critical value $\widehat{c}_{\alpha}$ as the upper $\alpha$-quantile of $\text{WICM}_{n}^*$. The null hypothesis is rejected when $\text{WICM}_n \ge \widehat{c}_{\alpha}$.
\end{enumerate}

Although the test statistic $\text{WICM}_n$ involves only one-dimensional numerical integral and thus avoids the curse of dimensionality, the integral admits a further simplification. For any even function $\varphi(t)$, we have
\begin{equation*}
    \begin{split}
        \text{WICM}_n &= \frac{1}{n}\sum_{j,k=1}^{n} \{g(X_j)-\bar{g}\}\{g(X_k)-\bar{g}\}  \int_{\mathbb{R}} \cos\{(\widehat{e}_j- \widehat{e}_k)t \} \varphi(t) dt\\
        & = \frac{1}{n}\sum_{j,k=1}^{n} \{g(X_j)-\bar{g}\}\{g(X_k)-\bar{g}\}  M_{\varphi}\{(\widehat{e}_j- \widehat{e}_k)\},
    \end{split}
\end{equation*}
where $M_{\varphi}(\varepsilon)= \int_{\mathbb{R}} \cos( \varepsilon t)\varphi(t)dt$. For a suitable choice of $\varphi(t)$, $M_{\varphi}(\widehat{e}_j-\widehat{e}_k)$ can be evaluated in closed form. Many weight functions $\varphi(t)$ are available for this purpose. In the numerical studies, we choose $\varphi(t)$ to be the standard normal density; see \citet{tan2025weighted} and \citet{escanciano2006consistent} for related discussions. It then follows that
\begin{equation*}
\text{WICM}_n = \frac{1}{n}\sum_{j,k=1}^{n}\{g(X_j)-\bar{g}\}\{g(X_k)-\bar{g}\}\exp\bigg\{-\frac{1}{2}(\widehat{e}_j-\widehat{e}_k)^2\bigg\}.
\end{equation*} 

\section{Theoretical Results}
\label{sec:theory}
This section studies the asymptotic behavior of the proposed test statistic under the null and alternative hypotheses and establishes the validity of the bootstrap approximation.

\subsection{Asymptotic properties of the test statistic}

We first establish the asymptotic properties of the test statistic $\text{WICM}_n$ under the null hypothesis, fixed alternatives, and local alternatives. In the assumptions below, we use $F(X)$ to denote a nonnegative measurable envelope function such that $\mathbb{E}\{F(X)^4\}<\infty$. We use the following regularity conditions.

\begin{assumption}
\label{as:model}
The function $m(X,\beta)$ is twice continuously differentiable in $\beta$. Let
$$\dot{m}(X,\beta)=\frac{\partial m(X,\beta) }{\partial \beta}, \ddot{m}(X,\beta)=\frac{\partial \dot{m}(X,\beta) }{\partial \beta^{\top}}.$$
Let $\dot{m}_k(X,\beta)$ denote the $k$th component of $\dot{m}(X,\beta)$ and $\ddot{m}_{kl}(X,\beta)$ denote the $(k,l)$th entry of $\ddot{m}(X,\beta)$, for $1\le k,l\le p$. Assume that for all $\beta\in U(\tilde\beta_0)$ and $1\le k,l\le p$, $|m(X,\beta)|\le F(X),
|\dot{m}_k(X,\beta)|\le F(X),
|\ddot{m}_{kl}(X,\beta)|\le F(X),
|g_0(X)|\le F(X)$, where $U(\tilde\beta_0)$ is a neighborhood of $\tilde\beta_0$. In addition, assume that $\mathbb{E}(\varepsilon^2)<\infty$.
\end{assumption}

\begin{assumption}
\label{as:var}
Let $\phi(X,\beta)=\{m(X)-m(X_i,\beta)\}\dot{m}(X,\beta)$ and $\dot{\phi}(X,\beta)=\partial\phi(X,\beta)/\partial\beta$. We use $\phi_k(X,\beta)$ to denote the $k$th component of $\phi(X,\beta)$ and $\dot\phi_{kl}(X,\beta)$ to denote the $(k,l)$th entry of $\dot\phi(X,\beta)$, where $1\le k,l\le p$. Then, $|\phi_k(X,\widetilde{\beta}_0)|\leq F(X)$ and $|\dot{\phi}_{kl}(X,\beta)|\leq F(X)$ for all $\beta\in U(\tilde\beta_0)$ and $1\le k,l\le p$, where $U(\tilde\beta_0)$ is a neighborhood of $\tilde\beta_0$.
\end{assumption}

\begin{assumption}
\label{as:eigen}
Let $\Sigma(\beta) = -\mathbb{E}\{\dot{\phi}(X,\beta)\}$. Assume that $\Sigma(\beta)$ is nonsingular and that there exist constants $0<C_1<C_2<\infty$ such that
$$0< C_1 \leq \lambda_{\min}\{\Sigma(\beta)\}\leq \lambda_{\max}\{\Sigma(\beta)\} \leq C_2 < \infty ,$$
where $\lambda_{\min}\{\Sigma(\beta)\}$ and $\lambda_{\max}\{\Sigma(\beta)\}$ denote the smallest and largest eigenvalues of $\Sigma(\beta)$, respectively, for all $\beta\in U(\tilde\beta_0)$.

Let $\Sigma_l(\beta)$ denote the $l$th row of $\Sigma(\beta)$, and let $\dot{\Sigma}_l(\beta)=\partial\Sigma_l(\beta)/\partial\beta$. Assume that there exists a constant $C>0$ such that
$$ \max_{1\leq l,k \leq p}|\lambda_k\{\dot{\Sigma}_l(\beta)\}| \leq C,$$
for all $\beta\in U(\tilde\beta_0)$, where $\{\lambda_k(\dot{\Sigma}_l(\beta)):1 \leq k \leq p \}$ are the eigenvalues of the matrix $\dot{\Sigma}_l(\beta)$.    
\end{assumption}

\begin{assumption}
\label{as:lip}
Let $\dot{\phi}_{kl}(X,\beta)$ and $\ddot{m}_{kl}(X,\beta)$ denote the $(k,l)$th entries of $\dot{\phi}(X,\beta)$ and $\ddot{m}(X,\beta)$, respectively, for $1\le k,l\le p$. Assume that, for any $\beta_1,\beta_2\in U(\tilde{\beta}_0)$,
\begin{eqnarray*}
	|\dot{\phi}_{\beta_{kl}}(x,\beta_1)-\dot{\phi}_{\beta_{kl}}(x,\beta_2)|&\leq& \sqrt{p}||\beta_1-\beta_2||F(x),\\
	|\ddot{m}_{\beta_{kl}}(x,\beta_1)-\ddot{m}_{\beta_{kl}}(x,\beta_2)|&\leq& \sqrt{p}||\beta_1-\beta_2||F(x),
\end{eqnarray*}    
for all $1\le k,l\le p$.
\end{assumption}

\begin{assumption}
    \label{as:inter}
    The vector $\widetilde{\beta}_0$ lies in the interior of the compact subset $\Theta$ and is the unique minimizer of \eqref{eq:population_ols}.
\end{assumption}

\begin{assumption}
\label{as:eigen_F}
    Assume that there exists a constant \(C>0\), such that
\begin{eqnarray*}
	&&\lambda_{\max}[\mathbb{E}\{F^2(X)\dot{m}(X,\beta)\dot{m}(X,\beta)^{\top}\}] \leq C,\\
	&&\lambda_{\max}[\mathbb{E}\{F(X)\ddot{m}(X,\beta) \} ] \leq C,
\end{eqnarray*}
 for all $\beta\in U(\tilde\beta_0)$.
\end{assumption} 

\begin{assumption}
    \label{as:weight}
    The weight function $\varphi(t)$ is positive and satisfies $\varphi(t)=\varphi(-t)$, $\int_{\mathbb{R}} \varphi(t)dt=O(1)$ and $\int_{\mathbb{R}}t^4 \varphi(t)dt=O(1)$.
\end{assumption}

Assumption \ref{as:model} requires that the regression function $m(x,\beta)$ is twice differentiable and the corresponding derivatives have bounded fourth-moment. This condition is weaker than that in \cite{tan2025weighted}, which assumes the existence of third-order derivatives of $m(x, \beta)$. Assumption \ref{as:var} is standard in the literature for establishing the convergence of the least squares estimator; see, for example, \cite{fan2004nonconcave,tan2025weighted}. Assumptions \ref{as:eigen} and \ref{as:lip} are imposed to control the error terms in the estimation of $\tilde{\beta}_0$. Assumption \ref{as:eigen} guarantees local nonsingularity and smoothness of the Jacobian matrix, whereas Assumption \ref{as:lip} imposes local Lipschitz continuity on $\dot{\phi}_{kl}(X,\beta)$ and $\ddot{m}_{kl}(X,\beta)$. Assumption \ref{as:inter} is standard in the M-estimation literature; see, for example, \cite{van2000asymptotic}. Taken together, Assumptions \ref{as:model}--\ref{as:inter} serve as regularity conditions for establishing the $\ell_2$-consistency and asymptotic linear expansion of $\widehat{\beta}_n - \widetilde{\beta}_0$ in high-dimensional settings \citep{tan2025weighted}. Assumption \ref{as:eigen_F} requires that the largest eigenvalues of second-moment matrices associated with the score function are uniformly bounded. This condition controls the variability of the estimating components in high-dimensional settings and is useful for bounding the error terms; see, for example, \cite{fan2004nonconcave,tan2025weighted}. Assumption \ref{as:weight} imposes mild regularity conditions on the weight function. Specifically, it requires symmetry, integrability, and finite higher-order moments. These conditions ensure that the weighted integrals are well-defined and control higher-order terms in the asymptotic analysis.

Under Assumptions \ref{as:model}--\ref{as:weight}, we first establish an asymptotic expansion for $\widehat{U}_{n}(t)$ under the null hypothesis $H_0$. Specifically, we show that
\begin{equation}\label{eq:asymp_U}
	\begin{split}
	\widehat{U}_{n}(t)
	&= \frac{1}{\sqrt{n}}\sum_{i=1}^{n}\bigg(g_0(X_i)[\cos(t\varepsilon_i)+\sin(t\varepsilon_i) -  \mathbb{E}\{\cos(t \varepsilon_i)+\sin(t \varepsilon_i)\}]  \\    
	&+W(t)^{\top}\Sigma^{-1}\dot{m}(X_i,\beta_0)\varepsilon_it \bigg)+R_n(t)  \\
	&=: U_n(t)+R_n(t),
	\end{split}
\end{equation}
where 
$$W(t)=\mathbb{E}[g_0(X)\dot{m}(X,\beta_0)\{\sin(t\varepsilon)-\cos(t\varepsilon)\}]$$
and the remainder $R_n(t)$ satisfies $$\int_{\mathbb{R}}	|R_n(t)|^2 \varphi(t)dt =o_p(1).$$
Thus, $\widehat{U}_{n}(t)$ admits a uniformly asymptotically linear representation. This decomposition separates the leading stochastic component from the negligible remainder term and provides the basis for characterizing the asymptotic behavior of $\text{WICM}_n$ under $H_0$. The proof of \eqref{eq:asymp_U} is given in the appendix. We then have the following result.

\begin{theorem}\label{theorem:1}
Suppose that Assumptions \ref{as:model}--\ref{as:weight} hold. If $\{p^3 \log (n)\}/n \rightarrow 0$, then, under the null hypothesis, $\text{WICM}_n$ converges in distribution as
\begin{equation}\label{eq:asymp_WICM}
	\text{WICM}_n \overset{d}{\longrightarrow}\int_{\mathbb{R}}|U_\infty(t)|^2 \varphi(t)dt,
\end{equation}
where $U_\infty(t)$ is a mean-zero Gaussian process with covariance function $K(s,t)$, defined as the pointwise limit of $K_n(s,t)$. Here $K_n(s,t)$ denotes the covariance function of $U_{n}(t)$,
\begingroup
\allowdisplaybreaks
\begin{align*}
	&K_n (s,t)\\
	=& \mathbb{E} \big( g^2_0(X) [\cos(s\varepsilon)+\sin(s\varepsilon) -  \mathbb{E}\{\cos(s \varepsilon)+\sin(s \varepsilon)\} ]\\
    &\times[\cos(t\varepsilon)+\sin(t\varepsilon) -  \mathbb{E}\{\cos(t \varepsilon)+\sin(t \varepsilon)\} ] \big) \\
	& +W(s)^{\top}\Sigma^{-1} \mathbb{E}\big( g_0(X)[\cos(t\varepsilon)+\sin(t\varepsilon) -  \mathbb{E}\{\cos(t \varepsilon)+\sin(t \varepsilon)\} ] \varepsilon  \dot{m}(X,\beta_0) \big)s \\
	&  +W(t)^{\top}\Sigma^{-1} \mathbb{E}\big( g_0(X)[\cos(s\varepsilon)+\sin(s\varepsilon) -  \mathbb{E}\{\cos(s \varepsilon)+\sin(s \varepsilon)\} ] \varepsilon  \dot{m}(X,\beta_0) \big)t\\
	&  +W(s)^{\top}\Sigma^{-1} \mathbb{E}\{ \varepsilon^2   \dot{m}(X,\beta_0) \dot{m}(X,\beta_0)^{\top}\}\Sigma^{-1} W(t)st.
\end{align*}
\endgroup
We assume further that the limiting covariance function $K(s,t)$ is not identically zero.
\end{theorem}

Theorem~\ref{theorem:1} characterizes the asymptotic distribution of the proposed test statistic under the null hypothesis. The covariance function $K_n(s,t)$ consists of three components: the first reflects the variation induced by the transformed error terms, the second captures the covariance between the score function and the transformed errors, and the third arises from the asymptotic variation of the parameter estimator. This decomposition makes explicit how parameter estimation affects the limiting distribution, thereby providing the theoretical foundation for asymptotically valid specification tests in high-dimensional regression models. 

This result also improves upon existing rate conditions in the literature. For single-index null models, \cite{tan2019adaptive} established weak convergence of the residual empirical process under the condition $ p = o(n^{1/5}) $. For more general multiple-index models, \cite{tan2025weighted} strengthened this condition to $ p = o( ( n/ \log (n) )^{1/3}) $. More recently, \cite{tan2025weighted} further improved the rate for general regression models to $ p = o(( n/ \log(n)^6 )^{1/3}) $. The rate condition in Theorem \ref{theorem:1} is weaker than that in previous work, thereby broadening the scope of the asymptotic theory.

After establishing the limiting null distribution, we next study the behavior of $\text{WICM}_n$ under alternatives. We consider a sequence of alternative models of the form
\begin{equation}\label{eq:alter_model}
H_{1n}:Y=m(X,\widetilde{\beta}_0)+r_nS(X)+\varepsilon,
\end{equation}
where $\widetilde{\beta}_0$ is the population least square estimator defined in \eqref{eq:population_ols}, $S(\cdot)$ is a real-valued nonconstant function satisfying $\mathbb{E}\{S(X)\}=0$. Let $r_n=n^{-\alpha}$, where $\alpha \in [0,1/2]$. When $\alpha=0$, model \eqref{eq:alter_model} reduces to a global alternative $H_1$; when $\alpha>0$, it corresponds to a sequence of local alternatives approaching the null.

The formulation in \eqref{eq:alter_model} is very general and can accommodate a broad class of alternative regression functions. In particular, for any regression function $m(X)$ satisfying suitable moment conditions, and $\widetilde{\beta}_0$ defined as the population least-squares estimator in \eqref{eq:population_ols}, we can write
$$Y=m(X)+\varepsilon=m(X;\widetilde{\beta}_0)+\underbrace{\{m(X)-m(X;\widetilde{\beta}_0)\}}_{=:r_nS(X)}+\varepsilon.$$
Thus, \eqref{eq:alter_model} can be viewed as a general local perturbation of the best population least-squares estimator $m(X;\widetilde{\beta}_0)$. In addition, since we focus on the least-squares estimator throughout the paper, the perturbation function $S(X)$ must satisfy an orthogonality condition. Specifically, under the alternative model \eqref{eq:alter_model}, the first-order optimality condition for $\widetilde{\beta}_0$, defined in \eqref{eq:population_ols}, implies
\begin{align*}
0&=\mathbb{E}[\{Y-m(X,\widetilde\beta_0)\}\dot{m}(X,\widetilde\beta_0)]\\
&=\mathbb{E}[\{m(X,\beta_0)+r_nS(X)+\varepsilon-m(X,\widetilde\beta_0)\}\dot{m}(X,\widetilde\beta_0)]\\
&=\mathbb{E}[\{m(X,\beta_0)+r_nS(X)-m(X,\widetilde\beta_0)\}\dot{m}(X,\widetilde\beta_0)]=\mathbb{E}[S(X)\dot{m}(X,\widetilde\beta_0)].
\end{align*}
Thus, the perturbation $S(X)$ must satisfy an orthogonality condition with respect to the score function $\dot{m}(X,\widetilde\beta_0)$. The following theorem shows that the proposed test is consistent against fixed alternatives and can detect local alternatives converging to the null at the rate $1/\sqrt{n}$. 

\begin{theorem}\label{thm:alt}
Suppose that Assumptions \ref{as:model}--\ref{as:weight} hold and that \((p^3\log n)/n \to 0\). Then:
\begin{enumerate}
\item If $\alpha=0$, corresponding to the fixed alternative $H_{1}$, and
$$
\int_{\mathbb{R}}|K^{(1)}(t)|^2\varphi(t)\,dt> C,
$$
for a constant $C>0$, then
\begin{equation*}
    \frac{1}{n}\int_{\mathbb{R}}|\widehat{U}_n(t)|^2\varphi(t)dt\stackrel{p}{\to}\int_{\mathbb{R}}|K^{(1)}(t)|^2\varphi(t)dt,
\end{equation*}
where
$$K^{(1)}(t):=\mathbb{E}[g_0(X_i)\{\cos(te_i)-\cos(t\varepsilon_i)+\sin(te_i)-\sin(t\varepsilon_i)\}],$$
and $e_i=\varepsilon_i+S(X_i)$.

\item If $\alpha\in(0,1/2)$, corresponding to a sequence of local alternatives $H_{1n}$, and
$$
\int_{\mathbb{R}}|K^{(2)}(t)|^2t^2\varphi(t)\,dt>C,
$$
for a constant $C>0$, then
\begin{equation*}
    \frac{1}{r_n^2n}\int_{\mathbb{R}}|\widehat{U}_n(t)|^2\varphi(t)dt\stackrel{p}{\to}\int_{\mathbb{R}}|K^{(2)}(t)|^2t^2\varphi(t)dt,
\end{equation*}
where
$$K^{(2)}(t):=\mathbb{E}[g_0(X_i)S(X_i)\{\cos(t\varepsilon_i)-\sin(t\varepsilon_i)\}].$$

\item If $\alpha=1/2$ and there exists a constant $C>0$ such that
$$\lambda_{\max}[\mathbb{E}\{g_0(X)^2S(X)^2\dot{m}(X,\Tilde{\beta}_0)\dot{m}(X,\Tilde{\beta}_0)^{\top}\}] \leq C,$$ then, under the local alternatives $H_{1n}$,
\begin{eqnarray*}
	\int_{\mathbb{R}}|\widehat{U}_{n}(t)|^2 \varphi(t)dt \stackrel{d}{\to}\int_{\mathbb{R}}|U_{\infty}(t)+K^{(2)}(t)t|^2 \varphi(t)dt,
\end{eqnarray*}
where $U_{\infty}(t)$ is the zero-mean Gaussian process defined in Theorem \ref{theorem:1}.
\end{enumerate}
\end{theorem}

Theorem~\ref{thm:alt} characterizes the asymptotic behavior of the test statistic $\text{WICM}_n$ under three regimes of alternatives, thereby providing a unified description of the power of the proposed test. We first consider fixed alternatives, under which the statistic grows linearly in $n$, implying consistency against global departures from the null model. We next consider local alternatives approaching the null at rate $n^{-\alpha}$ for $\alpha\in(0,1/2)$. In this case, after normalization by $n r_n^2$, the test statistic converges in probability to a positive constant. Consequently, the proposed test is consistent against local alternatives separated from the null by more than the $n^{-1/2}$ scale. Finally, we address the boundary case $r_n = n^{-1/2}$. Under additional regularity conditions, the statistic converges in distribution to a limit that differs from its null distribution. This characterizes the asymptotic behavior of the test statistic at the critical local rate and shows that the proposed test retains nontrivial power against such local alternatives.

\subsection{Bootstrap approximation}

Theorem~\ref{theorem:1} shows that the asymptotic null distribution of the proposed test statistic depends on the unknown error distribution and therefore does not admit a closed-form critical value. A bootstrap procedure is proposed to approximate the null distribution of the test statistic. In this subsection, we establish the asymptotic validity of the proposed bootstrap approximation. We impose the following regularity condition.

\begin{assumption}
    \label{as:kernel}
    The kernel density function $l(\cdot)$ is a positive, symmetric, and twice continuously differentiable function such that $\int_{\mathbb{R}} tl(t)dt =0$ and $\int_{\mathbb{R}} t^4l(t)dt < \infty$. The smoothing parameter $v_n$ satisfies $v_n=o_p(1)$ and $\log n=o(nv_n^4)$.
\end{assumption}

Assumption~\ref{as:kernel} specifies regularity conditions on the kernel density function $l(\cdot)$ and the smoothing parameter $v_n$. These conditions ensure uniform convergence of the smoothed residual density estimator to the true error density and are standard in the analysis of smooth residual bootstrap procedures; see \citet{tan2025weighted}. The following theorem shows that the proposed smooth residual bootstrap consistently approximates the null distribution of the $\text{WICM}_n$ test.

\begin{theorem}\label{thm:bootstrap}
	Suppose that Assumption~\ref{as:kernel} and the conditions of Theorems~\ref{theorem:1} and~\ref{thm:alt} are satisfied.
    \begin{enumerate}
        \item Under the null hypothesis $H_0$,
    \begin{equation*}
    \mathrm{WICM}_n^{*}
    \overset{d}{\longrightarrow}
    \int_{\mathbb{R}}|U_\infty^*(t)|^{2}\varphi(t)dt,
     \end{equation*}
    where $U_\infty^*(t)$ has the same distribution as the Gaussian process in Theorem~\ref{theorem:1}, and hence coincides with the limiting null distribution.

\item Under the local alternative $H_{1n}$ with $r_n=o(1)$,
     \begin{equation*}
    \mathrm{WICM}_n^{*}
    \overset{d}{\longrightarrow}
    \int_{\mathbb{R}}|U^*_\infty(t)|^{2}\varphi(t)dt,
     \end{equation*}
    where $U^*_\infty(t)$ has the same distribution as the Gaussian process in Theorem~\ref{theorem:1}, and hence coincides with the limiting null distribution.

\item Under the global alternative $H_1$, $\mathrm{WICM}_n^{*}$ converges in distribution to a finite random variable whose limiting distribution may differ from the null limiting distribution of $\mathrm{WICM}_n$.
    \end{enumerate}
\end{theorem}

Theorem~\ref{thm:bootstrap} establishes the asymptotic validity of the smooth residual bootstrap for approximating the null distribution of the test statistic $\text{WICM}_n$. In particular, under the null hypothesis and local alternatives, the proposed bootstrap procedure consistently reproduces the asymptotic distribution of $\text{WICM}_n$ under the null hypothesis and therefore provides a valid basis for inference. Under fixed alternatives, although the limiting distribution of $\text{WICM}_n^*$ need not coincide with the null limiting distribution, it still converges to a well-defined finite distribution. Combined with the divergence of $\text{WICM}_n$ under global alternatives, this implies consistency of the bootstrap-based test. As a consequence of Theorem~\ref{thm:bootstrap}, Corollary~\ref{cor:bootstrap_test} shows that the proposed bootstrap procedure is asymptotically valid and remains effective under both local and fixed alternatives.

\begin{corollary}\label{cor:bootstrap_test}
Suppose that conditions of Theorems~\ref{thm:bootstrap} are satisfied. Let $\widehat{c}_{\alpha}$ denote the $(1-\alpha)$-quantile of the bootstrap distribution of $\mathrm{WICM}_n^{*}$. Then, the following results hold.
\begin{enumerate}
    \item[(i)] Under the null hypothesis \(H_0\),
    \begin{equation*}
    \lim_{n\to\infty} P\bigl(\mathrm{WICM}_n > \widehat{c}_{\alpha}\bigr)=\alpha.
    \end{equation*}

    \item[(ii)] Under the fixed alternative $H_1$ and the local alternative $H_{1n}$ with $r_n=n^{-\alpha}$ for $\alpha\in[0,1/2)$,
    \begin{equation*}
    \lim_{n\to\infty} P\bigl(\mathrm{WICM}_n > \widehat{c}_{\alpha}\bigr)=1.
    \end{equation*}

    \item[(iii)] Under the local alternative $H_{1n}$ with $r_n=n^{-1/2}$,
    \begin{equation*}
    \lim_{n\to\infty} P\bigl(\mathrm{WICM}_n > \widehat{c}_{\alpha}\bigr)>0.
    \end{equation*}
\end{enumerate}
\end{corollary}

\subsection{The choice of weight function $g(\cdot)$}

In this section, we study the choice of the weight function $g(\cdot)$ for enhancing the power of the proposed test against local alternatives. By Corollary~\ref{cor:bootstrap_test}, the proposed test has asymptotic power one for alternatives when $r_n=n^{-\alpha}$ for $\alpha\in[0,1/2)$. Hence, the remaining issue is to focus on the boundary case $r_n=n^{-1/2}$ and to choose $g(\cdot)$ to maximize the test's power. Recall that under the null hypothesis $H_0$, the limiting distribution of the test statistic is determined by
\begin{equation*}
\begin{split}
U_{n}(t)
&= \frac{1}{\sqrt{n}}\sum_{i=1}^{n}\big(g_0(X_i)[\cos(t\varepsilon_i)+\sin(t\varepsilon_i) -  \mathbb{E}\{\cos(t \varepsilon_i)+\sin(t \varepsilon_i)\} ] \\
&+ W(t)^{\top}\Sigma^{-1}\dot{m}(X_i,\beta_0)t\varepsilon_i\big)
\end{split}
\end{equation*}
where $W(t)=\mathbb{E}[g_0(X)\dot{m}(X,\beta_0)\{\sin(t\varepsilon)-\cos(t\varepsilon)\}]$. According to the proof of Theorem~\ref{thm:bootstrap}, under the local alternative $H_{1n}$ with $r_n=n^{-1/2}$, the test statistic is determined by
\begin{equation*}
\begin{split}
   \ U_{n}(t)+ \mathbb{E}[g_0(X)S(X)\{-\sin(t\varepsilon)+\cos(t\varepsilon) \}] t +R^1_n(t), 
\end{split}	
\end{equation*}
where the remainder $R^1_n(t)$ satisfies $\int_{\mathbb{R}}	|R^1_n(t)|^2 \varphi(t)dt =o_p(1)$. To obtain high power against the local alternative $H_{1n}$ with $r_n=n^{-1/2}$, we seek a weight function $g(x)$ that maximizes $\mathbb{E}\{g_0(X)S(X)\}$. An application of the Cauchy--Schwarz inequality shows that, under $H_{1n}$, the optimal choice of the weight function is
\begin{equation*}
	g(X)\propto m(X,\widetilde{\beta}_0)- m(X).
\end{equation*}
This choice, however, is not practically useful. Under the null hypothesis $H_0$, we have $m(X)=m(X,\widetilde{\beta}_0)$, so that $g(X)\equiv 0$. Consequently, $\widehat{U}_n(t)\equiv 0$, and the resulting test is degenerate. Therefore, in this case, the asymptotic theory of the proposed test is no longer meaningful.

To avoid this degeneracy, we replace $m(X,\widetilde{\beta}_0)$ by the projection of $m(X)$ onto the linear space spanned by $\dot m(X,\widetilde{\beta}_0)$ and define the weight function as
\begin{equation*}
	g(X)= \dot{m}(X,\widetilde{\beta}_0)^{\top}\Sigma^{-1}\mathbb{E}\{\dot{m}(X,\widetilde{\beta}_0)m(X)\}-m(X).
\end{equation*}
This construction avoids degeneracy under the null, while still preserving the first-order direction of departure under the alternative. A similar approach was considered by \cite{tan2025weighted}. Since $g(X)$ depends on both $ m(X) $ and the unknown parameter $\widetilde{\beta}_0$, we replace them by their estimates to obtain an estimator of $g(X)$. In this paper, we consider two ways of estimating $g(X)$: one for directional alternatives and the other for nonparametric alternatives.

For directional alternatives, we assume that the conditional mean function belongs to a given parametric class $\mathcal{K}=\{s(x,\theta):\theta \in \mathbb{R}^q \}$, which is distinct from the parametric family under the null. If $m(X) \in \mathcal{K} $, then there exists $\theta_0 \in \mathbb{R}^q$, such that  $m(X) = s(x,\theta_0)$.  We estimate $\theta_0$ by least squares,
\begin{equation*}
\widehat{\theta}_n =\operatorname{argmin}_{\theta\in\mathbb{R}^q} \frac{1}{n} \sum_{i=1}^{n}\{Y_i - s(X_i,{\theta})\}^2.
\end{equation*}
Based on this estimator, we define the weight function $g(X)$ as
\begin{equation}
\begin{split}
\label{eq:direct}
	g(X)=&\dot{m}(X,\widehat{\beta}_n)^{\top}\bigg\{\frac{1}{n}\sum_{i=1}^{n} \dot{m}(X_i,\widehat{\beta}_n)\dot{m}(X_i,\widehat{\beta}_n)^{\top}\bigg\}^{-1}\\
    &\times\bigg\{\frac{1}{n}\sum_{i=1}^{n}\dot{m}(X_i,\widehat{\beta}_n)s(X_i,\widehat{\theta}_n)\bigg\}-s(X,\widehat{\theta}_n)
\end{split}
\end{equation}

For nonparametric alternatives, the regression function $m(X)$ is left unspecified, so one natural approach is to estimate it nonparametrically. However, such methods suffer severely from the curse of dimensionality in high-dimensional settings. To overcome this difficulty, \cite{tan2025weighted} proposed an alternative approach based on Fourier expansion. Note that we use $\dot{m}_k(X,\beta)$ to denote the $k$th component of $\dot{m}(X,\beta)$ for $1\le k\le p$. Let $L^2(F_X)$ denote the Hilbert space of square-integrable functions equipped with inner product $\langle g,h\rangle=\int gh dF_X$, where $F_X$ is the distribution of $X$. By the definition of $\tilde{\beta}_0$, the function $m(X)-m(X,\widetilde{\beta}_0)$ is orthogonal to the linear space spaned by $\{\dot{m}_k(X,\widetilde{\beta}_0)\}_{k=1}^{p}$. We apply the Gram--Schmidt procedure to $\{\dot{m}_k\}_{k=1}^{p}$ to obtain an orthonormal basis $g_1,\cdots,g_p,g_{p+1},\cdots$ of $L^2(F_X)$, where $\text{span}(\{\dot{m}_k\}_{k=1}^{p})=\text{span}(\{g_k\}_{k=1}^{p})$. In practice, we truncate the expansion at a finite level $l$. Based on the corresponding Fourier expansion, we estimate $m(X)$ by
\begin{equation*}
\widehat{m}_l(X)=m(X,\widehat{\beta}_n)+\sum_{i=p+1}^{p+l}\widehat{a}_ig_i(X),
\end{equation*}
where 
$$\widehat{a}_i=\frac{1}{n}\sum_{j=1}^{n}\widehat{e}_jg_i(X_j)\quad\text{and}\quad\widehat{e}_j=Y_j-m(X_j,\widehat{\beta}_n).$$
Accordingly, we define the weight function by
\begin{equation}
\begin{split}
\label{eq:nonpara}
	g(X)=&\dot{m}(X,\widehat{\beta}_n)^{\top}\bigg\{\frac{1}{n}\sum_{i=1}^{n} \dot{m}(X_i,\widehat{\beta}_n)\dot{m}(X_i,\widehat{\beta}_n)^{\top}\bigg\}^{-1}\\
    &\times\bigg\{\frac{1}{n}\sum_{i=1}^{n}\dot{m}(X_i,\widehat{\beta}_n)\widehat{m}_l(X_i)\bigg\}-\widehat{m}_l(X).
\end{split}
\end{equation}
The choice of the additional basis functions $\{g_i\}_{i=1}^{\infty}$ and the truncation level $l$ is left to the researcher. Additional details on this construction for nonparametric alternatives can be found in \cite{tan2025weighted}.

\section{Numerical Results}
\label{sec:num}
\subsection{Simulation studies}
In this subsection, we present simulation studies to evaluate the finite-sample performance of the proposed test statistic under various combinations of the covariate dimension $p$ and the sample size $n$. We denote by $\text{WICM}^{(1)}_n$ the test statistic constructed using $g(X)$ for directional alternatives in \eqref{eq:direct}, and by $\text{WICM}^{(2)}_n$ the test statistic constructed using $g(X)$ for nonparametric alternatives \eqref{eq:nonpara}. We compare the proposed test statistics with several existing procedures. Specifically, we consider the tests of \cite{tan2025weighted}, both without and with the transformation, denoted by $\text{CM}_n$ and $\text{TCM}_n$, respectively. We also include the integrated conditional moment test of \cite{bierens1982consistent}, denoted by $\text{ICM}_n$, and the projected Cramér--von Mises test of \cite{escanciano2006consistent}, denoted by $\text{PCvM}_n$. The test statistics $\text{ICM}_n$ and $\text{PCvM}_n$ are not asymptotically pivotal, since their limiting null distributions depend on the unknown data-generating mechanism. We therefore use the wild bootstrap to determine the critical values, following \cite{escanciano2006consistent} and \cite{lavergne2012one}. For our proposed tests $\text{WICM}^{(1)}_n$ and $\text{WICM}^{(2)}_n$, we use the smooth residual bootstrap with smoothing parameter $v_n = 0.2$, as suggested by \cite{dette2007new}. The significance level is set at $0.05$. All simulation results are based on $1,000$ Monte Carlo replications, with bootstrap critical values computed from $B = 500$ bootstrap samples.

We consider several data-generating models in the simulations, including a single-index model, multiple-index models with low-dimensional structure, and one model with no low-dimensional structure. The data are generated from the following regression models:
\begin{eqnarray*}
	H_{1}: Y &=& \beta_0^{\top}X + a \cos (0.6 \pi \beta_0^{\top}X)+\varepsilon,\\
    H_{2}: Y &=& \beta_0^{\top}X  + a (\beta_1^{\top}X )^2 + \varepsilon, \\
	H_{3}: Y &=& \beta_0^{\top}X + a \{(\beta_1^{\top}X )^2 + \cos(0.5\pi \beta_1^{\top}X)+(\beta_0^{\top}X )(\beta_1^{\top}X ) \} +\varepsilon,\\
	H_{4}: Y &=& X_{1} + a \{|X_{2}| + X_{3} - X^2_{4} + X^3_{5} +X_{6}X_{7} \\
    &&+\cos(\pi X_{8} ) + \sin(0.5 \pi X_{9}X_{10})\}+\varepsilon,
\end{eqnarray*}		
where the coefficient vectors in the above models are $\beta_0=(1,\dots,1)^\top/\sqrt{p}$ and $\beta_1=(1,\dots,1,0,\ldots,0)^\top/\sqrt{p/2}$, so that only the first $p/2$ components of $\beta_1$ are nonzero. Note that $H_{4}$ does not admit a dimension-reduction structure. The regression error $\varepsilon$ is generated from a standard normal distribution. The covariate vector $X$ is independently from a mean-zero multivariate normal distribution with covariance matrix $\Sigma$. We consider two covariance structures, $\Sigma_1=I_p$ and $\Sigma_2=(2^{-|i-j|})_{p\times p}$. In this setup, $a=0$ represents the null hypothesis, while $a\neq 0$ represents the alternative hypothesis. The results under $\Sigma_2$ are presented in the appendix.

When implementing $\text{WICM}_n^{(1)}$ and $\text{WICM}_n^{(2)}$, the weight function $g(X)$ must be specified. For directional alternatives, $g(X)$ is chosen according to \eqref{eq:direct}. For nonparametric alternatives, following \cite{tan2025weighted}, we use sufficient dimension reduction to construct an orthonormal basis that approximates $m(X)$. Specifically, let $\mathcal{S}_{(Y\mid X)}$ denote the central subspace of $Y$ given $X$, which is the intersection of all subspaces $\operatorname{span}(B)$ such that $Y \perp \!\!\! \perp  X | B ^{\top}X$. Under mild conditions, $\mathcal{S}{(Y\mid X)}$ exists; see \cite{cook2009regression}. If $\mathcal{S}{(Y\mid X)}=\operatorname{span}(B)$, then $m(X)=E(Y|X)=E(Y|B^TX)$, where $B=(B_1,\ldots,B_s)$ is a $p\times s$ orthogonal matrix, and $s\le p$ is the structural dimension. To estimate $\mathcal{S}_{(Y\mid X)}$ in high-dimensional settings, we use cumulative slicing estimation (CSE; \cite{zhu2010dimension}) and the minimum ridge-type eigenvalue ratio estimator (MERE; \cite{zhu2017adaptive}) to determine $s$. These methods are easy to implement and remain valid when $p^2/n\to 0$. Let $\widehat{B}=(\widehat{B}_1,\ldots,\widehat{B}_{\widehat{s}})$ denote the resulting estimator of $B$. Since $\dot m(x,\widetilde\beta_0)=(x_1,\ldots,x_p)^\top$ for $H_1$ to $H_3$, the Gram--Schmidt procedure yields the orthonormal basis $\{ (\widehat{B}_i^{\top}x)^2,(\widehat{B}_i^{\top}x)^3,(\widehat{B}_i^{\top}x)^4, 1 \leq i \leq \widehat{s} \}$. The weight function $g(X)$ is then constructed according to \eqref{eq:nonpara}.

\begin{table}[ht!]
\caption{Empirical sizes and powers of $\text{WICM}_n^{(1)}$, $\text{WICM}_n^{(2)}$, $\text{CM}_n$, $\text{TCM}_n$, $\text{ICM}_n$, and $\text{PCvM}_n$ for model $H_{1}$.}
\label{table:h1_1}
	\centering
	{\small\tiny\hspace{5cm}
		\renewcommand{\arraystretch}{1}\tabcolsep 0.2cm
		\begin{tabular}{cccccccccccc}
			\hline
			&\multicolumn{1}{c}{a} &\multicolumn{1}{c}{n=100} &\multicolumn{1}{c}{n=100} &\multicolumn{1}{c}{n=100}&\multicolumn{1}{c}{n=100} &\multicolumn{1}{c}{n=100}  &\multicolumn{1}{c}{n=200} &\multicolumn{1}{c}{n=400}
			&\multicolumn{1}{c}{n=600}  \\
			&&\multicolumn{1}{c}{p=2} &\multicolumn{1}{c}{p=4}
			&\multicolumn{1}{c}{p=6}
			&\multicolumn{1}{c}{p=8} &\multicolumn{1}{c}{p=10} &\multicolumn{1}{c}{p=14} &\multicolumn{1}{c}{p=19}  &\multicolumn{1}{c}{p=22}\\
			\hline
			$\text{WICM}_n^{(1)}$
			&0.0       &0.058 &0.053 &0.064 &0.048 &0.064 &0.052 &0.047 &0.046\\
			&0.1       &0.089 &0.209 &0.106 &0.089 &0.113 &0.173 &0.269 &0.328\\
			&0.2       &0.209 &0.246 &0.264 &0.257 &0.274 &0.435 &0.714 &0.875\\
			&0.3       &0.443 &0.466 &0.501 &0.499 &0.545 &0.815 &0.971 &0.999\\
			&0.4       &0.667 &0.701 &0.729 &0.753 &0.796 &0.967 &1.000 &1.000\\
			&0.5       &0.854 &0.869 &0.892 &0.922 &0.921 &0.999 &1.000 &1.000\\	
			\hline
			$\text{WICM}_n^{(2)}$
			&0.0       &0.041 &0.053 &0.052 &0.056 &0.060 &0.057 &0.061 &0.054\\
			&0.1       &0.073 &0.067 &0.058 &0.075 &0.061 &0.096 &0.122 &0.185\\
			&0.2       &0.111 &0.105 &0.114 &0.096 &0.096 &0.165 &0.303 &0.488\\
			&0.3       &0.222 &0.204 &0.190 &0.159 &0.141 &0.324 &0.661 &0.838\\
			&0.4       &0.350 &0.329 &0.286 &0.262 &0.222 &0.525 &0.885 &0.970\\
			&0.5       &0.542 &0.465 &0.417 &0.372 &0.321 &0.661 &0.962 &0.998\\
			\hline
            $\text{CM}_n$
            & 0.0 & 0.056 & 0.047 & 0.048 & 0.045 & 0.069 & 0.053 & 0.068 & 0.057 \\
            & 0.1 & 0.065 & 0.058 & 0.064 & 0.071 & 0.073 & 0.071 & 0.100 & 0.096 \\
            & 0.2 & 0.096 & 0.071 & 0.073 & 0.079 & 0.081 & 0.111 & 0.204 & 0.293 \\
            & 0.3 & 0.133 & 0.125 & 0.114 & 0.116 & 0.119 & 0.201 & 0.430 & 0.623 \\
            & 0.4 & 0.212 & 0.192 & 0.187 & 0.147 & 0.139 & 0.347 & 0.647 & 0.896 \\
            & 0.5 & 0.333 & 0.301 & 0.255 & 0.258 & 0.222 & 0.473 & 0.859 & 0.982 \\
            \hline
            $\text{TCM}_n$
            & 0.0   & 0.045 & 0.064 & 0.051 & 0.049 & 0.055 & 0.053 & 0.058 & 0.043 \\
            & 0.1   & 0.063 & 0.063 & 0.045 & 0.041 & 0.051 & 0.050  & 0.065 & 0.045 \\
            & 0.2   & 0.070  & 0.062 & 0.054 & 0.056 & 0.053 & 0.049 & 0.055 & 0.066 \\
            & 0.3   & 0.065 & 0.058 & 0.054 & 0.043 & 0.051 & 0.049 & 0.062 & 0.078 \\
            & 0.4   & 0.087 & 0.067 & 0.068 & 0.053 & 0.049 & 0.060  & 0.058 & 0.088 \\
            & 0.5   & 0.094 & 0.077 & 0.062 & 0.070  & 0.054 & 0.072 & 0.088 & 0.113 \\
            \hline          
			$ \text{ICM}_n $
            & 0.0      & 0.052 & 0.058 & 0.026 & 0.009 & 0.001 & 0.000 & 0.000 & 0.000 \\
            & 0.1    & 0.058 & 0.055 & 0.050  & 0.007 & 0.001 & 0.000 & 0.000 & 0.000 \\
           & 0.2    & 0.115 & 0.088 & 0.047 & 0.010 & 0.005 & 0.000 & 0.000 & 0.000 \\
           & 0.3    & 0.188 & 0.153 & 0.078 & 0.023 & 0.008 & 0.000 & 0.000 & 0.000 \\
           & 0.4    & 0.265 & 0.229 & 0.134 & 0.037 & 0.011 & 0.000 & 0.000 & 0.000 \\
           & 0.5    & 0.441 & 0.331 & 0.214 & 0.065 & 0.014 & 0.000 & 0.000 & 0.000 \\
           \hline
			$\text{PCvM}_n $
			&0.0       &0.048 &0.075 &0.052 &0.066 &0.070 &0.060 &0.059 &0.045\\
			&0.1       &0.062 &0.061 &0.074 &0.077 &0.060 &0.072 &0.071 &0.073\\
			&0.2       &0.070 &0.073 &0.083 &0.080 &0.092 &0.087 &0.126 &0.145\\
			&0.3       &0.096 &0.096 &0.099 &0.110 &0.075 &0.112 &0.191 &0.265\\
			&0.4       &0.134 &0.122 &0.136 &0.112 &0.116 &0.173 &0.269 &0.401\\
			&0.5       &0.200 &0.154 &0.148 &0.146 &0.135 &0.216 &0.378 &0.526\\
			\hline
	\end{tabular}}
\end{table}

\begin{table}[ht!]
\caption{Empirical sizes and powers of $\text{WICM}_n^{(1)}$, $\text{WICM}_n^{(2)}$, $\text{CM}_n$, $\text{TCM}_n$, $\text{ICM}_n$, and $\text{PCvM}_n$ for model $H_{2}$.}
\label{table:h2_1}
	\centering
	{\small\tiny\hspace{5cm}
		\renewcommand{\arraystretch}{1}\tabcolsep 0.2cm
		\begin{tabular}{cccccccccccc}
			\hline
			&\multicolumn{1}{c}{a} &\multicolumn{1}{c}{n=100} &\multicolumn{1}{c}{n=100} &\multicolumn{1}{c}{n=100}&\multicolumn{1}{c}{n=100} &\multicolumn{1}{c}{n=100}  &\multicolumn{1}{c}{n=200} &\multicolumn{1}{c}{n=400}
			&\multicolumn{1}{c}{n=600}  \\
			&&\multicolumn{1}{c}{p=2} &\multicolumn{1}{c}{p=4}
			&\multicolumn{1}{c}{p=6}
			&\multicolumn{1}{c}{p=8} &\multicolumn{1}{c}{p=10} &\multicolumn{1}{c}{p=14} &\multicolumn{1}{c}{p=19}  &\multicolumn{1}{c}{p=22}\\
			\hline
			$\text{WICM}_n^{(1)}$
            &0.0  & 0.050 & 0.064 & 0.065 & 0.051 & 0.063 & 0.059 & 0.050 & 0.053 \\ 
            &0.1  & 0.352 & 0.240 & 0.199 & 0.193 & 0.196 & 0.250 & 0.405 & 0.571 \\ 
            &0.2 & 0.816 & 0.703 & 0.637 & 0.545 & 0.553 & 0.789 & 0.964 & 0.997 \\ 
            &0.3 & 0.974 & 0.930 & 0.913 & 0.868 & 0.849 & 0.983 & 0.999 & 1.000 \\ 
            &0.4 & 0.997 & 0.995 & 0.993 & 0.985 & 0.959 & 0.999 & 1.000 & 1.000 \\ 
            &0.5 & 1.000 & 0.999 & 0.999 & 0.995 & 0.996 & 1.000 & 1.000 & 1.000 \\
			\hline
			$\text{WICM}_n^{(2)}$
           & 0.0  & 0.041 & 0.041 & 0.039 & 0.044 & 0.051 & 0.050  & 0.048 & 0.051 \\
           & 0.1  & 0.094 & 0.075 & 0.079 & 0.094 & 0.090  & 0.132 & 0.245 & 0.346 \\
           & 0.2  & 0.249 & 0.232 & 0.189 & 0.221 & 0.213 & 0.413 & 0.723 & 0.926 \\
           & 0.3  & 0.453 & 0.413 & 0.404 & 0.400 & 0.388 & 0.714 & 0.939 & 0.994 \\
           & 0.4  & 0.631 & 0.568 & 0.584 & 0.566 & 0.533 & 0.856 & 0.985 & 1.000     \\
           & 0.5  & 0.731 & 0.677 & 0.657 & 0.656 & 0.659 & 0.919 & 0.996 & 1.000 \\
            \hline
            $\text{CM}_n$
             & 0.0     & 0.065 & 0.035 & 0.039 & 0.068 & 0.068 & 0.066 & 0.053 & 0.069 \\
             & 0.1  & 0.098 & 0.107 & 0.099 & 0.081 & 0.112 & 0.155 & 0.217 & 0.365 \\
             & 0.2   & 0.239 & 0.238 & 0.204 & 0.261 & 0.223 & 0.427 & 0.660  & 0.883 \\
             & 0.3  & 0.415 & 0.435 & 0.396 & 0.382 & 0.384 & 0.670  & 0.899 & 0.986 \\
             & 0.4   & 0.606 & 0.546 & 0.507 & 0.518 & 0.500 & 0.805 & 0.967 & 0.998 \\
             & 0.5  & 0.634 & 0.642 & 0.609 & 0.586 & 0.597 & 0.870  & 0.982 & 0.998 \\
            \hline
            $ \text{TCM}_n $
            &0.0 & 0.045 & 0.064 & 0.051 & 0.049 & 0.055 & 0.053 & 0.058 & 0.043 \\
            &0.1 & 0.099 & 0.090 & 0.073 & 0.086 & 0.083 & 0.102 & 0.108 & 0.095 \\
            &0.2 & 0.158 & 0.123 & 0.124 & 0.120 & 0.113 & 0.131 & 0.149 & 0.216 \\
            &0.3 & 0.196 & 0.170 & 0.171 & 0.142 & 0.148 & 0.160 & 0.240 & 0.331 \\
            &0.4 & 0.237 & 0.222 & 0.178 & 0.182 & 0.148 & 0.198 & 0.253 & 0.368 \\
            &0.5 & 0.266 & 0.244 & 0.218 & 0.189 & 0.182 & 0.236 & 0.248 & 0.403 \\
            \hline
			$ \text{ICM}_n $
			&0.0       &0.051 &0.055 &0.025 &0.030 &0.000 &0.000 &0.000 &0.000\\
			&0.1       &0.090 &0.084 &0.035 &0.030 &0.000 &0.000 &0.000 &0.000\\
			&0.2       &0.306 &0.193 &0.101 &0.010 &0.000 &0.000 &0.000 &0.000\\
			&0.3       &0.611 &0.429 &0.207 &0.021 &0.000 &0.000 &0.000 &0.000\\
			&0.4       &0.862 &0.706 &0.366 &0.048 &0.001 &0.000 &0.000 &0.000\\
			&0.5       &0.949 &0.845 &0.538 &0.090 &0.001 &0.000 &0.000 &0.000\\
			\hline
			$\text{PCvM}_n$
			&0.0       &0.054 &0.062 &0.047 &0.065 &0.060 &0.065 &0.053 &0.057\\
			&0.1       &0.162 &0.148 &0.163 &0.183 &0.186 &0.294 &0.498 &0.673\\
			&0.2       &0.419 &0.445 &0.463 &0.480 &0.470 &0.744 &0.969 &0.995\\
			&0.3       &0.722 &0.772 &0.737 &0.739 &0.732 &0.970 &1.000 &1.000\\
			&0.4       &0.926 &0.910 &0.929 &0.916 &0.898 &0.998 &1.000 &1.000\\
			&0.5       &0.979 &0.979 &0.983 &0.976 &0.972 &0.999 &1.000 &1.000\\
			\hline
	\end{tabular}}
\end{table}

\begin{table}[ht!]
\caption{Empirical sizes and powers of $\text{WICM}_n^{(1)}$, $\text{WICM}_n^{(2)}$, $\text{CM}_n$, $\text{TCM}_n$, $\text{ICM}_n$, and $\text{PCvM}_n$ for model $H_{3}$.}
\label{table:h3_1}
	\centering
	{\small\tiny\hspace{5cm}
		\renewcommand{\arraystretch}{1}\tabcolsep 0.2cm
		\begin{tabular}{cccccccccccc}
			\hline
			&\multicolumn{1}{c}{a} &\multicolumn{1}{c}{n=100} &\multicolumn{1}{c}{n=100} &\multicolumn{1}{c}{n=100}&\multicolumn{1}{c}{n=100} &\multicolumn{1}{c}{n=100}  &\multicolumn{1}{c}{n=200} &\multicolumn{1}{c}{n=400}
			&\multicolumn{1}{c}{n=600}  \\
			&&\multicolumn{1}{c}{p=2} &\multicolumn{1}{c}{p=4}
			&\multicolumn{1}{c}{p=6}
			&\multicolumn{1}{c}{p=8} &\multicolumn{1}{c}{p=10} &\multicolumn{1}{c}{p=14} &\multicolumn{1}{c}{p=19}  &\multicolumn{1}{c}{p=22}\\
			\hline
			$\text{WICM}_n^{(1)}$
          &0.0    & 0.044 & 0.050 & 0.058 & 0.058 & 0.073 & 0.056 & 0.058 & 0.064 \\ 
          &0.05   & 0.250 & 0.171 & 0.164 & 0.167 & 0.199 & 0.240 & 0.311 & 0.391 \\ 
          &0.10   & 0.583 & 0.459 & 0.391 & 0.354 & 0.383 & 0.506 & 0.723 & 0.881 \\ 
          &0.15   & 0.839 & 0.732 & 0.665 & 0.621 & 0.586 & 0.833 & 0.985 & 0.998 \\ 
          &0.20   & 0.945 & 0.906 & 0.857 & 0.815 & 0.798 & 0.955 & 1.000 & 1.000 \\ 
          &0.25   & 0.989 & 0.960 & 0.951 & 0.924 & 0.924 & 0.993 & 1.000 & 1.000 \\ 
          \hline
			$\text{WICM}_n^{(2)}$
          & 0.0    & 0.040  & 0.040   & 0.041  & 0.045  & 0.055  & 0.051  & 0.048  & 0.051  \\
          & 0.05   & 0.090  & 0.091  & 0.081  & 0.077  & 0.100    & 0.141  & 0.251  & 0.387  \\
          & 0.10   & 0.268  & 0.230   & 0.218  & 0.266  & 0.218  & 0.440   & 0.799  & 0.934  \\
          & 0.15   & 0.496  & 0.471  & 0.450   & 0.449  & 0.450   & 0.782  & 0.980   & 1.000       \\
          & 0.20   & 0.709  & 0.688  & 0.661  & 0.648  & 0.637  & 0.949  & 1.000   & 1.000      \\
          & 0.25   & 0.845  & 0.839  & 0.819  & 0.802  & 0.794  & 0.979  & 1.000   & 1.000       \\
             \hline
             $\text{CM}_n$
          & 0.0    & 0.057 & 0.050  & 0.056 & 0.072 & 0.059 & 0.058 & 0.056 & 0.058 \\
          & 0.05   & 0.101 & 0.104 & 0.098 & 0.084 & 0.095 & 0.153 & 0.251 & 0.373 \\
          & 0.10   & 0.275 & 0.233 & 0.240  & 0.261 & 0.258 & 0.439 & 0.763 & 0.907 \\
          & 0.15   & 0.484 & 0.479 & 0.445 & 0.435 & 0.462 & 0.755 & 0.970  & 1.000       \\
          & 0.20   & 0.689 & 0.666 & 0.650  & 0.630  & 0.629 & 0.911 & 0.997 & 1.000       \\
          & 0.25   & 0.808 & 0.762 & 0.773 & 0.804 & 0.724 & 0.965 & 0.999 & 1.000       \\
             \hline
            $\text{TCM}_n$
          & 0.0    & 0.072  & 0.055  & 0.056  & 0.052  & 0.046  & 0.055  & 0.073  & 0.063  \\
          & 0.05   & 0.113  & 0.099  & 0.079  & 0.085  & 0.070   & 0.086  & 0.113  & 0.133  \\
          & 0.10   & 0.152  & 0.153  & 0.145  & 0.117  & 0.105  & 0.147  & 0.152  & 0.257  \\
          & 0.15   & 0.194  & 0.197  & 0.182  & 0.173  & 0.154  & 0.196  & 0.267  & 0.377  \\
          & 0.20   & 0.261  & 0.245  & 0.185  & 0.187  & 0.173  & 0.233  & 0.306  & 0.497  \\
          & 0.25   & 0.285  & 0.256  & 0.253  & 0.219  & 0.207  & 0.247  & 0.389  & 0.528  \\
			\hline
            $ \text{ICM}_n$
			& 0.0      & 0.053  & 0.051  & 0.024  & 0.008  & 0.002  & 0.000  & 0.000  & 0.000       \\
            & 0.05   & 0.125  & 0.099  & 0.074  & 0.035  & 0.009  & 0.000   & 0.000   & 0.000       \\
            & 0.10   & 0.340   & 0.290   & 0.202  & 0.114  & 0.056  & 0.008  & 0.004  & 0.001  \\
            & 0.15   & 0.647  & 0.559  & 0.442  & 0.288  & 0.197  & 0.141  & 0.141  & 0.190   \\
            & 0.20   & 0.881  & 0.792  & 0.713  & 0.526  & 0.410   & 0.485  & 0.631  & 0.840   \\
            & 0.25   & 0.963  & 0.936  & 0.874  & 0.761  & 0.658  & 0.738  & 0.900    & 0.981  \\
			\hline
			$\text{PCvM}_n $
			& 0.0    & 0.060   & 0.049  & 0.058  & 0.057  & 0.069  & 0.061  & 0.052  & 0.050   \\
            & 0.05   & 0.175  & 0.154  & 0.180   & 0.189  & 0.191  & 0.296  & 0.505  & 0.686  \\
            & 0.10   & 0.455  & 0.479  & 0.473  & 0.457  & 0.478  & 0.772  & 0.975  & 0.999  \\
            & 0.15   & 0.781  & 0.786  & 0.779  & 0.782  & 0.797  & 0.982  & 1.000  & 1.000     \\
            & 0.20   & 0.951  & 0.954  & 0.951  & 0.939  & 0.941  & 0.999  & 1.000  & 1.000      \\
            & 0.25   & 0.992  & 0.991  & 0.990   & 0.983  & 0.988  & 1.000  & 1.000  & 1.000      \\
			\hline
	\end{tabular}}
\end{table}

\begin{table}[ht!]
\caption{Empirical sizes and powers of $\text{WICM}_n^{(1)}$, $\text{WICM}_n^{(2)}$, $\text{CM}_n$, $\text{TCM}_n$, $\text{ICM}_n$, and $\text{PCvM}_n$ for model $H_{4}$.}
\label{table:h4_1}
	\centering
	{\tiny\hspace{5cm}
		\renewcommand{\arraystretch}{1}\tabcolsep 0.6cm
		\begin{tabular}{cccccccccccc}
			\hline
			&\multicolumn{1}{c}{a} &\multicolumn{1}{c}{n=100}  &\multicolumn{1}{c}{n=200} &\multicolumn{1}{c}{n=400}
			&\multicolumn{1}{c}{n=600}  \\
			&&\multicolumn{1}{c}{p=10} &\multicolumn{1}{c}{p=14} &\multicolumn{1}{c}{p=19}  &\multicolumn{1}{c}{p=22}\\
			\hline
			$\text{WICM}_n^{(1)}$
            &0.0  & 0.080 & 0.061 & 0.057 & 0.043 \\ 
            &0.1  & 0.270 & 0.361 & 0.470 & 0.517 \\ 
            &0.2  & 0.440 & 0.552 & 0.627 & 0.675 \\ 
            &0.3  & 0.540 & 0.614 & 0.637 & 0.652 \\ 
            &0.4  & 0.620 & 0.642 & 0.655 & 0.680 \\ 
            &0.5  & 0.590 & 0.622 & 0.677 & 0.706 \\ 		
			\hline
			$\text{WICM}_n^{(2)}$
            & 0.0    & 0.050  & 0.051 & 0.058 & 0.049 \\
            & 0.1    & 0.057 & 0.065 & 0.067 & 0.054  \\
            & 0.2    & 0.075 & 0.081 & 0.112 & 0.118  \\
            & 0.3    & 0.130  & 0.193 & 0.265 & 0.330  \\
            & 0.4    & 0.227 & 0.289 & 0.472 & 0.610  \\
            & 0.5    & 0.289 & 0.432 & 0.638 & 0.799  \\
            \hline
            $ \text{CM}_n$
            & 0.0    & 0.070 & 0.057 & 0.056 & 0.058  \\
            & 0.1    & 0.062 & 0.062 & 0.071 & 0.057  \\
            & 0.2    & 0.097 & 0.106 & 0.101 & 0.108  \\
            & 0.3    & 0.157 & 0.178 & 0.252 & 0.284  \\
            & 0.4    & 0.212 & 0.296 & 0.423 & 0.580  \\
            & 0.5    & 0.293 & 0.415 & 0.621 & 0.774  \\
            \hline  
            $\text{TCM}_n$
            & 0.0    & 0.051  & 0.059 & 0.067 & 0.049  \\
            & 0.1    & 0.070  & 0.057 & 0.059 & 0.063  \\
            & 0.2    & 0.071  & 0.089 & 0.115 & 0.110   \\
            & 0.3    & 0.130  & 0.136 & 0.191 & 0.225  \\
            & 0.4    & 0.154  & 0.187 & 0.254 & 0.274  \\ 
            & 0.5    & 0.198  & 0.250  & 0.318 & 0.415  \\
            \hline     
			$ \text{ICM}_n$
            & 0.0      & 0.004 & 0.000 & 0.000 & 0.000 \\
            & 0.1    & 0.002 & 0.000 & 0.000 & 0.000  \\
            & 0.2    & 0.001 & 0.000 & 0.000 & 0.000  \\
            & 0.3    & 0.002 & 0.000 & 0.000 & 0.000  \\
            & 0.4    & 0.003 & 0.000 & 0.000 & 0.000  \\
            & 0.5    & 0.002 & 0.000 & 0.000 & 0.000  \\
			\hline
			$\text{PCvM}_n$
			  & 0.0    & 0.060 & 0.070 & 0.061 & 0.050    \\
            & 0.1    & 0.057 & 0.077 & 0.084 & 0.077  \\
            & 0.2    & 0.069 & 0.075 & 0.102 & 0.153  \\
            & 0.3    & 0.086 & 0.090  & 0.149 & 0.180   \\
            & 0.4    & 0.083 & 0.114 & 0.180  & 0.205  \\
            & 0.5    & 0.113 & 0.129 & 0.204 & 0.251  \\
			\hline			
	\end{tabular}}
\end{table}

The simulation results are reported in Tables~\ref{table:h1_1}--\ref{table:h4_1}. Overall, $\text{WICM}_n^{(1)}$, $\text{WICM}_n^{(2)}$, $\text{CM}_n$, $\text{TCM}_n$, and $\text{PCvM}_n$ maintain the empirical significance level across the models considered. However, the $\text{ICM}_n$ test performs satisfactorily only in low-dimensional settings and becomes severely conservative as the dimension increases. For instance, under $H_1$, the empirical size of $\text{ICM}_n$ is close to the nominal level when $p=2$ or $p=4$, but quickly decreases toward zero as $p$ increases; a similar pattern is observed under $H_2$, $H_3$, and $H_4$.

Under the single-index alternative $H_1$, $\text{WICM}_n^{(1)}$ exhibits the highest empirical power in almost all settings. The power of $\text{WICM}_n^{(2)}$ is also competitive, but it is consistently lower than that of $\text{WICM}_n^{(1)}$, suggesting that the weight function used in $\text{WICM}_n^{(1)}$ is more effective for detecting this type of departure from the null model. The $\text{CM}_n$ test attains only moderate power, while $\text{TCM}_n$ exhibits a substantial power loss after the transformation, in agreement with \cite{tan2025weighted}. Since $\text{WICM}_n^{(2)}$ and $\text{CM}_n$ target the same nonparametric alternative, the power advantage of $\text{WICM}_n^{(2)}$ appears to arise from the use of an exponential weight function in place of the indicator function used in $\text{CM}_n$. Although $\text{PCvM}_n$ maintains empirical size control, its power is considerably lower than that of the proposed procedures. Overall, the proposed tests achieve the best power performance among the methods considered.

The results under the multiple-index alternative $H_2$ show a similar pattern. The proposed test $\text{WICM}_n^{(1)}$ is the most powerful procedure across nearly all combinations of $(n,p)$ and signal strength. The $\text{PCvM}_n$ test performs well under this model and typically has higher power than $\text{WICM}_n^{(2)}$, $\text{CM}_n$, and $\text{TCM}_n$, although it generally remains less powerful than $\text{WICM}_n^{(1)}$. The performance of $\text{CM}_n$ is broadly comparable to that of $\text{WICM}_n^{(2)}$ in several settings, whereas $\text{TCM}_n$ is again less powerful. The $\text{ICM}_n$ test has nontrivial power only when the dimension is small, but its empirical size distortion in higher dimensions. For model $H_3$, $\text{PCvM}_n$ is highly competitive and tends to achieve slightly higher power in several settings. Nevertheless, the proposed statistics remain comparable to $\text{PCvM}_n$ and continue to perform well across all dimensions and sample sizes. In particular, $\text{WICM}_n^{(1)}$ has strong power for weak and moderate signals, whereas $\text{WICM}_n^{(2)}$ is broadly comparable to, and often slightly more powerful than, $\text{CM}_n$. 

For model $H_4$, where the regression function does not possess a simple dimension-reduction structure, the proposed weighted tests continue to perform well. The statistic $\text{WICM}_n^{(1)}$ is generally the most powerful method, especially for weak and moderate alternatives. The power of $\text{WICM}_n^{(2)}$ increases markedly along the displayed sequence of $(n,p)$ values and can become comparable to, or even exceed, that of $\text{WICM}_n^{(1)}$ for stronger alternatives. The $\text{CM}_n$ test also gains power as the signal increases, but it is generally less powerful than the proposed tests. By contrast, $\text{TCM}_n$ and $\text{PCvM}_n$ have substantially lower power under $H_4$, and $\text{ICM}_n$ almost never rejects in the high-dimensional settings.

Taken together, the simulation results demonstrate that the proposed weighted tests maintain satisfactory size control and provide substantial power gains across a wide range of alternatives. Relative to the $\text{CM}_n$ and $\text{TCM}_n$ procedures of \cite{tan2025weighted}, the proposed tests are broadly competitive and, in several cases, deliver noticeably higher empirical power.

\subsection{Real data example}

In this subsection, we apply the proposed test to the Geographical Origin of Music data set, which was first analyzed by \cite{zhou2014predicting}. The data set contains $1,059$ observations, where the response is latitude $Y$, and the predictor vector $X=(X_1,\cdots,X_{68})^\top$ consists of $68$ audio features of the track. For simplicity, all variables are standardized separately. We consider fitting the linear regression model
\begin{equation*}
    Y=\beta^\top X+\varepsilon.
\end{equation*}
We then assess whether this linear model is adequate for the data. Treating the alternative as nonparametric, we use the same method to construct the weight function $g(x)$ as in the simulation studies. The value of $\text{WICM}_n$ is $3.3037$, with a $p$-value approximately equal to $0$ by the proposed bootstrap procedure. This provides strong evidence against the adequacy of the linear regression model for predicting the response.

To further illustrate this point, Figure \ref{Figure}(a) displays the scatter plot of $Y$ against the fitted values $\widehat{\beta}^\top X$, while Figure \ref{Figure}(b) shows the scatter plot of the residuals $\widehat{e}=Y-\widehat{\beta}^\top X$ against the fitted values $\widehat{\beta}^\top X$. Both plots suggest that the relationship between $Y$ and $X$ is not well described by a linear model. This indicates that a more flexible model is needed for this data set.

\begin{figure}[ht]
\centering
\includegraphics[width=\linewidth]{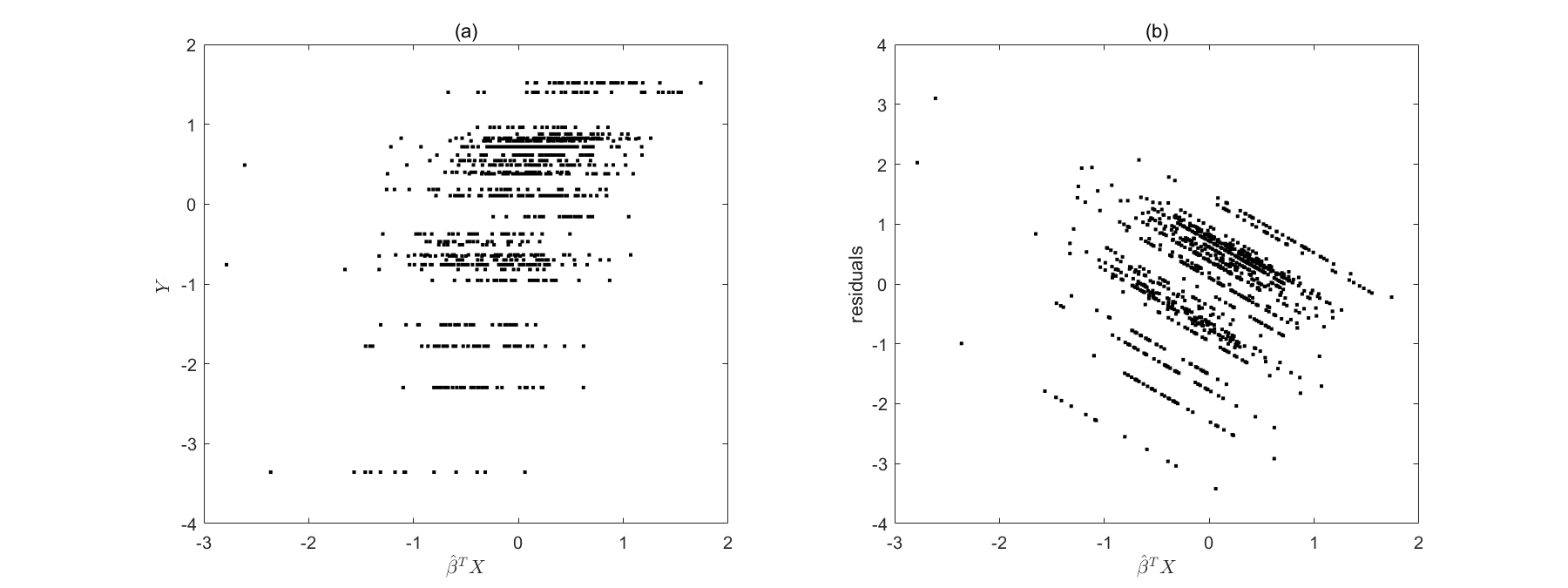}
\caption{(a) The scatter plot of $Y$ against the fitted values $\widehat{\beta}^{\top}X$, and (b) the scatter plot of the residuals against the fitted values $\widehat{\beta}^{\top}X$ for the Geographical Origin of Music data set.}
\label{Figure}
\end{figure}

\section{Discussion}
\label{sec:con}
Although widely used, the ICM test exhibits fundamentally different asymptotic behavior in high-dimensional settings, and the associated wild bootstrap is no longer valid. To address this issue, we propose a test based on weighted residual processes, together with a smooth residual bootstrap to approximate its null distribution. We establish the asymptotic properties of the proposed test under both the null and alternative hypotheses, showing that it admits a nondegenerate limiting distribution under the null, is consistent against fixed alternatives, and can detect local alternatives converging to the null at the rate $1/\sqrt{n}$. Simulation results show that the proposed test controls the nominal level well and outperforms the existing test in terms of power. 

There are several important directions for future research. One is to extend the proposed methodology to high-dimensional settings where $p>n$. Another important direction is to improve the power of the proposed test to detect more general forms of model misspecification, including the variance structure. It would also be useful to investigate more adaptive, data-driven choices of the weight function to improve power against particular alternatives.

\appendix
\allowdisplaybreaks
\section{Proofs of Theoretical Results for the Test Statistic}

The appendix contains the proofs of the theoretical results stated in the main text.

\label{sec:main_thm}
\begin{proof}[Proof of Theorem \ref{theorem:1}]
Recall that the test statistic is 
\begin{equation*}
\begin{split}
	\widehat{U}_n(t) =\frac{1}{\sqrt{n}}\sum_{i=1}^{n}\{g(X_i)-\bar{g}\}\{\cos(t\widehat{e}_i)+\sin(t\widehat{e}_i)\},
\end{split}
\end{equation*}
for $t\in\mathbb{R}$, where $\bar{g}=n^{-1}\sum_{i=1}^{n}g(X_i)$ and the estimated residual 
$$\widehat{e}_i=Y_i-m(X_i,\widehat{\beta}_n)=m(X_i)-m(X_i,\widehat{\beta}_n)+\varepsilon_i.$$
Under the null hypothesis $H_0$, we have $m(X_i)-m(X_i,\beta_0)$ and thus 
$$\widehat{e}_i=Y_i-m(X_i,\widehat{\beta}_n)=m(X_i,\beta_0)-m(X_i,\widehat{\beta}_n)+\varepsilon_i.$$
Next, we decompose the proof into the following steps.

\textbf{Step 1}. First, we show that
$$\widehat{U}_n(t)=U_n(t)+R_n(t) \quad  \mbox{and}  \quad \int_{\mathbb{R}}^{}|R_n(t)|^2\varphi(t)dt=o_p(1),$$
where $R_n(t)$ is a remainder and $U_n(t)$ is defined as
\begin{equation*}
\begin{split}
	U_n(t)=&\frac{1}{\sqrt{n}}\sum_{i=1}^{n}\big(g_0(X_i)[\cos(t\varepsilon_i)+\sin(t\varepsilon_i) -  E\{\cos(t \varepsilon_i)+\sin(t \varepsilon_i)\} ] \\
    &+ W(t)^{\top}\Sigma^{-1}\dot{m}(X_i,\beta_0)\varepsilon_i t\big).
\end{split}
\end{equation*}
For clarity about the factor $t$, set
$$
 W(t)=\mathbb E[g_0(X)\{\sin(t\varepsilon)-\cos(t\varepsilon)\}\dot m(X,\beta_0)]\text{ and } B(t)=tW(t).
$$
Consider the following decomposition of $\widehat{U}_n(t)$ as

\begin{align*}
	\widehat{U}_n(t)=&\frac{1}{\sqrt{n}}\sum_{i=1}^{n}g_0(X_i)\{\cos(t\widehat{e}_i)+\sin(t\widehat{e}_i)\}\\
    &+\frac{1}{\sqrt{n}}\sum_{i=1}^{n}\{E[g(X)]-\bar{g}\}\{\cos(t\widehat{e}_i)+\sin(t\widehat{e}_i)\}\\
    =&\underbrace{\frac{1}{\sqrt{n}}\sum_{i=1}^{n}g_0(X_i)(\cos(t\varepsilon_i)+\sin(t\varepsilon_i))}_{\widehat{U}_{n1}(t)}\\
    &+\underbrace{\frac{1}{\sqrt{n}}\sum_{i=1}^{n}g_0(X_i)\{\cos(t\widehat{e}_i)-\cos(t\varepsilon_i)+\sin(t\widehat{e}_i)-\sin(t\varepsilon_i)\}}_{\widehat{U}_{n2}(t)}\\
    &-\underbrace{\frac{1}{\sqrt{n}}\sum_{i=1}^{n} g_0(X_i) \frac{1}{n} \sum_{i=1}^{n}\{\cos(t \varepsilon_i)+\sin(t \varepsilon_i)\}}_{\widehat{U}_{n3}(t)}\\
    &-\underbrace{\frac{1}{\sqrt{n}}\bigg\{\sum_{i=1}^{n} \frac{1}{n}g_0(X_i)\bigg\}\sum_{i=1}^{n} \{\cos(t\widehat{e}_i)-\cos(t \varepsilon_i)\}}_{\widehat{U}_{n4}(t)}\\
    &-\underbrace{\frac{1}{\sqrt{n}}\bigg\{\sum_{i=1}^{n} \frac{1}{n}g_0(X_i)\bigg\} \sum_{i=1}^{n}\{
	\sin(t\widehat{e}_i)- \sin(t \varepsilon_i)\}}_{\widehat{U}_{n5}(t)}
\end{align*}
where we use the definition: $g_0(X_i)=g(X_i)-E[g(X)]$. Then, we deal with $\widehat{U}_{n2}(t),\widehat{U}_{n3}(t),\widehat{U}_{n4}(t),\widehat{U}_{n5}(t)$ separately.

\textit{(1) First}, we deal with the term $\widehat{U}_{n2}(t)$. By Taylor's expansion, we have
\begin{align*}
	\widehat{U}_{n2}(t) =& \frac{1}{\sqrt{n}}\sum_{i=1}^{n}g_0(X_i)\{\cos(t\widehat{e}_i)-\cos(t\varepsilon_i)+\sin(t\widehat{e}_i)-\sin(t\varepsilon_i)\} \\
	=& -\frac{1}{\sqrt{n}}\sum_{i=1}^{n}g_0(X_i)\{\sin(t\varepsilon_i)(\widehat{e}_i-\varepsilon_i)t+\frac{1}{2}\cos(t\widetilde\varepsilon_i)(\widehat{e}_i-\varepsilon_i)^2t^2\}\\
	&+\frac{1}{\sqrt{n}}\sum_{i=1}^{n}g_0(X_i)\{\cos(t\varepsilon_i)(\widehat{e}_i-\varepsilon_i)t-\frac{1}{2}\sin(t\widetilde\varepsilon_i)(\widehat{e}_i-\varepsilon_i)^2t^2\}
	\\
	=&\underbrace{-\frac{1}{\sqrt{n}}\sum_{i=1}^{n}g_0(X_i)\sin(t\varepsilon_i)(\widehat{e}_i-\varepsilon_i)t}_{\widehat{U}_{n21}(t)}\\
    &\underbrace{+\frac{1}{\sqrt{n}}\sum_{i=1}^{n}g_0(X_i)\frac{1}{2}\cos(t\widetilde\varepsilon_i)(\widehat{e}_i-\varepsilon_i)^2t^2}_{\widehat{U}_{n22}(t)}\\
    &\underbrace{+\frac{1}{\sqrt{n}}\sum_{i=1}^{n}g_0(X_i)\cos(t\varepsilon_i)(\widehat{e}_i-\varepsilon_i)t}_{\widehat{U}_{n23}(t)}\\
    &\underbrace{-\frac{1}{\sqrt{n}}\sum_{i=1}^{n}g_0(X_i)\frac{1}{2}\sin(t\widetilde\varepsilon_i)(\widehat{e}_i-\varepsilon_i)^2t^2}_{\widehat{U}_{n24}(t)},
\end{align*}
where $\Tilde{\varepsilon}_i$ is an intermediate value lying on the line segment between $\widehat{e}_i$ and $\varepsilon_i$. We then consider the term $\widehat{U}_{n21}(t),\widehat{U}_{n22}(t),\widehat{U}_{n23}(t),\widehat{U}_{n24}(t)$ separately.

For the term $\widehat{U}_{n21}(t)$,  by Taylor's expansion, we have
\begin{align*}
	\widehat{U}_{n21}(t)=& -\frac{1}{\sqrt{n}}\sum_{i=1}^{n}g_0(X_i)\sin(t\varepsilon_i)(\widehat{e}_i-\varepsilon_i)t\\
	=& \frac{1}{\sqrt{n}}\sum_{i=1}^{n}g_0(X_i)\sin(t\varepsilon_i)\{m(X_i,\widehat{\beta}_n)-m(X_i,\beta_0)\}t\\
	=& \underbrace{\frac{1}{\sqrt{n}}\sum_{i=1}^{n}g_0(X_i)\sin(t\varepsilon_i)(\widehat{\beta}_n-\beta_0)^{\top}\dot{m}(X_i,\beta_0)}_{\widehat{U}_{n211}(t)}\\
    &+\frac{1}{2}\underbrace{\frac{1}{\sqrt{n}}\sum_{i=1}^{n}g_0(X_i)\sin(t\varepsilon_i)(\widehat{\beta}_n-\beta_0)^T\Ddot{m}(X_i,\Tilde{\beta})(\widehat{\beta}_n-\beta_0)}_{\widehat{U}_{n212}(t)},
\end{align*}
where $\Tilde{\beta}$ lies between $\widehat{\beta}_n$ and $\beta_0$. For the first term $\widehat{U}_{n211}(t)$, we define 
$$M_1(t)= \mathbb{E}[g_0(X)\sin(t\varepsilon)\dot{m}(X,\beta_0)]t.$$
Furthermore,
\begin{align*}
	\widehat{U}_{n211}(t)=& \underbrace{\sqrt{n} (\widehat{\beta}_n-\beta_0)^{\top} \bigg(
	\frac{1}{n}\sum_{i=1}^{n}g_0(X_i)\sin(t\varepsilon_i)\dot{m}(X_i,\beta_0)t-M_1(t)\bigg)}_{\widehat{U}_{n2111}(t)}\\
    &+\underbrace{\sqrt{n}  (\widehat{\beta}_n-\beta_0)^{\top}M_1(t)}_{\widehat{U}_{n2112}(t)}. 
\end{align*}
Note that
\begin{align*}
	&\int_{\mathbb{R}}	|\widehat{U}_{n2111}(t)|^2 \varphi(t)dt\\
	 \leq& n \| \widehat{\beta}_n-\beta_0 \|_2^2\times \int_{\mathbb{R}} \bigg|\bigg|\frac{1}{n}\sum_{i=1}^{n}g_0(X_i)\sin(t\varepsilon_i)\dot{m}(X_i,\beta_0)t-M_1(t)\bigg|\bigg|_2^2 \varphi(t)dt
\end{align*}
where we use the Cauchy–Schwarz inequality. Furthermore,
\begin{align*}
	&\mathbb{E}  \bigg[\int_{\mathbb{R}} \bigg|\bigg|	\frac{1}{n}\sum_{i=1}^{n}g_0(X_i)\sin(t\varepsilon_i)\dot{m}(X_i,\beta_0)t-M_1(t) \bigg|\bigg|^2 \varphi(t)dt 	\bigg]
	\\
	=&\int_{\mathbb{R}}   \mathbb{E}
	\bigg[\bigg|\bigg|\frac{1}{n}\sum_{i=1}^{n}g_0(X_i)\sin(t\varepsilon_i)\dot{m}(X_i,\beta_0)t-M_1(t)\bigg|\bigg|^2\bigg]     \varphi(t) dt\\
	=& \sum_{k=1}^{p} \int_{\mathbb{R}} \mathbb{E} \bigg|\frac{1}{n}\sum_{i=1}^{n}g_0(X_i)\sin(t\varepsilon_i)\dot{m}_k(X_i,\beta_0)t-M_{1k}(t)\bigg|^2     \varphi(t) dt\\
= & \frac{1}{n}\sum_{k=1}^{p}\mathbb{V}[g_0(X)\sin(t\varepsilon)\dot{m}_k(X,\beta_0)]\times \int_{\mathbb{R}} t^2 \varphi(t) dt \\
     \leq& \frac{p}{n}\max_{1\leq k\leq p}\mathbb{E}[g_0(X)^2\dot{m}_k(X,\beta_0)^2]\times\int_{\mathbb{R}} t^2 \varphi(t) dt\leq C\frac{p}{n},
\end{align*}
where $M_{1k}(t)$ is the $k$th component of $M_{1}(t)$ and we use the definition of $M_{1}(t)$ in the third equality, we use Assumption \ref{as:model} and \ref{as:weight} in the third inequality. Thus, by Lemma \ref{lem:beta}, we have
\begin{equation*}
    \int_{\mathbb{R}}	|\widehat{U}_{n2111}(t)|^2 \varphi(t)dt=O_p\bigg(n\times\frac{p}{n}\times\frac{p}{n}\bigg)=o_p(1).
\end{equation*}
Furthermore, by Lemma \ref{lem:beta}, we have
\begin{align*}
    \widehat{U}_{n2112}(t)=\frac{1}{\sqrt{n}}
\sum_{i=1}^n M_1(t)^{\top}\Sigma^{-1}\dot{m}(X_i,\beta_0)\varepsilon_i+R_{n2111}(t),
\end{align*}
where
\begin{align*}
    &\mathbb{E}\bigg[\int_{\mathbb{R}}|R_{n2111}(t)|^2\varphi(t)dt\bigg]\leq O_p\!\left(\sqrt{\frac{p^3}{n}}\right)\times\int_{\mathbb{R}}\|M_{1}(t)\|_2^2\varphi(t)dt\\
    \leq&O_p\!\left(\sqrt{\frac{p^3}{n}}\right)\times\sup_{t\in\mathbb{R}}\|\mathbb{E}[g_0(X)\sin(t\varepsilon)\dot{m}(X,\beta_0)]\|_2^2\times\int_{\mathbb{R}}t^2\varphi(t)dt,
\end{align*}
where we use the Cauchy–Schwarz inequality in the first inequality. Let
\begin{equation*}
W_1(t)
:=
\mathbb{E}\!\left[g_0(X)\sin(t\varepsilon)\dot m(X,\beta_0)\right].
\end{equation*}
To show that \(\|W_1(t)\|_2^2\) is uniformly bounded in \(t\), it suffices to bound its one-dimensional projections. For any \(\alpha\in\mathbb{R}^p\),
\begin{equation*}
\alpha^\top W_1(t)
=
\mathbb{E}\!\left[g_0(X)\sin(t\varepsilon)\,\alpha^\top \dot m(X,\beta_0)\right].
\end{equation*}
By the Cauchy--Schwarz inequality,
\begin{equation*}
|\alpha^\top W_1(t)|^2
\le
\mathbb{E}\left[\sin^2(t\varepsilon)\right]
\mathbb{E}\left[\{g_0(X)\alpha^\top \dot m(X,\beta_0)\}^2\right].
\end{equation*}
Since \(\sin^2(t\varepsilon)\le 1\), we have
\begin{align*}
&|\alpha^\top W_1(t)|^2\\
\le&\mathbb{E}\!\left[\{g_0(X)\alpha^\top \dot m(X,\beta_0)\}^2\right]=\alpha^\top\mathbb{E}\!\left[g_0(X)^2\dot m(X,\beta_0)\dot m(X,\beta_0)^\top\right]\alpha.    
\end{align*}
By Assumption \ref{as:eigen_F}, we have $\lambda_{\max}\!\left(\mathbb{E}\!\left[g_0(X)^2\dot m(X,\beta_0)\dot m(X,\beta_0)^\top\right]\right)\le C$, then it follows that
\begin{equation*}
|\alpha^\top W_1(t)|^2 \le C\|\alpha\|_2^2
\qquad\text{for all } \alpha\in\mathbb{R}^p,
\end{equation*}
uniformly in \(t\). Therefore,
\begin{equation*}
\|W_1(t)\|_2
=
\sup_{\|\alpha\|_2=1} |\alpha^\top W_1(t)|
\le C^{1/2},
\end{equation*}
uniformly in \(t\), and hence, we have
\begin{eqnarray*}
	\int_{\mathbb{R}}|R_{n2111}(t)|^2\varphi(t)dt=o_p(1).
\end{eqnarray*}

For the term $\widehat{U}_{n212}(t)$, we define $W_2(t)=\mathbb{E}[g_0(X)\sin(t\varepsilon)\Ddot{m}(X,\beta_0)]$. Then, we have
\begin{align*}
	&\int_{\mathbb{R}}|\widehat{U}_{n212}(t)|^2 \varphi(t)dt\\
    =&  \int_{\mathbb{R}}|\frac{1}{\sqrt{n}}\sum_{i=1}^{n}g_0(X_i)\sin(t\varepsilon_i)(\widehat{\beta}_n-\beta_0)^{\top}\Ddot{m}(X_i,\Tilde\beta)(\widehat{\beta}_n-\beta_0)t|^2 \varphi(t)dt
	\\
	\leq&  \underbrace{3n
	\int_{\mathbb{R}}\bigg|(\widehat{\beta}_n-\beta_0)^{\top}\bigg[ \frac{1}{n}\sum_{i=1}^{n}g_0(X_i)\sin(t\varepsilon_i)\bigg\{\Ddot{m}(X_i,\beta_0)-\Ddot{m}(X_i,\Tilde\beta)\bigg\}t\bigg](\widehat{\beta}_n-\beta_0)\bigg|^2 \varphi(t) dt}_{\widehat{U}_{n2121}(t)}\\
    &+  \underbrace{3n
	\int_{\mathbb{R}}\bigg|(\widehat{\beta}_n-\beta_0)^{\top}\bigg[ \frac{1}{n}\sum_{i=1}^{n}g_0(X_i)\sin(t\varepsilon_i)\Ddot{m}(X_i,\beta_0)-W_2(t)\bigg](\widehat{\beta}_n-\beta_0)\bigg|^2t^2 \varphi(t) dt}_{\widehat{U}_{n2122}(t)}\\
	&+  \underbrace{3n \int_{\mathbb{R}} |(\widehat{\beta}_n-\beta_0)^{\top}W_2(t)(\widehat{\beta}_n-\beta_0)t|^2\varphi(t) dt}_{\widehat{U}_{n2123}(t)},
\end{align*}
where we use the inequality $(a+b+c)^2\leq 3a^2+3b^2+3c^2$. Then, we deal with each term separately. For $\widehat{U}_{n2121}(t)$, we have
\begin{align*}
    &\widehat{U}_{n2121}(t)\\
    \leq& 3n\|\widehat{\beta}_n-\beta_0\|_2^4\cdot\int_{\mathbb{R}}\bigg\|\frac{1}{n}\sum_{i=1}^{n}g_0(X_i)\sin(t\varepsilon_i)\{\Ddot{m}(X_i,\beta_0)-\Ddot{m}(X_i,\Tilde\beta)\}\bigg\|^2_2t^2 \varphi(t) dt\\
    \leq& 3n\|\widehat{\beta}_n-\beta_0\|_2^4\cdot\int_{\mathbb{R}}p^2\bigg\|\frac{1}{n}\sum_{i=1}^{n}g_0(X_i)\sin(t\varepsilon_i)\{\Ddot{m}(X_i,\beta_0)-\Ddot{m}(X_i,\Tilde\beta)\}\bigg\|^2_{\infty}t^2 \varphi(t) dt\\  
    \leq& 3n\|\widehat{\beta}_n-\beta_0\|_2^4\cdot\|\tilde{\beta}-\beta_0\|_2^2\cdot\int_{\mathbb{R}}p^3\bigg\|\frac{1}{n}\sum_{i=1}^{n}F(X_i)^2\bigg\|^2_{\infty}t^2 \varphi(t) dt\\
    \leq& O_p\bigg(n\times\frac{p^2}{n^2}\times\frac{p}{n}\times p^3\bigg)=o_p(1),
\end{align*}
where we use $\|\cdot\|_2$ to denote the spectral norm and $\|\cdot\|_{\infty}$ to denote the entrywise max norm for matrices, we use the definition of spectral norm in the first inequality, we use the matrix inequality in the second inequality, we use the Assumption \ref{as:lip} in the third inequality, we use Lemma \ref{lem:beta}, Assumption \ref{as:weight}, Assumption \ref{as:model} and the law of large numbers in the forth inequality.

For $\widehat{U}_{n2122}(t)$, we have
\begin{align*}
&\mathbb{E}[\widehat{U}_{n2122}(t)]\\
\leq&3n \|\widehat{\beta}_n-\beta_0\|_2^4\times \mathbb{E}\bigg[ \int_{\mathbb{R}}\bigg\| \frac{1}{n}\sum_{i=1}^{n}g_0(X_i)\sin(t\varepsilon_i)\Ddot{m}(X_i,\beta_0) - W_2(t) \bigg\|_2^2t^2 \varphi(t) dt \bigg]\\
\leq&3n \|\widehat{\beta}_n-\beta_0\|_2^4\times \mathbb{E}\bigg[ \int_{\mathbb{R}}\bigg\| \frac{1}{n}\sum_{i=1}^{n}g_0(X_i)\sin(t\varepsilon_i)\Ddot{m}(X_i,\beta_0) - W_2(t) \bigg\|_F^2 t^2\varphi(t) dt \bigg]\\
\leq&3n \|\widehat{\beta}_n-\beta_0\|_2^4\times \sum_{k=1}^{p}\sum_{l=1}^{p}\mathbb{E}\bigg[ \int_{\mathbb{R}}\bigg| \frac{1}{n}\sum_{i=1}^{n}g_0(X_i)\sin(t\varepsilon_i)\Ddot{m}_{kl}(X_i,\beta_0) - W_{2,kl}(t) \bigg|^2 t^2\varphi(t) dt \bigg]\\
\leq&3n \|\widehat{\beta}_n-\beta_0\|_2^4\times \sum_{k=1}^{p}\sum_{l=1}^{p} \int_{\mathbb{R}}\mathbb{V}\bigg\{ \frac{1}{n}\sum_{i=1}^{n}g_0(X_i)\sin(t\varepsilon_i)\Ddot{m}_{kl}(X_i,\beta_0)\bigg\} t^2\varphi(t) dt \\
\leq & Cn\times\frac{p^2}{n^2}\times\frac{p^2}{n}=o(1),
\end{align*}
where we use the definition of spectral norm in the first inequality, we use the notation $\|\cdot\|_F$ to denote the Frobenius Norm and the definition of Frobenius Norm in the second inequality, and we use the Lemma \ref{lem:beta} and Assumption \ref{as:weight} in the fifth inequality. For $\widehat{U}_{n2123}(t)$, we have
\begin{align*}
    &\widehat{U}_{n2123}(t)
    \leq 3Cn\|\widehat{\beta}_n-\beta_0\|_2^4\times\int_{\mathbb{R}} t^2\varphi(t) dt\leq 3Cn\times\frac{p^2}{n^2}\times\int_{\mathbb{R}} t^2\varphi(t) dt=o_p(1),
\end{align*}
where we use the definition of spectral norm in the first inequality, and we use the Lemma \ref{lem:beta} in the second inequality, and Assumption \ref{as:weight} in the third inequality.

Similar, for $\widehat{U}_{n23}(t)$, by replacing the $\sin(t\varepsilon)$ in the proof of $\widehat{U}_{n21}(t)$ to $\cos(t\varepsilon)$, we have
\begin{equation*}
    \widehat{U}_{n23}(t)=\frac{1}{\sqrt{n}}
\sum_{i=1}^n M_2(t)^{\top}\Sigma^{-1}\dot{m}(X_i,\beta_0)\varepsilon_i+R_{n23}(t),
\end{equation*}
where 
\begin{equation*}
    M_2(t)=-\mathbb{E}[g_0(X)\cos(t\varepsilon)\dot{m}(X,\beta_0)]t
\end{equation*}
and
\begin{equation*}
    \int_{\mathbb{R}} |R_{n23}(t)|^2\varphi(t)dt=o_p(1).
\end{equation*}

Next, we deal with the term $\widehat{U}_{n22}(t)$. Note that
\begin{align*}
	\widehat{U}_{n22}(t) =&\frac{1}{2\sqrt{n}}\sum_{i=1}^{n}g_0(X_i)\cos(t\varepsilon_i)(\widehat{e}_i-\varepsilon_i)^2t^2 \\
	=& \frac{1}{2\sqrt{n}}\sum_{i=1}^{n}g_0(X_i)\cos(t\varepsilon_i)(m(X_i,\widehat{\beta}_n)-m(X_i,\beta_0))^2t^2\\
	=& \underbrace{\frac{1}{\sqrt{n}}\sum_{i=1}^{n}g_0(X_i)\cos(t\varepsilon_i)\{(\widehat{\beta}_n-\beta_0)^{\top}\dot{m}(X_i,\beta_0)\}^2t^2}_{\widehat{U}_{n221}(t)}\\
    &+\underbrace{\frac{1}{\sqrt{n}}\sum_{i=1}^{n}g_0(X_i)\cos(t\varepsilon_i)[(\widehat{\beta}_n-\beta_0)^{\top}\{\dot{m}(X_i,\beta_0)-\dot{m}(X_i,\Tilde\beta)\}]^2t^2}_{\widehat{U}_{n222}(t)},
\end{align*}
where we use Taylor's expansion and $\Tilde{\beta}$ lies between $\widehat{\beta}_n$ and $\beta_0$. Let
\begin{eqnarray*}
	W_3(t)=\mathbb{E}[g_0(X)\cos(t\varepsilon)\dot{m}(X,\beta_0)\dot{m}^{\top}(X,\beta_0)].
\end{eqnarray*}
Then, we have
\begin{align*}
	& 4\int_{\mathbb{R}}|\widehat{U}_{n221}(t)|^2 \varphi(t)dt \\
	=&\int_{\mathbb{R}}\bigg|\frac{1}{\sqrt{n}}\sum_{i=1}^{n}g_0(X_i)\cos(t\varepsilon_i)(\widehat{\beta}_n-\beta_0)^{\top}\dot{m}(X_i,\beta_0)\dot{m}(X_i,\beta_0)^{\top}(\widehat{\beta}_n-\beta_0)\bigg|^2t^4\varphi(t)dt
	\\
	 \leq & n
	\int_{\mathbb{R}}\bigg|(\widehat{\beta}_n-\beta_0)^{\top}\bigg\{\frac{1}{n}\sum_{i=1}^{n}g_0(X_i)\cos(t\varepsilon_i)\dot{m}(X_i,\beta_0)\dot{m}(X_i,\beta_0)^{\top}-W_3(t)\bigg\}(\widehat{\beta}_n-\beta_0)\bigg|^2t^4 \varphi(t) dt\\
	&+n\int_{\mathbb{R}} |(\widehat{\beta}_n-\beta_0)^{\top}W_3(t)(\widehat{\beta}_n-\beta_0)|^2t^4\varphi(t) dt\\
	 \leq & n\|\widehat{\beta}_n-\beta_0\|_2^4\times \int_{\mathbb{R}} \bigg\|\frac{1}{n}\sum_{i=1}^{n}g_0(X_i)\cos(t\varepsilon_i)\dot{m}(X_i,\beta_0)\dot{m}(X_i,\beta_0)^{\top}-W_3(t)\bigg\|_2^2t^4 \varphi(t) dt\\
	& + C n\| \widehat{\beta}_n-\beta_0 \|_2^4,
\end{align*}
where we use the triangle inequality in the second inequality, and we use the Assumption \ref{as:eigen_F} in the third inequality. Note that
\begin{align*}
	& \mathbb{E}\bigg[ \int_{\mathbb{R}} \bigg\| \frac{1}{n}\sum_{i=1}^{n}g_0(X_i)\cos(t\varepsilon_i)\dot{m}(X_i,\beta_0)\dot{m}(X_i,\beta_0)^{\top} - W_3(t) \|^2 t^4\varphi(t) dt \bigg]\\
	=&  \sum_{k=1}^{p}\sum_{l=1}^{p}\int_{\mathbb{R}} \mathbb{E} \bigg[ \frac{1}{n}\sum_{i=1}^{n}g_0(X_i)\cos(t\varepsilon_i)\dot{m}_k(X_i,\beta_0)\dot{m}_l(X_i,\beta_0) -W_{3,kl}(t)\bigg]^2t^4   \varphi(t) dt \\
\leq &\frac{1}{n} \sum_{k=1}^{p}\sum_{l=1}^{p} \int_{\mathbb{R}} \mathbb{V} \{g_0(X_i)\cos(t\varepsilon_i)\dot{m}_k(X_i,\beta_0)\dot{m}_l(X_i,\beta_0)\}t^4   \varphi(t) dt \\
	 \leq & \frac{1}{n} \sum_{k=1}^{p}\sum_{l=1}^{p} \mathbb{E}[g_0(X)\cos(t\varepsilon)\dot{m}_j(X,\beta_0)\dot{m}_k(X,\beta_0)]^2 \leq C \frac{p^2}{n}.
\end{align*}
Combining this with Lemma \ref{lem:beta} and Assumption \ref{as:weight}, we obtain that
$$ \int_{\mathbb{R}}|\widehat{U}_{n211}(t)|^2 \varphi(t)dt = O_p(p^2/n)=o_p(1) .$$
For the remaining term, we have
\begin{align*}
    &4\int_{\mathbb{R}}| \widehat{U}_{n222}(t)|^2 \varphi(t)dt\\
    =&n\int_{\mathbb{R}}\bigg|\frac{1}{n}\sum_{i=1}^{n}g_0(X_i)\cos(t\varepsilon_i)[(\widehat{\beta}_n-\beta_0)^{\top}\{\dot{m}(X_i,\beta_0)-\dot{m}(X_i,\Tilde\beta)\}]^2\bigg|^2t^4\varphi(t)dt\\
    \leq&\|\widehat{\beta}_n-\beta_0\|_2^4\times n\int_{\mathbb{R}}\bigg|\frac{1}{n}\sum_{i=1}^{n}g_0(X_i)\cos(t\varepsilon_i)\times\|\dot{m}(X_i,\beta_0)-\dot{m}(X_i,\Tilde\beta)\|_2^2\bigg|^2t^4\varphi(t)dt\\
    \leq&\|\widehat{\beta}_n-\beta_0\|_2^4\times np^2\int_{\mathbb{R}}\bigg|\frac{1}{n}\sum_{i=1}^{n}g_0(X_i)\cos(t\varepsilon_i)\times\|\dot{m}(X_i,\beta_0)-\dot{m}(X_i,\Tilde\beta)\|_\infty^2\bigg|^2t^4\varphi(t)dt\\
    \leq&\|\widehat{\beta}_n-\beta_0\|_2^8\times np^4\int_{\mathbb{R}}\bigg|\frac{1}{n}\sum_{i=1}^{n}F(X_i)^3\bigg|^2t^4\varphi(t)dt\\
    =&O_p\bigg(\frac{p^4}{n^4}\times np^4\bigg)=o_p(1),
\end{align*}
where we use the Cauchy–Schwarz inequality in the second inequality, we use the inequality $\|\cdot\|_2\leq\sqrt{p}\|\cdot\|_{\infty}$, we use the Assumption \ref{as:lip} in the forth inequality, and we use Lemma \ref{lem:beta} and Assumption \ref{as:eigen} and \ref{as:weight} in the last inequality. Thus, we have
\begin{equation*}
     \int_{\mathbb{R}}	|\widehat{U}_{n22}(t)|^2 \varphi(t)dt=o_p(1).
\end{equation*}
Similarly, we have
\begin{equation*}
     \int_{\mathbb{R}}	|\widehat{U}_{n24}(t)|^2 \varphi(t)dt=o_p(1).
\end{equation*}

\textit{(2) Second}, we deal with the term $\widehat{U}_{n4}(t)$. Note, that
\begin{equation*}
	\widehat{U}_{n4}(t) =  \frac{1}{\sqrt{n}}\sum_{i=1}^{n} g_0(X_i) \frac{1}{n} \sum_{i=1}^{n} (\cos(t\widehat{e}_i) -\cos(t \varepsilon_i) ).
\end{equation*}
By Taylor expansion, we have
\begin{align*}
	\widehat{U}_{n4}(t)=& - \frac{1}{\sqrt{n}}\sum_{i=1}^{n} g_0(X_i) \frac{1}{n} \sum_{i=1}^{n}\sin(t\Tilde{\varepsilon}_i)(\widehat{e}_i-\varepsilon_i)t
\end{align*}
where $\Tilde{\varepsilon}_i$ lies between $\widehat{e}_i$ and $\varepsilon_i$. Furthermore, by Taylor's expansion, we can obtain
\begin{align*}
	\widehat{U}_{n4}(t) =& - \frac{1}{\sqrt{n}}\sum_{i=1}^{n} g_0(X_i) \frac{1}{n} \sum_{i=1}^{n}\sin(t\Tilde{\varepsilon}_i)(\widehat{e}_i-\varepsilon_i)t\\
	=& \frac{1}{\sqrt{n}}\sum_{i=1}^{n} g_0(X_i) \frac{1}{n} \sum_{i=1}^{n} \sin(t\Tilde{\varepsilon}_i)\{m(X_i,\widehat{\beta}_n)-m(X_i,\beta_0)\}t\\
	=& \underbrace{\frac{1}{\sqrt{n}}\sum_{i=1}^{n} g_0(X_i) \frac{1}{n} \sum_{i=1}^{n} \sin(t\Tilde{\varepsilon}_i)\{(\widehat{\beta}_n-\beta_0)^{\top}\dot{m}(X_i,\beta_0)\}t}_{\widehat{U}_{n41}(t)}\\
	&+\underbrace{\frac{1}{\sqrt{n}}\sum_{i=1}^{n} g_0(X_i) \frac{1}{n} \sum_{i=1}^{n} \sin(t\Tilde{\varepsilon}_i)[(\widehat{\beta}_n-\beta_0)^{\top}\{\dot{m}(X_i,\Tilde{\beta})-\dot{m}(X_i,\beta_0)\}]t}_{\widehat{U}_{n42}(t)}
\end{align*}
where $\Tilde{\beta}$ lies between $\widehat{\beta}_n$ and $\beta_0$. For the term $\widehat{U}_{n41}(t)$, we have
\begin{align*}
	&\int_{\mathbb{R}}	|\widehat{U}_{n41}(t)|^2 \varphi(t)dt \\
    =& \int_{\mathbb{R}}\bigg|\frac{1}{\sqrt{n}}\sum_{i=1}^{n} g_0(X_i) \frac{1}{n} \sum_{i=1}^{n} \sin(t\Tilde{\varepsilon}_i)(\widehat{\beta}_n-\beta_0)^{\top}\dot{m}(X_i,\beta_0)\bigg|^2t^2 \varphi(t)dt\\
	=& \int_{\mathbb{R}}\bigg|\frac{1}{\sqrt{n}}\sum_{i=1}^{n} g_0(X_i)\bigg|^2\cdot\bigg|\frac{1}{n} \sum_{i=1}^{n} \sin(t\Tilde{\varepsilon}_i)(\widehat{\beta}_n-\beta_0)^{\top}\dot{m}(X_i,\beta_0)\bigg|^2t^2\varphi(t)dt\\
	 \leq &\underbrace{3\int_{\mathbb{R}}\bigg|\frac{1}{\sqrt{n}}\sum_{i=1}^{n} g_0(X_i)\bigg|^2\cdot\bigg|(\widehat{\beta}_n-\beta_0)^{\top}\bigg[\frac{1}{n} \sum_{i=1}^{n} \sin(t\varepsilon_i)\dot{m}(X_i,\beta_0)-W_4(t)\bigg]\bigg|^2t^2\varphi(t)dt}_{\widehat{U}_{n411}(t)}\\
     &+\underbrace{3\int_{\mathbb{R}}\bigg|\frac{1}{\sqrt{n}}\sum_{i=1}^{n} g_0(X_i)\bigg|^2\cdot\bigg|\frac{1}{n} \sum_{i=1}^{n} \{\sin(t\varepsilon_i)-\sin(t\Tilde{\varepsilon}_i)\}(\widehat{\beta}_n-\beta_0)^{\top}\dot{m}(X_i,\beta_0)\bigg|^2t^2\varphi(t)dt}_{\widehat{U}_{n412}(t)}\\
	&+   \underbrace{3\int_{\mathbb{R}}\bigg|\frac{1}{\sqrt{n}}\sum_{i=1}^{n} g_0(X_i)\bigg|^2\cdot|(\widehat{\beta}_n-\beta_0)^{\top}W_4(t)|^2t^2\varphi(t)dt}_{\widehat{U}_{n413}(t)},
\end{align*}
where we use the inequality $(a+b+c)^2\leq 3a^2+3b^2+3c^2$ in the third inequality and we introduce the notation $W_4(t):=\mathbb{E}[\sin(t\varepsilon)\dot{m}(X,\beta_0)]$. For the term $\widehat{U}_{n411}(t)$, we have
\begin{align*}
\mathbb{E}[\widehat{U}_{n411}(t)]&\leq O(\|\widehat{\beta}_n-\beta_0\|_2^2)\times\int_{\mathbb{R}}\mathbb{E}\bigg\|\frac{1}{n} \sum_{i=1}^{n} \sin(t\varepsilon_i)\dot{m}(X_i,\beta_0)-W_4(t)\bigg\|^2t^2\varphi(t)dt\\
&\leq O\bigg(\frac{p}{n}\bigg)\times\sum_{j=1}^{p}\frac{1}{n}\mathbb{V}\{ \sin(t\varepsilon_i)\dot{m}_j(X_i,\beta_0)\}\times\int_{\mathbb{R}}t^2\varphi(t)dt=o(1),
\end{align*}
where we use the Cauchy–Schwarz inequality in the first inequality, we use Lemma \ref{lem:beta} and $\dot{m}_j(X_i,\beta_0)$ to denote $j$th coordinate of $\dot{m}(X_i,\beta_0)$ in the second inequality, and we use Assumption \ref{as:weight} in the last inequality. For the term $\widehat{U}_{n412}(t)$, we have
\begin{align*}
&\mathbb{E}[\widehat{U}_{n412}(t)]\\
\leq& O(\|\widehat{\beta}_n-\beta_0\|_2^2)\times\int_{\mathbb{R}}\mathbb{E}\bigg\|\frac{1}{n} \sum_{i=1}^{n} \{\sin(t\varepsilon_i)-\sin(t\Tilde{\varepsilon}_i)\}\dot{m}(X_i,\beta_0)\bigg\|^2t^2\varphi(t)dt\\
\leq &O(\|\widehat{\beta}_n-\beta_0\|_2^2)\times\sum_{j=1}^{p}\int_{\mathbb{R}}\mathbb{E}\bigg|\frac{1}{n} \sum_{i=1}^{n} \{\sin(t\varepsilon_i)-\sin(t\Tilde{\varepsilon}_i)\}\dot{m}_j(X_i,\beta_0)\bigg|^2t^2\varphi(t)dt\\
\leq &O(\|\widehat{\beta}_n-\beta_0\|_2^2)\times\sum_{j=1}^{p}\int_{\mathbb{R}}\mathbb{E}\bigg|\frac{1}{n} \sum_{i=1}^{n} 2\cdot F(X_i) \bigg|^2t^2\varphi(t)dt\\
\leq &O(\|\widehat{\beta}_n-\beta_0\|_2^2)\times\sum_{j=1}^{p}\int_{\mathbb{R}}\mathbb{E}\{4F(X)^2\}t^2\varphi(t)dt=o(1),
\end{align*}
where we use the Cauchy–Schwarz inequality in the first inequality, we use Assumption \ref{as:model} in the third and fourth inequality, and we use Assumption \ref{as:weight} and Lemma \ref{lem:beta} in the last equality.

For the term $\widehat{U}_{n413}(t)$, we have
\begin{align*}
&\mathbb{E}[\widehat{U}_{n413}(t)]\\
\leq& O(\|\widehat{\beta}_n-\beta_0\|_2^2)\times\int_{\mathbb{R}}\|W_4(t)\|^2t^2\varphi(t)dt\\
\leq& O(\|\widehat{\beta}_n-\beta_0\|_2^2)\times p\int_{\mathbb{R}}t^2\varphi(t)dt=o(1),
\end{align*}
where we use the Cauchy–Schwarz inequality in the first inequality, and we use Assumption \ref{as:weight} and Lemma \ref{lem:beta} in the last equality.

For the term $\widehat{U}_{n42}(t)$, we have
\begin{align*}
	&\int_{\mathbb{R}}	|\widehat{U}_{n42}(t)|^2 \varphi(t)dt \\
   =&\int_{\mathbb{R}}\bigg|\frac{1}{\sqrt{n}}\sum_{i=1}^{n} g_0(X_i)\bigg|^2\times\bigg| \frac{1}{n} \sum_{i=1}^{n} \sin(t\Tilde{\varepsilon}_i)[(\widehat{\beta}_n-\beta_0)^{\top}\{\dot{m}(X_i,\Tilde{\beta})-\dot{m}(X_i,\beta_0)\}]\bigg|^2t^2\varphi(t)dt\\ 
    \leq&\int_{\mathbb{R}}\bigg|\frac{1}{\sqrt{n}}\sum_{i=1}^{n} g_0(X_i)\bigg|^2\times\|\widehat{\beta}_n-\beta_0\|_2^2\times \bigg\| \frac{1}{n} \sum_{i=1}^{n} \sin(t\Tilde{\varepsilon}_i)\{\dot{m}(X_i,\Tilde{\beta})-\dot{m}(X_i,\beta_0)\}\bigg\|^2t^2\varphi(t)dt \\
    \leq&\int_{\mathbb{R}}\bigg|\frac{1}{\sqrt{n}}\sum_{i=1}^{n} g_0(X_i)\bigg|^2\times\|\widehat{\beta}_n-\beta_0\|_2^2\times p\bigg| \frac{1}{n} \sum_{i=1}^{n} \sqrt{p}\|\Tilde{\beta}-\beta_0\|_2 F(X_i)\bigg|^2t^2\varphi(t)dt \\
    \leq&\int_{\mathbb{R}}\bigg|\frac{1}{\sqrt{n}}\sum_{i=1}^{n} g_0(X_i)\bigg|^2\times\|\widehat{\beta}_n-\beta_0\|_2^4\times p^2\bigg| \frac{1}{n} \sum_{i=1}^{n} F(X_i)\bigg|^2t^2\varphi(t)dt \\
    =&O_p\bigg(\frac{p^2}{n^2}\times p^2\bigg)=o_p(1),
\end{align*}
where we use the Cauchy–Schwarz inequality in the second inequality, we use Assumption \ref{as:lip} in the third inequality, we use Lemma \ref{lem:beta} and Assumption \ref{as:weight} in the last inequality. Thus, we have $\int_{\mathbb{R}}	|\widehat{U}_{n4}(t)|^2 \varphi(t)dt=o_p(1)$. 

Similarly, the proof for $\widehat{U}_{n4}(t)$ carries over to $\widehat{U}_{n5}(t)$ by replacing $\cos(\cdot)$ with $\sin(\cdot)$, which yields 
\begin{equation*}
\int_{\mathbb{R}} |\widehat{U}_{n5}(t)|^2 \varphi(t)\,dt = o_p(1).    
\end{equation*}

\textit{(3) Third}, we deal with the term $\widehat{U}_{n3}(t)$. Note, that
\begin{align*}
	\widehat{U}_{n3}(t)  =& \frac{1}{\sqrt{n}}\sum_{i=1}^{n} g_0(X_i) \frac{1}{n} \sum_{i=1}^{n} \{\cos(t \varepsilon_i)+\sin(t \varepsilon_i)\}\\
	=& \underbrace{\frac{1}{\sqrt{n}}\sum_{i=1}^{n} g_0(X_i) \frac{1}{n} \sum_{i=1}^{n}[\cos(t \varepsilon_i)+\sin(t \varepsilon_i)-\mathbb{E}\{\cos(t \varepsilon_i)+\sin(t \varepsilon_i)\}]}_{\widehat{U}_{n31}(t)}\\
	& + \underbrace{\frac{1}{\sqrt{n}}\sum_{i=1}^{n} g_0(X_i) \mathbb{E}\{\cos(t \varepsilon_i)+\sin(t \varepsilon_i)\}}_{\widehat{U}_{n32}(t)}.
\end{align*}
For the term $\widehat{U}_{n31}(t)$, we have
\begin{align*}
	&\mathbb{E}\int_{\mathbb{R}}	|\widehat{U}_{n31}(t)|^2 \varphi(t)dt \\
    =&\int_{\mathbb{R}}	\mathbb{E}\bigg|\frac{1}{\sqrt{n}}\sum_{i=1}^{n} g_0(X_i) \bigg|^2\times\bigg|\frac{1}{n} \sum_{i=1}^{n}[\cos(t \varepsilon_i)+\sin(t \varepsilon_i)-\mathbb{E}\{\cos(t \varepsilon_i)+\sin(t \varepsilon_i)\}]\bigg|^2 \varphi(t)dt\\
    =&O(1)\times\frac{1}{n}\int_{\mathbb{R}}	\mathbb{V}\{\cos(t \varepsilon)+\sin(t \varepsilon)\} \varphi(t)dt=o(1).
\end{align*}

Altogether we obtain
\begin{eqnarray*}
	\widehat{U}_{n}(t)
	&=&  \frac{1}{\sqrt{n}}\sum_{i=1}^{n}\big(g_0(X_i)[\cos(t\varepsilon_i)+\sin(t\varepsilon_i) -  \mathbb{E}\{\cos(t \varepsilon_i)+\sin(t \varepsilon_i)\} ]\\ 
    &&+ W(t)^{\top}\varepsilon_i\Sigma^{-1}\dot{m}(X_i,\beta_0)t \big)+R_n(t)  \\
	&=:& U_n(t)+R_n(t),
\end{eqnarray*}
where the remainder $R_n(t)$ satisfies $\int_{\mathbb{R}}|R_n(t)|^2 \varphi(t)dt =o_p(1).$

\textbf{Step 2}. Next, we establish that
\begin{equation*}
    U_n \Rightarrow U_\infty \quad \text{in } \mathbb{C}(\Pi),
\end{equation*}
where $\Rightarrow$ denotes weak convergence of the whole stochastic process in the function space $\mathbb{C}(\Pi)$, $\Pi$ is an arbitrary compact subset of $\mathbb{R}$, $\mathbb{C}(\Pi)$ is the space of real-valued continuous functions on $\Pi$, and $U_\infty$ is a mean-zero Gaussian process with covariance function $K(s,t)$. Thus, by the Continuous Mapping Theorem, we obtain
\begin{eqnarray*}
	\int_{\Pi}	|U_n(t)|^2 \varphi(t)dt \stackrel{d}{\to}
	\int_{\Pi}	|U_\infty (t)|^2 \varphi(t)dt,
\end{eqnarray*}
where $\stackrel{d}{\to}$ denotes convergence in distribution.

First, we show that $U_n(t)$ is asymptotically tight. Recall that
\begin{align*}
U_n(t)=& \frac{1}{\sqrt{n}}\sum_{i=1}^{n}[g_0(X_i)(\cos(t\varepsilon_i)+\sin(t\varepsilon_i) -  \mathbb{E}\{\cos(t \varepsilon_i)+\sin(t \varepsilon_i)\} ] \\
&+ W(t)^{\top}\Sigma^{-1} \frac{1}{\sqrt{n}}\sum_{i=1}^{n}  \varepsilon_i\dot{m}(X_i,\beta_0)t
\end{align*}
For any $s,t \in \Pi $, we have
\begin{align*}
	&\mathbb{E}|U_n(s)-U_n(t)|^2  \\
	 \leq &  3\mathbb{E}\bigg|\frac{1}{\sqrt{n}}\sum_{i=1}^{n}\big(g_0(X_i)[\cos(s\varepsilon_i)-\mathbb{E}\{\cos(s \varepsilon_i)\}-\cos(t\varepsilon_i) + \mathbb{E}\{\cos(t \varepsilon_i)\}]\big)\bigg|^2 \\
	&+3\mathbb{E}\bigg|\frac{1}{\sqrt{n}}\sum_{i=1}^{n}\big(g_0(X_i)[\sin(s\varepsilon_i)-\mathbb{E}\{\sin(s \varepsilon_i)\}-\sin(t\varepsilon_i) +  \mathbb{E}\{\sin(t \varepsilon_i)\}]\big)\bigg|^2 \\
	&+ 3 \mathbb{E}\bigg|\{W(s)-W(t)\}^{\top} \Sigma^{-1} \frac{1}{\sqrt{n}}\sum_{i=1}^{n}  \varepsilon_i \dot{m}(X_i,\beta_0) \bigg|^2\\
=& 3\mathbb{E}\bigg|\frac{1}{\sqrt{n}}\sum_{i=1}^{n}\big(g_0(X_i)[\cos(s\varepsilon_i)-\mathbb{E}\{\cos(s \varepsilon_i)\}-\cos(t\varepsilon_i) + \mathbb{E}\{\cos(t \varepsilon_i)\}]\big)\bigg|^2 \\
	&+3\mathbb{E}\bigg|\frac{1}{\sqrt{n}}\sum_{i=1}^{n}\big(g_0(X_i)[\sin(s\varepsilon_i)-\mathbb{E}\{\sin(s \varepsilon_i)\}-\sin(t\varepsilon_i) +  \mathbb{E}\{\sin(t \varepsilon_i)\}]\big)\bigg|^2 \\
	& + 3 \{W(s)s-W(t)t\}^{\top}\Sigma^{-1} \mathbb{E}\{\varepsilon^2 \dot{m}(X,\beta_0) \dot{m}(X,\beta_0)^{\top}\} \Sigma^{-1}\{W(s)s-W(t)t\}\\
	 \leq & C\mathbb{E}\{g^2_0(X)(s \varepsilon-t \varepsilon)^2\}+ C \| W(s)s - W(t)t \|^2 \\
	 \leq & C(s-t)^2\mathbb{E}\{g^2_0(X) \varepsilon ^2\}+ C \bigg\| \frac{\partial B(\xi)}{\partial t }(s-t) \bigg\|^2,
\end{align*}
where we use the inequality $(a+b+c)^2\leq 3a^2+3b^2+3c^2$ in the first inequality, Taylor's expansion and Assumption \ref{as:eigen_F} in the third inequality, we use Taylor's expansion in the fourth inequality, and $\xi $ lies between $s$ and $t$. Note that
\begin{eqnarray*}
	\frac{\partial W(\xi)}{\partial t }= \mathbb{E}[g_0(X)\cos( \xi \varepsilon  )\dot{m}(X,\beta_0) ] \xi  + \mathbb{E}[ g_0(X)\sin( \xi \varepsilon  )\dot{m}(X,\beta_0) ] \xi,
\end{eqnarray*}
where is bounded by some constant. Thus, the coefficient to be controlled is $B(t)=tW_0(t)$. Its derivative is
$$
 B'(t)=W_0(t)+t\,\mathbb E[g_0(X)\varepsilon\{\cos(t\varepsilon)+\sin(t\varepsilon)\}\dot m(X,\beta_0)].
$$
On every compact $\Pi$, by the Cauchy--Schwarz inequality, the independence of $X$ and $\varepsilon$ under $H_0$, and Assumption~\ref{as:eigen_F}, we have $\sup_{t\in\Pi}\|B'(t)\|_2<\infty$. Consequently $\|B(s)-B(t)\|_2\le C_\Pi|s-t|$. Consequently, we have
$$ \mathbb{E}|U_n(s)-U_n(t)|^2 \leq C|s-t|^2 .$$
By Kolmogorov's tightness criterion, we have that $U_n(t)$ is asymptotically tight.

Now, we consider the convergence of finite-dimensional distributions of $U_n(t)$. For any $t_1,\cdots, t_m\in \Pi $, let $H_{ni}=(H_{ni}(t_1),\cdots, H_{ni}(t_m))^{\top}$, where
\begin{align*}
	H_{ni}(t)=& \frac{1}{\sqrt{n}}\bigg(g_0(X_i)[\cos(t\varepsilon_i)+\sin(t\varepsilon_i) -  \mathbb{E}\{\cos(t \varepsilon_i)+\sin(t \varepsilon_i)\} ] \\
    &+ W(t)^{\top}\varepsilon_i\Sigma^{-1}\dot{m}(X_i,\beta_0)t \bigg)
\end{align*}
It follows from the same arguments as in Theorem 3 of \cite{tan2025weighted} that $\{H_{ni}\}_{i=1}^{n}$ satisfies the conditions of the Lindeberg--Feller central limit theorem. Hence, $\{U_n(t_1),\dots,U_n(t_m)\}$ converges in distribution to a multivariate Gaussian distribution. Then, by Theorem 13.1 of \cite{billingsley2013convergence}, together with the convergence of the covariance matrices, we obtain that
\begin{eqnarray*}
	 U_n \Rightarrow U_\infty\text{ in } \mathbb{C}(\Pi),\quad\text{and}\int_{\Pi}	|U_n(t)|^2 \varphi(t)dt \stackrel{d}{\to}
	\int_{\Pi}	|U_\infty (t)|^2 \varphi(t)dt.
\end{eqnarray*}

For the covariance function of the Gaussian process $U(t)$, we only need to calculate $ \sum_{j=1}^n\text{Cov}\{H_{nj}(s),H_{ni}(t)\}$. Thus, we have
\begin{align*}
	&K_n (s,t)=n\text{Cov}\{H_{nj}(s),H_{nj}(t)\}\\
	&= \mathbb{E} \big( g^2_0(X) [\cos(s\varepsilon)+\sin(s\varepsilon) -  \mathbb{E}\{\cos(s \varepsilon)+\sin(s \varepsilon)\} ]\\
    &\times[\cos(t\varepsilon)+\sin(t\varepsilon) -  \mathbb{E}\{\cos(t \varepsilon)+\sin(t \varepsilon)\} ] \big) \\
	& +W(s)^{\top}\Sigma^{-1} \mathbb{E}\big( g_0(X)[\cos(t\varepsilon)+\sin(t\varepsilon) -  \mathbb{E}\{\cos(t \varepsilon)+\sin(t \varepsilon)\} ] \varepsilon  \dot{m}(X,\beta_0) \big)s \\
	&  +W(t)^{\top}\Sigma^{-1} E\big( g_0(X)[\cos(s\varepsilon)+\sin(s\varepsilon) -  \mathbb{E}\{\cos(s \varepsilon)+\sin(s \varepsilon)\} ] \varepsilon  \dot{m}(X,\beta_0) \big)t\\
	&  +W(s)^{\top}\Sigma^{-1} E\{ \varepsilon^2   \dot{m}(X,\beta_0) \dot{m}(X,\beta_0)^{\top}\}\Sigma^{-1} W(t)st.
\end{align*}

By conditions of Theorem \ref{theorem:1}, we have $K_n(s,t)\rightarrow K(s,t)$ point-wisely in $(s,t)$ for $(s,t)\in\Pi$. Then, it remains to show that the convergence from the compact set to $\mathbb R$, which needs a uniform tail bound.  Since the trigonometric part is bounded, $\sup_t\|W_0(t)\|_2<\infty$, and $\Sigma^{-1}$ is uniformly bounded, the covariance satisfies
$$
 \sup_n K_n(t,t)\le C(1+t^2).
$$
Assumption~\ref{as:weight} then gives
$$
 \lim_{M\to\infty}\sup_n\int_{|t|>M}K_n(t,t)\varphi(t)\,dt=0.
$$

\textbf{Step 3}. We will show that
\begin{eqnarray*}
	\int_{\mathbb{R}}	|U_n(t)|^2 \varphi(t)dt \rightarrow
	\int_{\mathbb{R}}	|U_\infty (t)|^2 \varphi(t)dt
	\quad {\rm in \ distribution  }.
\end{eqnarray*}
The proof follows the same line as that of Theorem 3 in \cite{tan2025weighted}. By Assumption \ref{as:model} and \ref{as:weight}, we have $\int_{\mathbb{R}} K(t,t)\varphi(t)dt< \infty$. Thus, fix any $\epsilon > 0$, there exists a compact subset $\Pi \subset \mathbb{R}$ such that
\begin{eqnarray*}
	\int_{\Pi^c} K(t,t)\varphi(t)dt< \frac{\epsilon^2}{2},
\end{eqnarray*}
where $\Pi^c$ is the complementary set of $\Pi$ in $\mathbb{R}$. Since $K(t,t)=\mathbb{E}|U_\infty (t)|^2$, it follows that
\begin{eqnarray*}
	\mathbb{E}\big(  \int_{\Pi^c} U_\infty (t)^2 \varphi(t)dt  \big) = 	\int_{\Pi^c} \mathbb{E} | U _\infty (t)|^2  \varphi(t)dt  < \frac{\epsilon^2}{2}.
\end{eqnarray*}
Moreover, since
\begin{eqnarray*}
	\mathbb{E} | U _n (t)|^2 =K_n(t,t) \quad {\rm and } \quad   K_n(t,t) \stackrel{n\to\infty}{\rightarrow} K(t,t),
\end{eqnarray*}
we have
\begin{eqnarray*}
	\mathbb{E}\bigg(  \int_{\Pi^c} U_n (t)^2 \varphi(t)dt  \bigg) = 	\int_{\Pi^c} K_n(t,t) \varphi(t)dt< \epsilon^2,
\end{eqnarray*}
for $n$ large enough. 

Now define
\begin{equation*}
J_n := \int_{\mathbb{R}} |U_n(t)|^2 \varphi(t)\,dt,
\qquad
J := \int_{\mathbb{R}} |U_\infty(t)|^2 \varphi(t)\,dt,
\end{equation*}
and decompose
\begin{equation*}
J_n = J_{n,\Pi}+R_{n,\Pi}, \qquad
J = J_{\Pi}+R_{\Pi},
\end{equation*}
where
\begin{equation*}
J_{n,\Pi} := \int_{\Pi} |U_n(t)|^2 \varphi(t)\,dt,
\qquad
R_{n,\Pi} := \int_{\Pi^c} |U_n(t)|^2 \varphi(t)\,dt,
\end{equation*}
and
\begin{equation*}
J_{\Pi} := \int_{\Pi} |U_\infty(t)|^2 \varphi(t)\,dt,
\qquad
R_{\Pi} := \int_{\Pi^c} |U_\infty(t)|^2 \varphi(t)\,dt.
\end{equation*}
Since $U_n \Rightarrow U_\infty$ in $\mathbb{C}(\Pi)$ and the map $f \mapsto \int_{\Pi} |f(t)|^2 \varphi(t)dt$ is continuous on $\mathbb{C}(\Pi)$, the continuous mapping theorem yields $J_{n,\Pi} \Rightarrow J_{\Pi}$. On the other hand, by Markov's inequality,
\begin{equation*}
\mathbb{P}(R_{n,\Pi}>\epsilon)
\le \frac{\mathbb{E}(R_{n,\Pi})}{\epsilon}
< \epsilon,
\qquad
\mathbb{P}(R_{\Pi}>\epsilon)
\le \frac{\mathbb{E}(R_{\Pi})}{\epsilon}
< \epsilon/2.
\end{equation*}

Let $x$ be any continuity point of the distribution function of $J$. Since $R_{n,\Pi}\ge 0$, we have
\begin{equation*}
\mathbb{P}(J_n\le x)
\le \mathbb{P}(J_{n,\Pi}\le x),
\end{equation*}
and
\begin{equation*}
\mathbb{P}(J_n\le x)
\ge \mathbb{P}(J_{n,\Pi}\le x-\epsilon)-\mathbb{P}(R_{n,\Pi}>\epsilon).
\end{equation*}
Therefore,
\begin{equation*}
\limsup_{n\to\infty}\mathbb{P}(J_n\le x)
\le \limsup_{n\to\infty}\mathbb{P}(J_{n,\Pi}\le x)
= \mathbb{P}(J_{\Pi}\le x),
\end{equation*}
and
\begin{equation*}
\liminf_{n\to\infty}\mathbb{P}(J_n\le x)
\ge \liminf_{n\to\infty}\mathbb{P}(J_{n,\Pi}\le x-\epsilon)-\epsilon
= \mathbb{P}(J_{\Pi}\le x-\epsilon)-\epsilon.
\end{equation*}
Since $R_{\Pi}\ge 0$, we also have
\begin{equation*}
\mathbb{P}(J\le x-\epsilon)\le \mathbb{P}(J_{\Pi}\le x-\epsilon),
\end{equation*}
and
\begin{equation*}
\mathbb{P}(J_{\Pi}\le x)
\le \mathbb{P}(J\le x+\epsilon)+\mathbb{P}(R_{\Pi}>\epsilon).
\end{equation*}
Hence
\begin{equation*}
\mathbb{P}(J\le x-\epsilon)-\epsilon
\le
\liminf_{n\to\infty}\mathbb{P}(J_n\le x)
\le
\limsup_{n\to\infty}\mathbb{P}(J_n\le x)
\le
\mathbb{P}(J\le x+\epsilon)+\epsilon/2.
\end{equation*}
Letting $\epsilon\to 0$ yields
\begin{equation*}
\lim_{n\to\infty}\mathbb{P}(J_n\le x)=\mathbb{P}(J\le x),
\end{equation*}
for every continuity point $x$ of the law of $J$. Consequently,
\begin{eqnarray*}
	\int_{\mathbb{R}}| {U}_n(t) |^2 \varphi(t)dt  \rightarrow  \int_{\mathbb{R}}| U_ \infty (t) |^2  \varphi(t)dt  \quad {\rm in \ distribution. }
\end{eqnarray*}

In Step 1, we have shown that $\widehat{U}_n(t)=U_n(t)+R_n(t)$, where
\begin{eqnarray*}
	\int_{\mathbb{R}} |R_n(t)|^2 \varphi(t)dt=o_p(1).
\end{eqnarray*}
Finally, the transfer from \(U_n\) to \(\widehat U_n=U_n+R_n\) follows
from
\[
 \left|\|\widehat U_n\|_\varphi^2-\|U_n\|_\varphi^2\right|
 \le \|R_n\|_\varphi
      \{2\|U_n\|_\varphi+\|R_n\|_\varphi\}=o_p(1),
\]
because \(\|U_n\|_\varphi=O_p(1)\) by Step 3. Thus we have
\begin{eqnarray*}
	\int_{\mathbb{R}}| \widehat{U}_n(t) |^2 \varphi(t)dt  \stackrel{d}{\to}\int_{\mathbb{R}}| U_ \infty (t) |^2  \varphi(t)dt.
\end{eqnarray*}
Hence, we complete the proof of Theorem \ref{theorem:1}.     
\end{proof}

\begin{proof}[Proof of Theorem \ref{thm:alt}]
Recall that
$$ 	
\widehat{U}_n(t)=\frac{1}{\sqrt{n}}\sum_{i=1}^{n}\{g(X_i)-\bar{g}\}\{\cos(t\widehat{e}_i)+\sin(t\widehat{e}_i)\},$$
where $\bar{g}=n^{-1}\sum_{i=1}^{n}g(X_i)$, $ t \in \mathbb{R}$. Under the alternative $H_{n1}$, we have
$$\widehat{e}_i=Y_i-m(X_i,\widehat{\beta}_n)=\varepsilon_i+m(X_i,\beta_0)+r_nS(X_i)-m(X_i,\widehat{\beta}_n).$$
Throughout the proof, the occurrence of $\beta_0$ in the preceding display is understood as $\widetilde\beta_0$.  Let
$$
 e_{n,i}=\varepsilon_i+r_nS(X_i),\qquad
 \|f\|_\varphi^2=\int_{\mathbb R}|f(t)|^2\varphi(t)\,dt .
$$
By Assumption~\ref{as:model}, we have $|S(X)|\le F(X)$ almost surely. The strict positivity statements also require the detectability conditions
$$
 \|K^{(1)}\|_\varphi>0\quad\text{and}\quad\|tK^{(2)}\|_\varphi>0
$$
in the corresponding cases. These conditions are necessary because the regularity assumptions alone do not rule out a weight $g(\cdot)$ that is orthogonal to the departure $S(\cdot)$.  By Lemma~\ref{lem:beta} and $\mathbb E\{S(X)\dot m(X,\widetilde\beta_0)\}=0$, we have
$$
 \|\widehat\beta_n-\widetilde\beta_0\|_2=O_p(\sqrt{p/n}),
$$
$$
 \sqrt n(\widehat\beta_n-\widetilde\beta_0)
 =\frac1{\sqrt n}\sum_{i=1}^n
 e_{n,i}\Sigma^{-1}\dot m(X_i,\widetilde\beta_0)
 +O_p(\sqrt{p^3/n}).
$$
Thus the parameter-estimation remainder bounds invoked from
Theorem~\ref{theorem:1} remain valid under \(H_{1n}\).

\textbf{Case 1}: We first consider the fixed alternative, corresponding to $\alpha=0$ and hence $r_n=1$. Similar to the arguments for $\widehat{U}_n$ in the proof of Theorem \ref{theorem:1}, $\widehat{U}_n$ can be decomposed into five parts:
\begin{align*}
	&\frac{1}{\sqrt{n}} \widehat{U}_n(t)
	= \underbrace{\frac{1}{n}\sum_{i=1}^{n}g_0(X_i)\{\cos(t\varepsilon_i)+\sin(t\varepsilon_i)\}}_{\widehat{U}^{(1)}_{n1}(t) }\\
    &+ \underbrace{\frac{1}{n}\sum_{i=1}^{n}g_0(X_i)\{\cos(t\widehat{e}_i)-\cos(t\varepsilon_i)+\sin(t\widehat{e}_i)-\sin(t\varepsilon_i)\}}_{\widehat{U}^{(1)}_{n2}(t) } \\
	& - \underbrace{\frac{1}{n}\sum_{i=1}^{n} g_0(X_i) \frac{1}{n} \sum_{i=1}^{n} \{\cos(t \varepsilon_i)+\sin(t \varepsilon_i)\}}_{\widehat{U}^{(1)}_{n3}(t)}\\
    &-\underbrace{\frac{1}{n}\sum_{i=1}^{n} g_0(X_i) \frac{1}{n} \sum_{i=1}^{n} \{\cos(t\widehat{e}_i)-\cos(t \varepsilon_i)\}}_{\widehat{U}^{(1)}_{n4}(t)}\\
	&- \underbrace{\frac{1}{n}\sum_{i=1}^{n} g_0(X_i) \frac{1}{n}\sum_{i=1}^{n}\{
	\sin(t\widehat{e}_i)- \sin(t \varepsilon_i)\}}_{\widehat{U}^{(1)}_{n5}(t)}.
\end{align*}

First, we consider the term $\widehat{U}^{(1)}_{n1} (t)$. By an argument analogous to that used for $\widehat{U}_{n1} (t)$ in Theorem \ref{theorem:1}, we have
\begin{eqnarray*}
	\int_{\mathbb{R}} |\widehat{U}^{(1)}_{n1} (t)| ^2 \varphi(t)dt = O_p(1/n).
\end{eqnarray*}
Similarly, we can obtain
\begin{eqnarray*}
	\int_{\mathbb{R}} |\widehat{U}^{(1)}_{n3} (t)| ^2 \varphi(t)dt = O_p(1/n).
\end{eqnarray*}

It remains to consider the remaining terms. For the second term $\widehat{U}^{(1)}_{n2}(t)$, we have
\begin{align*}
	\widehat{U}^{(1)}_{n2}(t)=& \frac{1}{n}\sum_{i=1}^{n}g_0(X_i)\{\cos(t\widehat{e}_i)-\cos(t\varepsilon_i)+\sin(t\widehat{e}_i)-\sin(t\varepsilon_i)\} \\
    =&\underbrace{\frac{1}{n}\sum_{i=1}^{n}g_0(X_i)\{\cos(t\widehat{e}_i)-\cos(te_i)+\sin(t\widehat{e}_i)-\sin(te_i)\}}_{\widehat{U}^{(1)}_{n21}(t)}\\
    &+\underbrace{\frac{1}{n}\sum_{i=1}^{n}g_0(X_i)\{\cos(te_i)-\cos(t\varepsilon_i)+\sin(te_i)-\sin(t\varepsilon_i)\}}_{\widehat{U}^{(1)}_{n22}(t)}.
\end{align*}
For the term $\widehat{U}^{(1)}_{n21}(t)$, by the proof of $\widehat{U}_{n2}(t)$ in Theorem \ref{theorem:1}, we have
\begin{eqnarray*}
	\int_{\mathbb{R}} |\widehat{U}^{(1)}_{n21} (t)| ^2 \varphi(t)dt = O_p(1/n).
\end{eqnarray*}

For the terms $\widehat{U}^{(1)}_{n4}(t)$ and $\widehat{U}^{(1)}_{n5}(t)$, we have
\begin{align*}
	\widehat{U}^{(1)}_{n4}(t)=& \frac{1}{n}\sum_{i=1}^{n}g_0(X_i)\frac{1}{n}\sum_{i=1}^{n}\{\cos(t\widehat{e}_i)-\cos(t\varepsilon_i)\} \\
    =&\underbrace{\frac{1}{n}\sum_{i=1}^{n}g_0(X_i)\frac{1}{n}\sum_{i=1}^{n}\{\cos(t\widehat{e}_i)-\cos(te_i)\}}_{\widehat{U}^{(1)}_{n41}(t)}\\
    &+\underbrace{\frac{1}{n}\sum_{i=1}^{n}g_0(X_i)\frac{1}{n}\sum_{i=1}^{n}\{\cos(te_i)-\cos(t\varepsilon_i)\}}_{\widehat{U}^{(1)}_{n42}(t)}.
\end{align*}
and
\begin{align*}
	\widehat{U}^{(1)}_{n5}(t)=& \frac{1}{n}\sum_{i=1}^{n}g_0(X_i)\frac{1}{n}\sum_{i=1}^{n}\{\sin(t\widehat{e}_i)-\sin(t\varepsilon_i)\} \\
    =&\underbrace{\frac{1}{n}\sum_{i=1}^{n}g_0(X_i)\frac{1}{n}\sum_{i=1}^{n}\{\sin(t\widehat{e}_i)-\sin(te_i)\}}_{\widehat{U}^{(1)}_{n51}(t)}\\
    &+\underbrace{\frac{1}{n}\sum_{i=1}^{n}g_0(X_i)\frac{1}{n}\sum_{i=1}^{n}\{\sin(te_i)-\sin(t\varepsilon_i)\}}_{\widehat{U}^{(1)}_{n52}(t)}.
\end{align*}
For the terms $\widehat{U}^{(1)}_{n41}(t)$ and $\widehat{U}^{(1)}_{n51}(t)$, by the proof of $\widehat{U}_{n4}(t)$ and $\widehat{U}_{n5}(t)$ in Theorem \ref{theorem:1}, we have
\begin{eqnarray*}
	\int_{\mathbb{R}} |\widehat{U}^{(1)}_{n41} (t)| ^2 \varphi(t)dt = O_p(1/n),
\end{eqnarray*}
and
\begin{eqnarray*}
	\int_{\mathbb{R}} |\widehat{U}^{(1)}_{n51} (t)| ^2 \varphi(t)dt = O_p(1/n).
\end{eqnarray*}

Altogether, we obtain that
\begin{eqnarray*}
	\frac{1}{\sqrt{n}}\widehat{U}_{n}(t)=K_n^{(1)}(t)+R_{n}^{(1)}(t),
\end{eqnarray*}
where $R_{n1}^{(1)}(t)$ is a remainder satisfying $\int_{\mathbb{R}} |R_{n1}^{(1)}(t)|^2 \varphi(t)dt = o_p(1)$, and  
\begin{equation*}
    K_n^{(1)}(t):=\widehat{U}^{(1)}_{n22}(t)+\widehat{U}^{(1)}_{n42}(t)+\widehat{U}^{(1)}_{n52}(t).
\end{equation*}
Furthermore, by the law of large numbers and Assumption \ref{as:model}, we have
\begin{equation*}
    \sup_{t\in\mathbb{R}}|\widehat{U}^{(1)}_{n42}(t)|=o_p(1)\text{ and }\sup_{t\in\mathbb{R}}|\widehat{U}^{(1)}_{n52}(t)|=o_p(1),
\end{equation*}
where we have used the boundedness of the cosine and sine functions. Thus, 
\begin{equation*}
    \int_{\mathbb{R}}|\widehat{U}^{(1)}_{n42}(t)|^2\varphi(t)dt=o_p(1)\text{ and } \int_{\mathbb{R}}|\widehat{U}^{(1)}_{n52}(t)|^2\varphi(t)dt=o_p(1).
\end{equation*}
Let $\Delta_n(t):=\widehat{U}^{(1)}_{n22}(t)-K^{(1)}(t)$, where 
$$K^{(1)}(t):=\mathbb{E}[g_0(X_i)\{\cos(te_i)-\cos(t\varepsilon_i)+\sin(te_i)-\sin(t\varepsilon_i)\}].$$
Then, by Fubini's theorem,
$$\mathbb E\int_{\mathbb R}|\Delta_n(t)|^2\varphi(t)dt=\int_{\mathbb R}\mathbb E|\Delta_n(t)|^2\varphi(t)dt$$
Since $\{(X_i,e_i,\varepsilon_i)\}_{i=1}^n$ are i.i.d., we have
$$\mathbb E|\Delta_n(t)|^2=\frac1n\mathrm{Var}[g_0(X_i)\{\cos(te_i)-\cos(t\varepsilon_i)+\sin(te_i)-\sin(t\varepsilon_i)\}].$$
Hence
\begin{align*}
&\mathbb E\int_{\mathbb R}|\Delta_n(t)|^2\varphi(t)dt\\
=&
\frac1n\int_{\mathbb R}\mathrm{Var}[g_0(X_i)\{\cos(te_i)-\cos(t\varepsilon_i)+\sin(te_i)-\sin(t\varepsilon_i)\}]\varphi(t)dt.    
\end{align*}
By Assumption \ref{as:model} and \ref{as:weight}, we have that 
$$\mathrm{Var}[g_0(X_i)\{\cos(te_i)-\cos(t\varepsilon_i)+\sin(te_i)-\sin(t\varepsilon_i)\}]\leq C,$$
for a positive constant. Therefore,
$$
\mathbb E\int_{\mathbb R}|\Delta_n(t)|^2\varphi(t)dt\to 0.
$$
It follows from Markov's inequality that
$$
\int_{\mathbb R}|\widehat{U}^{(1)}_{n22}(t)-K^{(1)}(t)|^2\varphi(t)dt=o_p(1).
$$
Next, note that
\begin{align*}
&\bigg|\int_{\mathbb R}|K_n^{(1)}(t)|^2\varphi(t)dt-\int_{\mathbb R}|K^{(1)}(t)|^2\varphi(t)dt\bigg|\\
\leq&
\int_{\mathbb R}|K_n^{(1)}(t)-K^{(1)}(t)|\times|K_n^{(1)}(t)+K^{(1)}(t)|\varphi(t)dt\\
\leq&\sqrt{\int_{\mathbb R}|K_n^{(1)}(t)-K^{(1)}(t)|^2\varphi(t)dt}\times\sqrt{\int_{\mathbb R}K_n^{(1)}(t)+K^{(1)}(t)|^2\varphi(t)dt},
\end{align*}
where we use the Cauchy--Schwarz inequality. By the previous proof, we have
$$\mathbb{E}\bigg|\int_{\mathbb R}|K_n^{(1)}(t)|^2\varphi(t)dt-\int_{\mathbb R}|K^{(1)}(t)|^2\varphi(t)dt\bigg|=o(1),$$
and, hence,
$$
\int_{\mathbb R}|K_n^{(1)}(t)|^2\varphi(t)dt-\int_{\mathbb R}|K^{(1)}(t)|^2\varphi(t)dt=o_p(1).
$$
Thus, we have
\begin{equation*}
    \frac{1}{n}\int_{\mathbb{R}}|\widehat{U}_n(t)|^2\varphi(t)dt\stackrel{p}{\to}\int_{\mathbb{R}}|K^{(1)}(t)|^2\varphi(t)dt=:C_1>0,
\end{equation*}
which completes the proof under the global alternative.

\textbf{Case 2}: We then consider the local alternative, corresponding to $\alpha\in(0,1/2)$ and hence $r_n=n^{-\alpha}$. Similar to the arguments for $\widehat{U}_n$ in the proof of Step 1, $\widehat{U}_n$ can be decomposed into five parts:
\begin{align*}
\frac{1}{r_n\sqrt{n}} \widehat{U}_n(t)
	=& \underbrace{\frac{1}{r_nn}\sum_{i=1}^{n}g_0(X_i)\{\cos(t\varepsilon_i)+\sin(t\varepsilon_i)\}}_{\widehat{U}^{(2)}_{n1}(t) }\\
    &+ \underbrace{\frac{1}{r_nn}\sum_{i=1}^{n}g_0(X_i)\{\cos(t\widehat{e}_i)-\cos(t\varepsilon_i)+\sin(t\widehat{e}_i)-\sin(t\varepsilon_i)\}}_{\widehat{U}^{(2)}_{n2}(t) } \\
	& - \underbrace{\frac{1}{r_nn}\sum_{i=1}^{n} g_0(X_i) \frac{1}{n} \sum_{i=1}^{n} \{\cos(t \varepsilon_i)+\sin(t \varepsilon_i)\}}_{\widehat{U}^{(2)}_{n3}(t)}\\
    &-\underbrace{\frac{1}{r_nn}\sum_{i=1}^{n} g_0(X_i) \frac{1}{n} \sum_{i=1}^{n} \{\cos(t\widehat{e}_i)-\cos(t \varepsilon_i)\}}_{\widehat{U}^{(2)}_{n4}(t)}\\
	&- \underbrace{\frac{1}{r_nn}\sum_{i=1}^{n} g_0(X_i) \frac{1}{n}\sum_{i=1}^{n}\{
	\sin(t\widehat{e}_i)- \sin(t \varepsilon_i)\}}_{\widehat{U}^{(2)}_{n5}(t)}.
\end{align*}

First, we consider the term $\widehat{U}^{(2)}_{n1} (t)$. By an argument analogous to that used for $\widehat{U}_{n1} (t)$ in Theorem \ref{theorem:1}, we have
\begin{eqnarray*}
	\int_{\mathbb{R}} |\widehat{U}^{(2)}_{n1} (t)| ^2 \varphi(t)dt = O_p\bigg(\frac{1}{r_n^2n}\bigg)=o_p(1).
\end{eqnarray*}
Similarly, we can obtain
\begin{eqnarray*}
	\int_{\mathbb{R}} |\widehat{U}^{(2)}_{n3} (t)| ^2 \varphi(t)dt = O_p\bigg(\frac{1}{r_n^2n}\bigg)=o_p(1).
\end{eqnarray*}

It remains to consider the remaining terms. For the second term $\widehat{U}^{(2)}_{n2}(t)$, we have
\begin{align*}
	\widehat{U}^{(2)}_{n2}(t)=& \frac{1}{r_nn}\sum_{i=1}^{n}g_0(X_i)\{\cos(t\widehat{e}_i)-\cos(t\varepsilon_i)+\sin(t\widehat{e}_i)-\sin(t\varepsilon_i)\} \\
    =&\underbrace{\frac{1}{r_nn}\sum_{i=1}^{n}g_0(X_i)\{\cos(t\widehat{e}_i)-\cos(te_i)+\sin(t\widehat{e}_i)-\sin(te_i)\}}_{\widehat{U}^{(2)}_{n21}(t)}\\
    &+\underbrace{\frac{1}{r_nn}\sum_{i=1}^{n}g_0(X_i)\{\cos(te_i)-\cos(t\varepsilon_i)+\sin(te_i)-\sin(t\varepsilon_i)\}}_{\widehat{U}^{(2)}_{n22}(t)}.
\end{align*}
For the term $\widehat{U}^{(2)}_{n21}(t)$, by the proof of $\widehat{U}_{n2}(t)$ in Theorem \ref{theorem:1}, we have
\begin{eqnarray*}
	\int_{\mathbb{R}} |\widehat{U}^{(2)}_{n21} (t)| ^2 \varphi(t)dt = O_p\bigg(\frac{1}{r_n^2n}\bigg)=o_p(1).
\end{eqnarray*}
Note that by definition, we have
\begin{equation*}
    e_i=Y_i- m(X_i,\widetilde\beta_0)=m(X_i, \widetilde\beta_0) + r_nS(X_i) + \varepsilon_i- m(X_i,\widetilde\beta_0)=r_nS(X_i) + \varepsilon_i.
\end{equation*}
Thus, we have
\begin{align*}
    &\widehat{U}^{(2)}_{n22}(t)\\
    =&\frac{1}{r_nn}\sum_{i=1}^{n}g_0(X_i)[\cos\{tr_nS(X_i) + t\varepsilon_i\}-\cos(t\varepsilon_i)+\sin\{tr_nS(X_i) + t\varepsilon_i\}-\sin(t\varepsilon_i)]\\
    =&\underbrace{\frac{tr_n}{r_nn}\sum_{i=1}^{n}g_0(X_i)S(X_i)\{\cos(t\varepsilon_i)-\sin(t\varepsilon_i)\}}_{\widehat{U}^{(2)}_{n221}}\\
&-\underbrace{\frac{t^2r_n^2}{2r_nn}\sum_{i=1}^{n}g_0(X_i)S(X_i)^2\{\cos(t\widetilde\varepsilon_i)+\sin(t\widetilde\varepsilon_i)\}}_{\widehat{U}^{(2)}_{n222}},
\end{align*}
where we use the Taylor's expansion and $\widetilde{\varepsilon}_i$ is between $\varepsilon_i$ and $r_nS(X_i) + \varepsilon_i$. For $\widehat{U}^{(2)}_{n222}(t)$, using the boundedness of the cosine and sine functions, we have
\begin{align*}
|\widehat{U}^{(2)}_{n222}(t)|&\leq\frac{|t|^2r_n}{2n}\sum_{i=1}^{n}2|g_0(X_i)S^2(X_i)| \\
&=|t|^2r_n\frac{1}{n}\sum_{i=1}^{n}|g_0(X_i)S^2(X_i)|,
\end{align*}
where we use the triangle inequality. Therefore,
\begin{align*}
\int_{\mathbb{R}}|\widehat{U}^{(2)}_{n222}(t)|^2\phi(t)\,dt
&\le
r_n^2
\Big(\frac{1}{n}\sum_{i=1}^{n}|g_0(X_i)S^2(X_i)|\Big)^2
\int_{\mathbb{R}} t^4\phi(t)\,dt \\
&= O_p(r_n^2)
= o_p(1),
\end{align*}
where we use the law of large numbers and the Assumptions \ref{as:model} and \ref{as:weight}.

For the terms $\widehat{U}^{(2)}_{n4}(t)$ and $\widehat{U}^{(2)}_{n5}(t)$, we have
\begin{align*}
	\widehat{U}^{(2)}_{n4}(t)=& \frac{1}{n}\sum_{i=1}^{n}g_0(X_i)\frac{1}{nr_n}\sum_{i=1}^{n}\{\cos(t\widehat{e}_i)-\cos(t\varepsilon_i)\} \\
    =&\underbrace{\frac{1}{n}\sum_{i=1}^{n}g_0(X_i)\frac{1}{r_nn}\sum_{i=1}^{n}\{\cos(t\widehat{e}_i)-\cos(te_i)\}}_{\widehat{U}^{(2)}_{n41}(t)}\\
    &+\underbrace{\frac{1}{n}\sum_{i=1}^{n}g_0(X_i)\frac{1}{r_nn}\sum_{i=1}^{n}\{\cos(te_i)-\cos(t\varepsilon_i)\}}_{\widehat{U}^{(2)}_{n42}(t)}.
\end{align*}
and
\begin{align*}
	\widehat{U}^{(2)}_{n5}(t)=& \frac{1}{n}\sum_{i=1}^{n}g_0(X_i)\frac{1}{r_nn}\sum_{i=1}^{n}\{\sin(t\widehat{e}_i)-\sin(t\varepsilon_i)\} \\
    =&\underbrace{\frac{1}{n}\sum_{i=1}^{n}g_0(X_i)\frac{1}{r_nn}\sum_{i=1}^{n}\{\sin(t\widehat{e}_i)-\sin(te_i)\}}_{\widehat{U}^{(2)}_{n51}(t)}\\
    &+\underbrace{\frac{1}{n}\sum_{i=1}^{n}g_0(X_i)\frac{1}{r_nn}\sum_{i=1}^{n}\{\sin(te_i)-\sin(t\varepsilon_i)\}}_{\widehat{U}^{(2)}_{n52}(t)}.
\end{align*}
For the terms $\widehat{U}^{(2)}_{n41}(t)$ and $\widehat{U}^{(2)}_{n51}(t)$, by the proof of $\widehat{U}_{n4}(t)$ and $\widehat{U}_{n5}(t)$ in Theorem \ref{theorem:1}, we have
\begin{eqnarray*}
	\int_{\mathbb{R}} |\widehat{U}^{(2)}_{n41} (t)| ^2 \varphi(t)dt = O_p\bigg(\frac{1}{r_nn}\bigg)=o_p(1),
\end{eqnarray*}
and
\begin{eqnarray*}
	\int_{\mathbb{R}} |\widehat{U}^{(2)}_{n51} (t)| ^2 \varphi(t)dt = O_p\bigg(\frac{1}{r_nn}\bigg)=o_p(1).
\end{eqnarray*}
For the term $\widehat{U}^{(2)}_{n52}(t)$, we have
\begin{align*}
    \widehat{U}^{(2)}_{n52}(t)=&\frac{1}{n}\sum_{i=1}^{n}g_0(X_i)\frac{1}{r_nn}\sum_{i=1}^{n}[\sin\{tr_nS(X_i)+t\varepsilon_i\}-\sin(t\varepsilon_i)]\\
    =&\underbrace{\frac{1}{n}\sum_{i=1}^{n}g_0(X_i)\frac{1}{r_nn}\sum_{i=1}^{n}tr_nS(X_i)\cos(t\varepsilon_i)}_{\widehat{U}^{(2)}_{n521}(t)}\\
    &-\underbrace{\frac{1}{n}\sum_{i=1}^{n}g_0(X_i)\frac{1}{2r_nn}\sum_{i=1}^{n}\{tr_nS(X_i)\}^2\sin(t\widetilde\varepsilon_i)}_{\widehat{U}^{(2)}_{n522}(t)},
\end{align*}
where we use the Taylor's expansion and $\widetilde{\varepsilon}_i$ is between $\varepsilon_i$ and $r_nS(X_i) + \varepsilon_i$. For $\widehat{U}^{(2)}_{n522}(t)$, using the boundedness of the sine functions, we have
\begin{align*}
|\widehat{U}^{(2)}_{n522}(t)|&\leq\frac{|t|^2r_n}{2n}\sum_{i=1}^{n}|S^2(X_i)|\times\bigg|\frac{1}{n}\sum_{i=1}^{n}g_0(X_i)\bigg|,
\end{align*}
where we use the triangle inequality. Therefore,
\begin{align*}
\int_{\mathbb{R}}|\widehat{U}^{(2)}_{n522}(t)|^2\phi(t)\,dt
&\le
r_n^2
\Big(\frac{1}{n}\sum_{i=1}^{n}|g_0(X_i)|\times\frac{1}{n}\sum_{i=1}^{n}|S^2(X_i)|\Big)^2
\int_{\mathbb{R}} t^4\phi(t)\,dt \\
&= O_p(r_n^2)
= o_p(1),
\end{align*}
where we use the law of large numbers and the Assumptions \ref{as:model} and \ref{as:weight}. Similarly, for the term $\widehat{U}^{(2)}_{n42}(t)$, we have
\begin{align*}
    \widehat{U}^{(2)}_{n42}(t)=&\frac{1}{n}\sum_{i=1}^{n}g_0(X_i)\frac{1}{r_nn}\sum_{i=1}^{n}[\cos\{tr_nS(X_i)+t\varepsilon_i\}-\cos(t\varepsilon_i)]\\
    =&\underbrace{-\frac{1}{n}\sum_{i=1}^{n}g_0(X_i)\frac{1}{r_nn}\sum_{i=1}^{n}tr_nS(X_i)\sin(t\varepsilon_i)}_{\widehat{U}^{(2)}_{n421}(t)}\\
    &-\underbrace{\frac{1}{n}\sum_{i=1}^{n}g_0(X_i)\frac{1}{2r_nn}\sum_{i=1}^{n}\{tr_nS(X_i)\}^2\cos(t\widetilde\varepsilon_i)}_{\widehat{U}^{(2)}_{n422}(t)},
\end{align*}
where
\begin{equation*}
    \int_{\mathbb{R}}|\widehat{U}^{(2)}_{n422}(t)|^2\phi(t)dt= o_p(1).
\end{equation*}
Furthermore, by the law of large numbers and Assumption \ref{as:model}, we have
\begin{align*}
    &\max\bigg\{\int_{\mathbb{R}}|\widehat{U}^{(2)}_{n421}(t)|^2\phi(t)dt,\int_{\mathbb{R}}|\widehat{U}^{(2)}_{n521}(t)|^2\phi(t)dt\bigg\}\\
    \leq& \bigg|\frac{1}{n}\sum_{i=1}^{n}g_0(X_i)\bigg|^2\times\bigg|\frac{1}{n}\sum_{i=1}^{n}|S_0(X_i)|\bigg|^2\times\int_{\mathbb{R}}t^2\varphi(t)dt=o_p(1),
\end{align*}
where we use the law of large numbers and the Assumptions \ref{as:model} and \ref{as:weight}

It remains to consider the term,
\begin{align*}
    \widehat{U}^{(2)}_{n221}=\frac{t}{n}\sum_{i=1}^{n}g_0(X_i)S(X_i)\{\cos(t\varepsilon_i)-\sin(t\varepsilon_i)\}.
\end{align*}
We define 
$$K^{(2)}(t):=\mathbb{E}[g_0(X_i)S(X_i)\{\cos(t\varepsilon_i)-\sin(t\varepsilon_i)\}].$$
For completeness, the required law of large numbers is
\begin{align*}
 &\mathbb E\big\|\widehat U_{n221}^{(2)}-tK^{(2)}\big\|_\varphi^2=\frac1n\int_{\mathbb R}t^2\operatorname{Var}\!\left[
 g_0(X)S(X)\{\cos(t\varepsilon)-\sin(t\varepsilon)\}\right]
 \varphi(t)\,dt\\
 &\le \frac{2\,\mathbb E\{g_0(X)^2S(X)^2\}}{n}
       \int_{\mathbb R}t^2\varphi(t)\,dt=o(1).
\end{align*}
Together with the preceding remainder bounds, this yields
$$
 \bigg\|\frac{\widehat U_n}{r_n\sqrt n}-tK^{(2)}\bigg\|_\varphi=o_p(1).
$$
Thus, by the similar arguments of Case 1, we have
\begin{equation*}
    \frac{1}{r_n^2n}\int_{\mathbb{R}}|\widehat{U}_n(t)|^2\varphi(t)dt\stackrel{p}{\to}\int_{\mathbb{R}}|\widehat{U}_{n221}^{(2)}(t)|^2\varphi(t)dt\stackrel{p}{\to}\int_{\mathbb{R}}|K^{(2)}(t)|^2t^2\varphi(t)dt=:C_2>0,
\end{equation*}
which completes the proof under the local alternative.

\textbf{Case 3}: We then consider the local alternative, corresponding to $\alpha=1/2$ and hence $r_n=n^{-1/2}$. Similar to the arguments for $\widehat{U}_n$ in the proof of Step 1, $\widehat{U}_n$ can be decomposed into five parts:
\begin{align*}
\widehat{U}_n(t)
	=& \underbrace{\frac{1}{\sqrt{n}}\sum_{i=1}^{n}g_0(X_i)\{\cos(t\varepsilon_i)+\sin(t\varepsilon_i)\}}_{\widehat{U}^{(3)}_{n1}(t) }\\
    &+ \underbrace{\frac{1}{\sqrt{n}}\sum_{i=1}^{n}g_0(X_i)\{\cos(t\widehat{\varepsilon}_i)-\cos(t\varepsilon_i)+\sin(t\widehat{\varepsilon}_i)-\sin(t\varepsilon_i)\}}_{\widehat{U}^{(3)}_{n2}(t) } \\
	& - \underbrace{\frac{1}{\sqrt{n}}\sum_{i=1}^{n} g_0(X_i) \frac{1}{n} \sum_{i=1}^{n} \{\cos(t \varepsilon_i)+\sin(t \varepsilon_i)\}}_{\widehat{U}^{(3)}_{n3}(t)}\\
    &-\underbrace{\frac{1}{\sqrt{n}}\sum_{i=1}^{n} g_0(X_i) \frac{1}{n} \sum_{i=1}^{n} \{\cos(t\widehat{\varepsilon}_i)-\cos(t \varepsilon_i)\}}_{\widehat{U}^{(3)}_{n4}(t)}\\
	&- \underbrace{\frac{1}{\sqrt{n}}\sum_{i=1}^{n} g_0(X_i) \frac{1}{n}\sum_{i=1}^{n}\{
	\sin(t\widehat{\varepsilon}_i)- \sin(t \varepsilon_i)\}}_{\widehat{U}^{(3)}_{n5}(t)}\\
    &+\underbrace{\frac{1}{\sqrt{n}}\sum_{i=1}^{n}g_0(X_i)\{\cos(t\widehat{e}_i)-\cos(t\widehat{\varepsilon}_i)+\sin(t\widehat{e}_i)-\sin(t\widehat{\varepsilon}_i)\}}_{\widehat{U}^{(3)}_{n6}(t) }\\
    &-\underbrace{\frac{1}{\sqrt{n}}\sum_{i=1}^{n} g_0(X_i) \frac{1}{n} \sum_{i=1}^{n} \{\cos(t\widehat{e}_i)-\cos(t \widehat{\varepsilon}_i)\}}_{\widehat{U}^{(3)}_{n7}(t)}\\
	&- \underbrace{\frac{1}{\sqrt{n}}\sum_{i=1}^{n} g_0(X_i) \frac{1}{n}\sum_{i=1}^{n}\{
	\sin(t\widehat{e}_i)- \sin(t \widehat{\varepsilon}_i)\}}_{\widehat{U}^{(3)}_{n8}(t)},
\end{align*}
where we introduce the notation
$$\widehat{\varepsilon}_i:=\varepsilon_i+m(X_i,\widetilde{\beta}_0)-m(X_i,\widehat{\beta}_n).$$
It can be seen that $\widehat{\varepsilon}_i$ serves as a pseudo-estimated residual. By the proof of Theorem \ref{theorem:1}, we have
\begin{equation*}
    \widehat{U}^{(3)}_{n1}(t)+\widehat{U}^{(3)}_{n2}(t)+\widehat{U}^{(3)}_{n3}(t)+\widehat{U}^{(3)}_{n4}(t)+\widehat{U}^{(3)}_{n5}(t)=U_n(t)+R_n(t),
\end{equation*}
where the remainder $R_n(t)$ satisfies $\int_{\mathbb{R}}|R_n(t)|^2\varphi(t)dt=o_p(1)$. 

It remains to consider the terms $\widehat{U}^{(3)}_{n6}(t),\widehat{U}^{(3)}_{n7}(t),\widehat{U}^{(3)}_{n8}(t)$. For the second term $\widehat{U}^{(3)}_{n6}(t)$, we have
\begin{align*}
	\widehat{U}^{(3)}_{n6}(t)=& \frac{1}{\sqrt{n}}\sum_{i=1}^{n}g_0(X_i)\{\cos(t\widehat{e}_i)-\cos(t\widehat{\varepsilon}_i)+\sin(t\widehat{e}_i)-\sin(t\widehat{\varepsilon}_i)\} \\
    =&\underbrace{\frac{1}{\sqrt n}\sum_{i=1}^{n}g_0(X_i)\{-\sin(t\varepsilon_i)t(\widehat{e}_i-\widehat{\varepsilon}_i)+\cos(t\varepsilon_i)t(\widehat{e}_i-\widehat{\varepsilon}_i)\}}_{\widehat{U}^{(3)}_{n61}(t)}\\
    &-\underbrace{\frac{1}{2\sqrt{n}}\sum_{i=1}^{n}g_0(X_i)\{\cos(t\widetilde{e}_i)(\widehat{e}_i-\varepsilon_i)^2-\cos(t\widetilde{\varepsilon}_i)(\widehat{\varepsilon}_i-\varepsilon_i)^2\}}_{\widehat{U}^{(3)}_{n62}(t)}\\
    &-\underbrace{\frac{1}{2\sqrt{n}}\sum_{i=1}^{n}g_0(X_i)\{\sin(t\widetilde{e}_i)(\widehat{e}_i-\varepsilon_i)^2-\sin(t\widetilde{\varepsilon}_i)(\widehat{\varepsilon}_i-\varepsilon_i)^2\}}_{\widehat{U}^{(3)}_{n63}(t)},
\end{align*}
where we use the second order Taylor expansion, $\widetilde{e}_i$ is in the line segment between $\widehat{e}_i$ and $\varepsilon_i$, and $\widetilde{\varepsilon}_i$ is in the line segment between $\widehat{\varepsilon}_i$ and $\varepsilon_i$. For the term $\widehat{U}^{(3)}_{n61}(t)$, by the fact
\begin{equation*}
    \widehat{e}_i-\widehat{\varepsilon}_i=r_n S(X_i),
\end{equation*}
we have
\begin{align*}
	\widehat{U}^{(3)}_{n61} (t)=&\frac{1}{n}\sum_{i=1}^{n}g_0(X_i)S(X_i)\{-\sin(t\varepsilon) +\cos(t\varepsilon)\}t\\
    &\stackrel{p}{\to}\mathbb{E}[g_0(X)S(X)\{-\sin(t\varepsilon) +\cos(t\varepsilon)]t,
\end{align*}
where we use the Assumption \ref{as:model} and the law of large numbers. For the term $\widehat{U}^{(3)}_{n62}(t)$, we have
\begin{align*}
    &\frac{1}{2\sqrt{n}}\sum_{i=1}^{n}g_0(X_i)\cos(t\widetilde{e}_i)(\widehat{e}_i-\varepsilon_i)^2\\
    =&\frac{1}{2\sqrt{n}}\sum_{i=1}^{n}g_0(X_i)\cos(t\widetilde{e}_i)\{r_nS(X_i)+m(X_i,\widetilde{\beta}_0)-m(X_i,\widehat{\beta}_n)\}^2\\
    =&\underbrace{\frac{1}{2\sqrt{n}}\sum_{i=1}^{n}g_0(X_i)\cos(t\widetilde{e}_i)r_n^2S(X_i)^2}_{\widehat{U}^{(3)}_{n621}(t)}\\
    &+\underbrace{\frac{1}{\sqrt{n}}\sum_{i=1}^{n}g_0(X_i)\cos(t\widetilde{e}_i)r_nS(X_i)\{m(X_i,\widetilde{\beta}_0)-m(X_i,\widehat{\beta}_n)\}}_{\widehat{U}^{(3)}_{n622}(t)}\\
    &+\underbrace{\frac{1}{2\sqrt{n}}\sum_{i=1}^{n}g_0(X_i)\cos(t\widetilde{e}_i)\{m(X_i,\widetilde{\beta}_0)-m(X_i,\widehat{\beta}_n)\}^2}_{\widehat{U}^{(3)}_{n623}(t)}.
\end{align*}
For the term $\widehat{U}^{(3)}_{n621}(t)$, we have
\begin{equation*}
    |\widehat{U}^{(3)}_{n621}(t)|\leq\frac{1}{\sqrt{n}}\frac{1}{n}\sum_{i=1}^{n}|g_0(X_i)S(X_i)|=o_p(1),
\end{equation*}
where we use the Assumption \ref{as:model} and the law of large numbers. For the term $\widehat{U}^{(3)}_{n623}(t)$, by the similar proof of $\widehat{U}_{22}(t)$ in Theorem \ref{theorem:1}, we have $\int_{\mathbb{R}}|\widehat{U}^{(3)}_{n623}(t)|^2\varphi(t)dt=o_p(1)$. For the term, $\widehat{U}^{(3)}_{n622}(t)$, we have
\begin{align*}
    \widehat{U}^{(3)}_{n622}(t)=&\underbrace{\frac{1}{n}\sum_{i=1}^{n}g_0(X_i)\cos(t\varepsilon_i)S(X_i)\{m(X_i,\widetilde{\beta}_0)-m(X_i,\widehat{\beta}_n)\}}{\widehat{U}^{(3)}_{n6221}(t)}\\
    &+\underbrace{\frac{1}{n}\sum_{i=1}^{n}g_0(X_i)\{\cos(t\widetilde{e}_i)-\cos(t\varepsilon_i)\}S(X_i)\{m(X_i,\widetilde{\beta}_0)-m(X_i,\widehat{\beta}_n)\}}_{\widehat{U}^{(3)}_{n6222}(t)}.
\end{align*}
By the similar proof of $\widehat{U}_{n21}(t)$ of Theorem \ref{theorem:1}, we have $\int_{\mathbb{R}}|\widehat{U}^{(3)}_{n6221}(t)|^2\varphi(t)dt=o_p(1)$. For the term $\widehat{U}^{(3)}_{n6222}(t)$, we have
\begin{align*}
    |\widehat{U}^{(3)}_{n6222}(t)|\leq&\frac{1}{n}\sum_{i=1}^{n}|g_0(X_i)S(X_i)|\cdot|m(X_i,\widetilde{\beta}_0)-m(X_i,\widehat{\beta}_n)|\\
    &\times|r_nS(X_i)+m(X_i,\widetilde{\beta}_0)-m(X_i,\widehat{\beta}_n)|t\\
    \leq&\frac{1}{n^{1.5}}\sum_{i=1}^{n}|g_0(X_i)S(X_i)^2|\cdot|m(X_i,\widetilde{\beta}_0)-m(X_i,\widehat{\beta}_n)|\\
    &+\frac{1}{n}\sum_{i=1}^{n}|g_0(X_i)S(X_i)|\cdot|m(X_i,\widetilde{\beta}_0)-m(X_i,\widehat{\beta}_n)|^2\\
    \leq&\frac{1}{n^{1.5}}\sum_{i=1}^{n}|g_0(X_i)S(X_i)^2|\cdot 2|F(X_i)|\\
    &+\frac{1}{n}\sum_{i=1}^{n}|g_0(X_i)S(X_i)|\cdot pF(X_i)^2\|\widetilde{\beta}_0-\widehat{\beta}_n\|_2^2=o_p(1)\cdot t,
\end{align*}
where we use the triangle inequality in the second inequality, we use Assumption \ref{as:model} in the third inequality, and we use the Assumption \ref{as:model} and Lemma \ref{lem:beta} in the last equality. Thus, by using similar arguments, we have
\begin{equation*}
    \int_{\mathbb{R}}|\widehat{U}^{(3)}_{n62} (t)+\widehat{U}^{(3)}_{n63} (t)|^2\varphi(t)dt=o_p(1).
\end{equation*}
For clarity, both quadratic terms in the preceding second-order Taylor expansion carry a factor $t^2$. A direct expansion records these factors and upgrades the pointwise limit to $L^2(\varphi)$. Since we have
$$\widehat e_i=\widehat\varepsilon_i+n^{-1/2}S(X_i),$$
and
$$
 \widehat U_{n6}^{(3)}(t)=\frac{t}{n}\sum_{i=1}^n g_0(X_i)S(X_i)\{\cos(t\widehat\varepsilon_i)-\sin(t\widehat\varepsilon_i)\} +\rho_{n6}(t),
$$
where Taylor's theorem gives $\|\rho_{n6}\|_\varphi=o_p(1)$. Writing \(d_{i,n}=\widehat\varepsilon_i-\varepsilon_i\), the Lipschitz bounds for sine and cosine, the envelope condition, and Lemma~\ref{lem:beta} yield
$$
 \bigg\|\frac{t}{n}\sum_{i=1}^n g_0(X_i)S(X_i)\bigg[\{\cos(t\widehat\varepsilon_i)\sin(t\widehat\varepsilon_i)\}-\{\cos(t\varepsilon_i)-\sin(t\varepsilon_i)\}\bigg]\bigg\|_\varphi=o_p(1).
$$
The preceding $L^2(\varphi)$ law of large numbers therefore implies
$$
 \|\widehat U_{n6}^{(3)}-tK^{(2)}\|_\varphi=o_p(1).
$$

It remains to consider the terms $\widehat{U}^{(3)}_{n7}(t)$ and $\widehat{U}^{(3)}_{n8}(t)$. For the term $\widehat{U}^{(3)}_{n7}(t)$, we have
\begin{align*}
\widehat{U}^{(3)}_{n7}(t)&=\underbrace{\frac{-1}{\sqrt{n}}\sum_{i=1}^{n} g_0(X_i) \frac{1}{n} \sum_{i=1}^{n}\sin(t\varepsilon_i)(t\widehat{e}_i-t \widehat{\varepsilon}_i)}_{\widehat{U}^{(3)}_{n71}(t)} \\
&-\underbrace{\frac{1}{2\sqrt{n}}\sum_{i=1}^{n} g_0(X_i) \frac{1}{n} \sum_{i=1}^{n}\{\cos(t\widetilde{e}_i)(t\widehat{e}_i-t\varepsilon_i)^2-\cos(t \widetilde{\varepsilon}_i)(t\widehat{\varepsilon}_i-t\varepsilon_i)^2\}}_{\widehat{U}^{(3)}_{n72}(t)}
\end{align*}
For the term $\widehat{U}^{(3)}_{n71}(t)$, we have
\begin{equation*}
    \widehat{U}^{(3)}_{n71}(t)=-\frac{1}{\sqrt{n}}\bigg\{\frac{1}{\sqrt{n}}\sum_{i=1}^{n} g_0(X_i)\bigg\} \times\bigg\{\frac{1}{n} \sum_{i=1}^{n}\sin(t\varepsilon_i)S(X_i)\bigg\}t=o_p(1)\times t,
\end{equation*}
where we use the fact that $\sum_{i=1}^{n} g_0(X_i)/\sqrt{n}=O_p(1)$ and the law of large number. For the term $\widehat{U}^{(3)}_{n72}(t)$, it remains to consider the first part,
\begin{align*}
    \frac{1}{2\sqrt{n}}\sum_{i=1}^{n} g_0(X_i) \frac{1}{n} \sum_{i=1}^{n}\cos(t\widetilde{e}_i)(t\widehat{e}_i-t\varepsilon_i)^2.
\end{align*}
By the proof of the term $\widehat{U}^{(3)}_{n62}(t)$, we have
\begin{equation*}
    \frac{1}{\sqrt{n}}\sum_{i=1}^{n}\cos(t\widetilde{e}_i)(t\widehat{e}_i-t\varepsilon_i)^2=o_p(1)\times t^2,
\end{equation*}
and thus we have
\begin{equation*}
    \int_{\mathbb{R}}|\widehat{U}^{(3)}_{n71} (t)+\widehat{U}^{(3)}_{n72} (t)|^2\varphi(t)dt=o_p(1).
\end{equation*}
It remains to consider the term $\widehat U_{n8}^{(3)}$. By using similar arguments to $\widehat U_{n7}^{(3)}$, we have 
$$\int_{\mathbb{R}}|\widehat U_{n8}^{(3)}|^2\varphi(t)dt=o_p(1).$$  
Thus, we have
$$
 \|\widehat U_{n7}^{(3)}+\widehat U_{n8}^{(3)}\|_\varphi=o_p(1).
$$
Combining all previous bounds and the result for $\widehat U_{n6}^{(3)}$ gives
$$
 \|\widehat U_n-U_n-tK^{(2)}\|_\varphi=o_p(1).
$$
Since $U_n\stackrel{d}{\to} U_\infty$ in $L^2(\varphi)$, Slutsky's theorem and the continuous mapping theorem yield
$$
 \int_{\mathbb R}|\widehat U_n(t)|^2\varphi(t)\,dt\stackrel{d}{\to}\int_{\mathbb R}|U_\infty(t)+tK^{(2)}(t)|^2\varphi(t)\,dt.
$$
Hence, we complete the proof of Theorem \ref{thm:alt}.
\end{proof}

We then present the proof of Theoretical Results for the Bootstrap.

\begin{proof}[Proof of Theorem \ref{thm:bootstrap}]

Let $\mathbb{P}^*$ denote the probability measure induced by the smoothed residual bootstrap conditional on the original sample $\{(X_i, Y_i)\}_{i=1}^{n}$ and let $\mathbb{E}^*$ denote the expectation under $\mathbb{P}^*$. Recall that it remains to consider the random variable
$$ \widehat{U}^*_n(t)
=\frac{1}{\sqrt{n}}\sum_{i=1}^{n}(g(X_i)-\bar{g})\{\cos(t\widehat{e}^*_i)+\sin(t\widehat{e}^*_i)\} ,$$ where $\bar{g}=n^{-1}\sum_{i=1}^{n}g(X_i)$, $ t \in \mathbb{R}$ and
$\widehat{e}^*_i=Y^*_i-m(X_i,\widehat{\beta}^*_n)=m(X_i,\widehat{\beta}_n)-m(X_i,\widehat{\beta}^*_n)+\varepsilon^*_i.$ Let
$$
 \mathcal F_n=\sigma\{(X_i,Y_i):1\le i\le n\},\qquad
 \|f\|_\varphi^2=\int_{\mathbb R}|f(t)|^2\varphi(t)\,dt,
$$
where $\mathcal F_n$ is the $\sigma$-field generated by the data. All bootstrap probability statements below are conditional on $\mathcal F_n$. Thus, the notion $o_{P^*}(1)$ denotes convergence to zero in $\mathbb P^*$-probability, in outer probability, and conditional weak convergence is understood in bounded-Lipschitz distance in probability. All limiting integrals below are with respect to $\varphi(t)\,dt$. To avoid ambiguity about the factor $t$, we use the notation
$$
 W_0(t,\beta)=\mathbb E[g_0(X)\{\sin(te)-\cos(te)\}\dot m(X,\beta)],\qquad B(t,\beta)=tW_0(t,\beta).
$$
Whenever the covariance is written with explicit factors $s$ and $t$, the symbols $W(s)$ and $W(t)$ refer to $W_0(s,\beta)$ and $W_0(t,\beta)$.

The centering terms must be handled jointly. We introduce the following notations:
$$
 h_t(u)=\cos(tu)+\sin(tu),\qquad
 a_{n,i}=g(X_i)-\bar g,\qquad
 \mu_n^*(t)=\mathbb E^*\{h_t(\varepsilon_1^*)\}.
$$
Since $\sum_i a_{n,i}=0$, the cancellation is exact. Similar to the arguments in Theorem \ref{theorem:1}, we decompose $\widehat{U}^*_n(t)$ as follows,
\begin{align*}
	\widehat{U}^*_n(t)=&\frac{1}{\sqrt{n}}\sum_{i=1}^{n}a_{n,i}\{\cos(t\widehat{e}_i^*)+\sin(t\widehat{e}_i^*)\}\\
    =&\underbrace{\frac{1}{\sqrt{n}}\sum_{i=1}^{n}a_{n,i}\{\cos(t\varepsilon_i^*)+\sin(t\varepsilon_i^*)\}}_{\widehat{U}_{n1}^*(t)}\\
    &+\underbrace{\frac{1}{\sqrt{n}}\sum_{i=1}^{n}a_{n,i}\{\cos(t\widehat{e}_i^*)-\cos(t\varepsilon_i^*)+\sin(t\widehat{e}_i^*)-\sin(t\varepsilon_i^*)\}}_{\widehat{U}_{n2}^*(t)}
\end{align*}
Let
$$
 \dot m_{n,i}=\dot m(X_i,\widehat\beta_n),\qquad
 \widehat A_n=\frac1n\sum_{i=1}^n
 \dot m_{n,i}\dot m_{n,i}^{\top}.
$$
A conditional Taylor expansion of the bootstrap least-squares equation gives
$$
 \sqrt n(\widehat\beta_n^*-\widehat\beta_n)=\widehat A_n^{-1}\frac1{\sqrt n}\sum_{i=1}^n\dot m_{n,i}\varepsilon_i^*+r_{\beta,n}^*,\qquad \|r_{\beta,n}^*\|_2=o_{P^*}(1).
$$
For the term $\widehat{U}_{n2}^*(t)$, with
$$
 \widehat B_n^*(t)=t\bigg\{\frac1n\sum_{i=1}^na_{n,i}\dot m_{n,i}\bigg\}\mathbb E^*\{\sin(t\varepsilon_1^*)-\cos(t\varepsilon_1^*)\},
$$
by the proof of Theorem \ref{theorem:1}, we have
\begin{align*}
\widehat{U}_{n2}^*(t) =&\frac{1}{\sqrt{n}}
\sum_{i=1}^n\widehat B_n^*(t)^\top\widehat A_n^{-1}
  \dot m_{n,i}\varepsilon_i^*+R_{n2}^*(t),
\end{align*}
where $\|R_{n2}^*\|_\varphi=o_{P^*}(1)$. Thus, we have
$$
 \widehat U_n^*(t)=\frac1{\sqrt n}\sum_{i=1}^n\psi_{n,i}^*(t)+R_{n,0}^*(t),\qquad \|R_{n,0}^*\|_\varphi=o_{P^*}(1),
$$
where
$$
 \psi_{n,i}^*(t)
 =a_{n,i}\{h_t(\varepsilon_i^*)-\mu_n^*(t)\}
 +\widehat B_n^*(t)^\top\widehat A_n^{-1}
  \dot m_{n,i}\varepsilon_i^* .
$$
The empirical matrices and coefficients are used because the covariates are fixed under the conditional bootstrap law.

We then move from the empirical matrices to the population matrices. Define
$$
\widetilde W_n^*(t)=t\bigg\{\frac1n\sum_{i=1}^ng_0(X_i)\dot m(X_i,\widehat\beta_n)\bigg\}\mathbb E^*\{\sin(t\varepsilon_1^*)-\cos(t\varepsilon_1^*)\}.
$$
Recall that
$$
a_{n,i}=g(X_i)-\bar g,\qquad
g_0(X_i)=g(X_i)-\mathbb E\{g(X)\}.
$$
Since $a_{n,i}-g_0(X_i)=\mathbb E\{g(X)\}-\bar g$, we have
\begin{align*}
D_{n1}^*(t)&:=\frac1{\sqrt n}\sum_{i=1}^n\{a_{n,i}-g_0(X_i)\}\{h_t(\varepsilon_i^*)-\mu_n^*(t)\}\\
&=\{\mathbb E[g(X)]-\bar g\}\times\frac1{\sqrt n}\sum_{i=1}^n\{h_t(\varepsilon_i^*)-\mu_n^*(t)\}.
\end{align*}
Because $|\bar g-\mathbb E\{g(X)\}|=O_p(n^{-1/2})$ and
$$
\mathbb E^*\bigg\|
\frac1{\sqrt n}\sum_{i=1}^n
\{h_t(\varepsilon_i^*)-\mu_n^*(\cdot)\}
\bigg\|_\varphi^2
\leq C\int_{\mathbb R}\varphi(t)\,dt=O(1),
$$
it follows that $\|D_{n1}^*\|_\varphi=o_{P^*}(1)$. Similarly, we have
$$
\widehat B_n^*(t)-\widetilde W_n^*(t)=t\{\mathbb E[g(X)]-\bar g\}\bigg\{\frac1n\sum_{i=1}^n\dot m(X_i,\widehat\beta_n)\bigg\}\mathbb E^*\{\sin(t\varepsilon_1^*)-\cos(t\varepsilon_1^*)\}.
$$
The moment assumptions imply
$$
\int_{\mathbb R}\|\widehat B_n^*(t)-\widetilde W_n^*(t)\|_2^2\varphi(t)\,dt=O_p(p/n).
$$
Moreover,
$$
\bigg\|\frac1{\sqrt n}\sum_{i=1}^n\dot m(X_i,\widehat\beta_n)\varepsilon_i^*\bigg\|_2
=O_{P^*}(\sqrt p).
$$
Hence, since $p^2/n\to0$,
$$
\bigg\|\{\widehat B_n^*-\widetilde W_n^*\}^{\top}\widehat A_n^{-1}\frac1{\sqrt n}\sum_{i=1}^n\dot m(X_i,\widehat\beta_n)\varepsilon_i^*\bigg\|_\varphi=o_{P^*}(1).
$$
Finally, let
$$
d_0(X)=\dot m(X,\widetilde\beta_0),\qquad
A_0=\mathbb E\{d_0(X)d_0(X)^\top\}.
$$
The empirical-matrix convergence and the inverse perturbation identity give
$$
\|\widehat A_n^{-1}-A_0^{-1}\|_{\mathrm{op}}=O_p(p/\sqrt n).
$$
Together with the preceding bound, Assumption~\ref{as:eigen_F}, and the condition $p^3/n\to0$, this yields
$$
\bigg\|\widetilde W_n^*(\cdot)^\top(\widehat A_n^{-1}-A_0^{-1})\frac1{\sqrt n}\sum_{i=1}^n\dot m(X_i,\widehat\beta_n)\varepsilon_i^*\bigg\|_\varphi=o_{P^*}(1).
$$
Absorbing these three negligible terms into the remainder gives
\begin{eqnarray*}
	\widehat{U}_{n}^*(t)
	&=&  \frac{1}{\sqrt{n}}\sum_{i=1}^{n}\big(g_0(X_i)[\cos(t\varepsilon_i^*)+\sin(t\varepsilon_i^*) -  \mathbb{E}^*\{\cos(t \varepsilon_i^*)+\sin(t \varepsilon_i^*)\} ]\\ 
    &&+ \widetilde{W}^*(t)^{\top}\varepsilon_i^*A_0^{-1}\dot{m}(X_i,\widehat{\beta}_n) \big)+R^*_n(t)  \\
	&=:& \widetilde{U}_n^*(t)+R_n^*(t),
\end{eqnarray*}
where the remainder $R_n^*(t)$ satisfies $\int_{\mathbb{R}}	|R_n^*(t)|^2 \varphi(t)dt =o_{P^*}(1).$

By Assumption \ref{as:kernel}, similar to the arguments for proving Lemma S3 in \cite{tan2025weighted} (or Lemma 2 in \cite{neumeyer2009smooth}), we have
\begin{equation*}
    \sup_{t\in\mathbb{R}}|f_e(t)-f_{\varepsilon^*}(t)|=o_p(1),
\end{equation*}
where $f_e(\cdot),f_{\varepsilon^*}(\cdot)$ are the probability density functions of the population-regression residual and smoothed bootstrap residual, respectively. Thus, for any integrable function $g(\cdot)$, we have
\begin{equation*}
    \mathbb{E}^*\{g(\varepsilon^*)\}=\mathbb{E}\{g(e)\}+o_p(1).
\end{equation*}
Then, define $\eta_n^*(t)=\mathbb E^*\{\sin(t\varepsilon_1^*)-\cos(t\varepsilon_1^*)\}$ and $\eta_e(t)=\mathbb E\{\sin(te)-\cos(te)\}$. Then $|\eta_n^*(t)-\eta_e(t)|\leq 2|t|W_2(F_n^*,F_e)$, where $F_n^*=\mathcal L^*(\varepsilon_1^*)$ denote the conditional law of a bootstrap residual, and let $F_e$ denote the distribution of the population residual $e$, $W_2(\cdot,\cdot)$ denotes the Wasserstein distance of order two. Hence, by Assumption~\ref{as:weight},
$$
\int_{\mathbb R}t^2|\eta_n^*(t)-\eta_e(t)|^2\varphi(t)\,dt\leq 4W_2^2(F_n^*,F_e)\int_{\mathbb R}t^4\varphi(t)\,dt=o_p(1).
$$
Define the product-law coefficient
$$
B_e(t,\beta)=t\,\mathbb E\{g_0(X)\dot m(X,\beta)\}
\,\eta_e(t).
$$
The empirical laws of large numbers and the preceding residual-law convergence imply
$$
\bigg\|
\bigg\{\widetilde W_n^*B_e(\cdot,\widehat\beta_n)\bigg\}^{\top}A_0^{-1}\frac1{\sqrt n}\sum_{i=1}^n\dot m(X_i,\widehat\beta_n)\varepsilon_i^*\bigg\|_\varphi=o_{P^*}(1)
$$
in outer probability. Therefore,
\begin{align*}
\widetilde U_n^*(t)=&\frac1{\sqrt n}\sum_{i=1}^n\bigg[g_0(X_i)\{\cos(t\varepsilon_i^*)+\sin(t\varepsilon_i^*)-\mu_e(t)\}\\
&+B_e(t,\widehat\beta_n)^\top A_0^{-1}\dot m(X_i,\widehat\beta_n)\varepsilon_i^*\bigg]+r_n^*(t),
\end{align*}
where $\|r_n^*\|_\varphi=o_{P^*}(1)$. Then, we consider the term 
\begin{align*}
    &\frac{1}{\sqrt{n}}\sum_{i=1}^{n}B_e(t,\widehat{\beta}_n)^{\top}\varepsilon_i^*A_0^{-1}\dot{m}(X_i,\widehat{\beta}_n)-\frac{1}{\sqrt{n}}\sum_{i=1}^{n}B_e(t,\widetilde{\beta}_0)^{\top}\varepsilon_i^*A_0^{-1}\dot{m}(X_i,\widetilde{\beta}_0)\\
    =&\underbrace{\frac{1}{\sqrt{n}}\sum_{i=1}^{n}\{B_e(t,\widehat{\beta}_n)-B_e(t,\widetilde{\beta}_0)\}^{\top}A_0^{-1}\dot{m}(X_i,\widehat{\beta}_n)\varepsilon_i^*}_{\widetilde{U}_{n1}^*(t)}\\
    &+\underbrace{\frac{1}{\sqrt{n}}\sum_{i=1}^{n}B_e(t,\widetilde{\beta}_0)^{\top}A_0^{-1}\{\dot{m}(X_i,\widehat{\beta}_n)-\dot{m}(X_i,\widetilde{\beta}_0)\}\varepsilon_i^*}_{\widetilde{U}_{n2}^*(t)}.
\end{align*}
By the fact that the expected residual $\varepsilon^*$ is $0$, it remains to consider the variance. Thus, we have
\begin{align*}
    &\text{Var}\{\widetilde{U}_{n1}^*(t)+\widetilde{U}_{n2}^*(t)\}\\
    \leq&C\|B_e(t,\widehat{\beta}_n)-B_e(t,\widetilde{\beta}_0)\|_2^2\times\frac{1}{n}\sum_{i=1}^{n}\|\dot{m}(X_i,\widehat{\beta}_n)\|_2^2\cdot\text{Var}(\varepsilon_i^*)\\
    &+C\|B_e(t,\widetilde{\beta}_0)\|_2^2\times\frac{1}{n}\sum_{i=1}^{n}\|\dot{m}(X_i,\widehat{\beta}_n)-\dot{m}(X_i,\widetilde{\beta}_0)\|_2^2\cdot\text{Var}(\varepsilon_i^*)=o_p(1),
\end{align*}
where we use the similar proof of $\widehat{U}_{n2}(t)$ in Theorem \ref{theorem:1} in the last inequality. Thus, it remains to consider the term,
\begin{align*}
    U_n^*(t)=&\frac{1}{\sqrt{n}}\sum_{i=1}^{n}\big(g_0(X_i)[\cos(t\varepsilon_i^*)+\sin(t\varepsilon_i^*) -  \mathbb{E}\{\cos(t e)+\sin(t e)\} ]\\ 
    &+ B_e(t,\widetilde{\beta}_0)^{\top}\varepsilon_i^*A_0^{-1}\dot{m}(X_i,\widetilde{\beta}_0) \big).
\end{align*}
Note that the preceding expansions and expressions do not depend on whether the null model is correctly specified. It therefore remains to analyze the null and alternative cases separately.

\textbf{Case 1:} We first consider the null hypothesis. Under $H_0$, we have $\widetilde\beta_0=\beta_0$ and $e=\varepsilon$, and $\varepsilon$ is independent of $X$. Independence of \(X\) and \(\varepsilon\) gives
$$
B(t)=tW(t)=t\mathbb E\{g_0(X)\dot m(X,\beta_0)\}\mathbb E\{\sin(t\varepsilon)-\cos(t\varepsilon)\}.
$$
Consequently, the process obtained in the preceding expansion becomes
\begin{align*}
U_n^*(t)=\frac1{\sqrt n}\sum_{i=1}^n\bigg[g_0(X_i)\{h_t(\varepsilon_i^*)-\mu(t)\}+B(t)^\top\Sigma^{-1}\dot m(X_i,\beta_0)\varepsilon_i^*\bigg].
\end{align*}
To obtain conditionally centered summands, let $\mu_n^*(t)=\mathbb E^*\{h_t(\varepsilon_1^*)\}$ and define
\begin{align*}
V_n^*(t)=\frac1{\sqrt n}\sum_{i=1}^n\xi_{n,i}^*(t),
\end{align*}
where
$$
\xi_{n,i}^*(t)=g_0(X_i)\{h_t(\varepsilon_i^*)-\mu_n^*(t)\}+B(t)^\top\Sigma^{-1}\dot m(X_i,\beta_0)\varepsilon_i^*.
$$
Conditional on $\mathcal F_n$, the functions $\{\xi_{n,i}^*\}_{i=1}^n$ are independent and satisfy $\mathbb E^*\{\xi_{n,i}^*(t)\}=0$. Note that we have $\|\mu_n^*-\mu\|_\varphi=o_p(1)$. Since
$$
\frac1{\sqrt n}\sum_{i=1}^n g_0(X_i)=O_p(1),
$$
we obtain
$$
\|U_n^*-V_n^*\|_\varphi\leq\bigg|\frac1{\sqrt n}\sum_{i=1}^n g_0(X_i)\bigg|\times\|\mu_n^*-\mu\|_\varphi=o_p(1).
$$
We next calculate the conditional covariance function of $V_n^*$. Define $c_n^*(s,t)=\mathbb E^*[\{h_s(\varepsilon_1^*)-\mu_n^*(s)\}\{h_t(\varepsilon_1^*)-\mu_n^*(t)\}]$ and $d_n^*(t)=\mathbb E^*[\varepsilon_1^*\{h_t(\varepsilon_1^*)-\mu_n^*(t)\}]$, and $v_n^*=\mathbb E^*(\varepsilon_1^{*2})$. To simplify the notation, we denote $\mathbb E_nq=\sum_{i=1}^n q(X_i)/n$ and $\dot m_i=\dot m(X_i,\beta_0)$. Then, the conditional covariance is
\begin{align*}
K_n^*(s,t):=&\operatorname{Cov}^*\{V_n^*(s),V_n^*(t)\}\\
=&c_n^*(s,t)\mathbb E_n\{g_0^2(X)\}\\
&+d_n^*(t)B(s)^\top\Sigma^{-1}\mathbb E_n\{g_0(X)\dot m(X,\beta_0)\}\\
&+d_n^*(s)B(t)^\top\Sigma^{-1}\mathbb E_n\{g_0(X)\dot m(X,\beta_0)\}\\
&+v_n^*B(s)^\top\Sigma^{-1}\mathbb P_n\{\dot m(X,\beta_0)\dot m(X,\beta_0)^\top\}\Sigma^{-1}B(t).
\end{align*}
Let $c(s,t)=\operatorname{Cov}\{h_s(\varepsilon),h_t(\varepsilon)\}$, $d(t)=\mathbb E[\varepsilon\{h_t(\varepsilon)-\mu(t)\}]$ and $v=\mathbb E(\varepsilon^2)$.
The $W_2$-convergence of the bootstrap residual law, together with uniform integrability of the bootstrap second moments, gives
$$
c_n^*(s,t)\stackrel{p}{\longrightarrow}c(s,t),\qquad
d_n^*(t)\stackrel{p}{\longrightarrow}d(t),\qquad
v_n^*\stackrel{p}{\longrightarrow}v.
$$
The empirical laws of large numbers also give
$$
\mathbb E_n\{g_0^2(X)\}\stackrel{p}{\longrightarrow}\mathbb E\{g_0^2(X)\},
$$
$$
\mathbb E_n\{g_0(X)\dot m(X,\beta_0)\}\stackrel{p}{\longrightarrow}\mathbb E\{g_0(X)\dot m(X,\beta_0)\},
$$
and
$$
\mathbb E_n\{\dot m(X,\beta_0)\dot m(X,\beta_0)^\top\}\stackrel{p}{\longrightarrow}\Sigma.
$$
Therefore, for every $s,t\in\Pi$, we have $K_n^*(s,t)-K_n(s,t)\stackrel{p}{\longrightarrow}0$, where
\begin{align*}
K_n(s,t)=&c(s,t)\mathbb E\{g_0^2(X)\}\\
&+d(t)B(s)^\top\Sigma^{-1}\mathbb E\{g_0(X)\dot m(X,\beta_0)\}\\
&+d(s)B(t)^\top\Sigma^{-1}\mathbb E\{g_0(X)\dot m(X,\beta_0)\}\\
&+vB(s)^\top\Sigma^{-1}B(t).
\end{align*}
This is the covariance function in Theorem~\ref{theorem:1}, written using $B(t)=tW(t)$. By assumption, we have $K_n(s,t)\longrightarrow K(s,t).$

It remains to verify conditional weak convergence of the process. Using the similar arguments as the proof of Theorem \ref{theorem:1}, we have
$$
\mathrm{WICM}_n^*\stackrel{d}{\to}\int_{\mathbb R}|\widehat U_n^*(t)|^2\varphi(t)\,dt\stackrel{d}{\to}\int_{\mathbb R}|U_\infty^*(t)|^2\varphi(t)\,dt.
$$

\textbf{Case 2:} We next consider the local alternatives. Under $H_{1n}$, we have $e_n=\varepsilon+r_nS(X)$, where $r_n=o(1)$. Let $F_{e,n}$ denote the marginal distribution of $e_n$, and let $E_n^\dagger\sim F_{e,n}$ be independent of $X$. Define $\mu_{e,n}(t)=\mathbb E\{h_t(E_n^\dagger)\}$ and $\mu_0(t)=\mathbb E\{h_t(\varepsilon)\}$,
and
$$
B_{e,n}(t)=t\,\mathbb E\{g_0(X)\dot m(X,\widetilde\beta_0)\}\mathbb E\{\sin(tE_n^\dagger)-\cos(tE_n^\dagger)\}.
$$
The corresponding null coefficient is
$$
B_0(t)=t\,\mathbb E\{g_0(X)\dot m(X,\widetilde\beta_0)\}\mathbb E\{\sin(t\varepsilon)-\cos(t\varepsilon)\}.
$$
The residual-bootstrap consistency established above gives $W_2(F_n^*,F_{e,n})=o_p(1).$ Moreover, by coupling $E_n^\dagger$ and $\varepsilon$ through $E_n^\dagger=\varepsilon+r_nS(X)$, we have
$$
W_2(F_{e,n},F_\varepsilon)\leq r_n\{\mathbb E S(X)^2\}^{1/2}=o(1).
$$
Therefore, we have $W_2(F_n^*,F_\varepsilon)\leq W_2(F_n^*,F_{e,n})+W_2(F_{e,n},F_\varepsilon)=o_p(1)$. Thus, we have
$$
\int_{\mathbb R}\|B_{e,n}(t)-B_0(t)\|_2^2\varphi(t)\,dt=o(1).
$$
The general bootstrap expansion obtained above gives
\begin{align*}
\widehat U_n^*(t)
=&
\frac1{\sqrt n}\sum_{i=1}^n\bigg[g_0(X_i)\{h_t(\varepsilon_i^*)-\mu_{e,n}(t)\}\\
&+B_{e,n}(t)^\top \Sigma^{-1}\dot m(X_i,\widetilde\beta_0)\varepsilon_i^*\bigg]+R_{n,\mathrm{loc}}^*(t),
\end{align*}
where $\|R_{n,\mathrm{loc}}^*\|_\varphi=o_{P^*}(1)$. Define the null-equivalent bootstrap process
\begin{align*}
V_n^*(t)=\frac1{\sqrt n}\sum_{i=1}^n\bigg[g_0(X_i)\{h_t(\varepsilon_i^*)-\mu_0(t)\}+B_0(t)^\top \Sigma^{-1}\dot m(X_i,\widetilde\beta_0)\varepsilon_i^*\bigg].
\end{align*}
The difference between the two leading processes can be written as
$$
\widehat U_n^*(t)-V_n^*(t)=D_{n1}^*(t)+D_{n2}^*(t),
$$
where
$$
D_{n1}^*(t)=
\bigg\{\frac1{\sqrt n}\sum_{i=1}^n g_0(X_i)\bigg\}\times\{\mu_0(t)-\mu_{e,n}(t)\},
$$
and
$$
D_{n2}^*(t)=\{B_{e,n}(t)-B_0(t)\}^\top \Sigma^{-1}\frac1{\sqrt n}\sum_{i=1}^n\dot m(X_i,\widetilde\beta_0)\varepsilon_i^*.
$$
Since
$$
\frac1{\sqrt n}\sum_{i=1}^n g_0(X_i)=O_p(1),
$$
the preceding bound for the centering functions yields $\|D_{n1}^*\|_\varphi=o_p(1)$. Conditional on $\mathcal F_n$, we know that $D_{n2}^*(t)$ has mean zero. Furthermore,
\begin{align*}
\mathbb E^*\|D_{n2}^*\|_\varphi^2\leq&
C\operatorname{Var}^*(\varepsilon_1^*)\int_{\mathbb R}\|B_{e,n}(t)-B_0(t)\|_2^2\varphi(t)\,dt=o_p(1),
\end{align*}
where Assumption~\ref{as:eigen_F} and the empirical law of large numbers are used in the inequality. Conditional Markov's inequality therefore gives $\|D_{n2}^*\|_\varphi=o_{P^*}(1)$. Consequently, $\|\widehat U_n^*-V_n^*\|_\varphi=o_{P^*}(1)$.

Then, by using the identical arguments used in Case 1, we therefore obtain
$$
\mathrm{WICM}_n^*\stackrel{d}{\to}\int_{\mathbb R}|\widehat U_n^*(t)|^2\varphi(t)\,dt\stackrel{d}{\to}\int_{\mathbb R}|U^*_\infty(t)|^2\varphi(t)\,dt.
$$

\textbf{Case 3:} We finally consider the global alternative $H_1$. In this case, $e=S(X)+\varepsilon$. Let $F_e$ denote the marginal distribution of $e$. Since we have $\mathbb E\{S(X)\}=0$ and $\mathbb E(\varepsilon)=0$, thus the residual $e$ is centered. The consistency of the smoothed residual bootstrap established above gives $W_2(F_n^*,F_e)=o_p(1)$, where $F_n^*=\mathcal L^*(\varepsilon_1^*)$.

An important distinction from the null and local-alternative cases is that, conditionally on the data, the bootstrap residuals are generated independently of the fixed covariates. Therefore, under $H_1$, the limiting bootstrap law is the product law $P_X\otimes F_e$, rather than the original joint law of $(X,e)$. Let
$$
d_0(X)=\dot m(X,\widetilde\beta_0),\qquad
A_0=\mathbb E\{d_0(X)d_0(X)^\top\},
\qquad
M_0=\mathbb E\{g_0(X)d_0(X)\}.
$$
For the bootstrap least-squares expansion, note that $A_0$ is uniformly nonsingular; that is, for some constants $0<c<C<\infty$,
$$
c\leq\lambda_{\min}(A_0)\leq\lambda_{\max}(A_0)\leq C.
$$
Let $E^\dagger\sim F_e$ be independent of $X$, and define
$$
h_t(u)=\cos(tu)+\sin(tu),\qquad
\mu_e(t)=\mathbb E\{h_t(E^\dagger)\},
$$
$$
\ell_e(t)=\mathbb E\{\sin(tE^\dagger)-\cos(tE^\dagger)\},
\qquad
B_e(t)=tM_0\ell_e(t).
$$
Recall the preceding conditional representation
$$
\widehat U_n^*(t)=\frac1{\sqrt n}\sum_{i=1}^n\psi_{n,i}^*(t)+R_n^*(t),
\qquad
\|R_n^*\|_\varphi=o_{P^*}(1),
$$
where
$$
\psi_{n,i}^*(t)=a_{n,i}\{h_t(\varepsilon_i^*)-\mu_n^*(t)\}+\widehat B_n^*(t)^\top A_0^{-1}d_{n,i}\varepsilon_i^*,
$$
with
$$
a_{n,i}=g(X_i)-\bar g,\qquad
d_{n,i}=\dot m(X_i,\widehat\beta_n),\qquad
\mu_n^*(t)=\mathbb E^*\{h_t(\varepsilon_1^*)\},
$$
and
$$
\widehat B_n^*(t)=t\bigg\{\frac1n\sum_{i=1}^na_{n,i}d_{n,i}\bigg\}\mathbb E^*\{\sin(t\varepsilon_1^*)-\cos(t\varepsilon_1^*)\}.
$$
The consistency of $\widehat\beta_n$, the empirical laws of large numbers, and $W_2(F_n^*,F_e)=o_p(1)$ imply
$$
\frac1n\sum_{i=1}^na_{n,i}d_{n,i}-M_0=o_p(1),
$$
and thus
$$
\int_{\mathbb R}\|\widehat B_n^*(t)-B_e(t)\|_2^2\varphi(t)\,dt=o_p(1).
$$
The limiting summand under the product law is
$$
\psi_t^e=g_0(X)\{h_t(E^\dagger)-\mu_e(t)\}+B_e(t)^\top A_0^{-1}d_0(X)E^\dagger.
$$
Consequently, we define $K_n^e(s,t)=\mathbb E\{\psi_s^e\psi_t^e\}$, where possible dependence on $n$ through $p=p_n$ is suppressed. For an explicit expression, let
$$
c_e(s,t)=\mathbb E\big[\{h_s(E^\dagger)-\mu_e(s)\}\{h_t(E^\dagger)-\mu_e(t)\}\big],
$$
$$
d_e(t)=\mathbb E[\mathbb E^\dagger\{h_t(E^\dagger)-\mu_e(t)\}],
\qquad
\sigma_e^2=\mathbb E\{(E^\dagger)^2\}.
$$
Using the independence of $E^\dagger$ and $X$, we obtain
\begin{align*}
K_n^e(s,t)=&\mathbb E\{g_0(X)^2\}\,c_e(s,t)+B_e(t)^\top A_0^{-1}M_0\,d_e(s)+B_e(s)^\top A_0^{-1}M_0\,d_e(t)\\
&+\sigma_e^2B_e(s)^\top A_0^{-1}B_e(t).
\end{align*}

Let
$$
C_n^*(s,t)=\frac1n\sum_{i=1}^n\mathbb E^*\{\psi_{n,i}^*(s)\psi_{n,i}^*(t)\}.
$$
The preceding convergence results imply, for every fixed $s,t$, we have $C_n^*(s,t)-K_n^e(s,t)\to_p 0$. Assume, as required for the stated Gaussian limit, that $K_n^e(s,t)\longrightarrow K^e(s,t)$ pointwise, where $K^e$ is the covariance function of an $L^2(\varphi)$-valued mean-zero Gaussian process $U_\infty^e$. Then, by using the identical arguments used in Case 1, we can obtain,
$$
\mathrm{WICM}_n^*\stackrel{d}{\to}\int_{\mathbb R}|U_\infty^e(t)|^2\varphi(t)\,dt.
$$
In particular, $\mathrm{WICM}_n^*=O_{P^*}(1).$
\end{proof}

\begin{proof}[Proof of Corollary~\ref{cor:bootstrap_test}]
Let
$$
 T_0=\int_{\mathbb R}|U_\infty(t)|^2\varphi(t)\,dt
$$
and let $q_{1-\alpha}$ be its $(1-\alpha)$-quantile. Conditional bootstrap convergence gives
$$
 \widehat c_\alpha\to_p q_{1-\alpha}.
$$
Under $H_0$, Theorem~\ref{theorem:1} and Slutsky's theorem imply
$$
 \mathbb P(\mathrm{WICM}_n>\widehat c_\alpha)\to\alpha.
$$
Under a fixed detectable alternative, $\mathrm{WICM}_n/n\to_pc_1>0$, while $\widehat c_\alpha=O_p(1)$; hence the rejection probability tends to one. If $r_n=n^{-\gamma}$, for $0<\gamma<1/2$, then
$$
 \frac{\mathrm{WICM}_n}{nr_n^2}\to_pc_2>0,\qquad nr_n^2\to\infty,
$$
and again $\widehat c_\alpha=O_p(1)$. Finally, when $r_n=n^{-1/2}$,
$$
 \mathbb P(\mathrm{WICM}_n>\widehat c_\alpha)\longrightarrow
 \mathbb P\bigg\{
 \int_{\mathbb R}|U_\infty(t)+tK^{(2)}(t)|^2
 \varphi(t)\,dt>q_{1-\alpha}\bigg\}>0.
$$
This proves all three assertions.
\end{proof}

\section{Additional Numerical Results}
\label{sec:additional_num}
This section presents additional simulation results under the covariance structure $\Sigma_2=(2^{-|i-j|})_{p\times p}$, with all other settings kept the same as in the main text. The results are qualitatively similar and support the same conclusions as those reported in the main paper.

\begin{table}[ht!]
\caption{Empirical sizes and powers of $\text{WICM}_n^{(1)}$, $\text{WICM}_n^{(2)}$, $\text{CM}_n$, $\text{TCM}_n$, $\text{ICM}_n$, and $\text{PCvM}_n$ for model $H_{1}$ and covariance structure $\Sigma_2$.}
\label{table:h1_2}
	\centering
	{\small\tiny\hspace{5cm}
		\renewcommand{\arraystretch}{1}\tabcolsep 0.2cm
		\begin{tabular}{cccccccccccc}
			\hline
			&\multicolumn{1}{c}{a} &\multicolumn{1}{c}{n=100} &\multicolumn{1}{c}{n=100} &\multicolumn{1}{c}{n=100}&\multicolumn{1}{c}{n=100} &\multicolumn{1}{c}{n=100}  &\multicolumn{1}{c}{n=200} &\multicolumn{1}{c}{n=400}
			&\multicolumn{1}{c}{n=600}  \\
			&&\multicolumn{1}{c}{p=2} &\multicolumn{1}{c}{p=4}
			&\multicolumn{1}{c}{p=6}
			&\multicolumn{1}{c}{p=8} &\multicolumn{1}{c}{p=10} &\multicolumn{1}{c}{p=14} &\multicolumn{1}{c}{p=19}  &\multicolumn{1}{c}{p=22}\\
			\hline	
			$\text{WICM}_n^{(1)}$
			&0.0       &0.062 &0.049 &0.054 &0.052 &0.055 &0.051 &0.061 &0.046\\
			&0.1       &0.086 &0.099 &0.099 &0.105 &0.124 &0.155 &0.269 &0.368\\
			&0.2       &0.249 &0.256 &0.257 &0.297 &0.272 &0.502 &0.771 &0.891\\
			&0.3       &0.462 &0.462 &0.521 &0.541 &0.584 &0.823 &0.973 &0.997\\
			&0.4       &0.699 &0.730 &0.753 &0.783 &0.778 &0.974 &1.000 &1.000\\
			&0.5       &0.876 &0.907 &0.929 &0.924 &0.938 &0.998 &1.000 &1.000\\	
			\hline
			$\text{WICM}_n^{(2)}$
			&0.0       &0.046 &0.049 &0.052 &0.043 &0.047 &0.054 &0.048 &0.059\\
			&0.1       &0.058 &0.069 &0.059 &0.057 &0.058 &0.066 &0.071 &0.063\\
			&0.2       &0.106 &0.077 &0.078 &0.073 &0.078 &0.075 &0.097 &0.096\\
			&0.3       &0.226 &0.156 &0.108 &0.117 &0.100 &0.104 &0.127 &0.161\\
			&0.4       &0.366 &0.227 &0.168 &0.120 &0.149 &0.147 &0.187 &0.191\\
			&0.5       &0.521 &0.291 &0.247 &0.216 &0.180 &0.212 &0.262 &0.355\\
			\hline
            $\text{CM}_n$
            & 0.0 & 0.045 & 0.061 & 0.063 & 0.061 & 0.071 & 0.046 & 0.062 & 0.078 \\
            & 0.1 & 0.056 & 0.067 & 0.054 & 0.053 & 0.058 & 0.065 & 0.068 & 0.060 \\
            & 0.2 & 0.086 & 0.096 & 0.081 & 0.096 & 0.074 & 0.084 & 0.088 & 0.095 \\
            & 0.3 & 0.159 & 0.113 & 0.091 & 0.102 & 0.086 & 0.130 & 0.130 & 0.138 \\
            & 0.4 & 0.247 & 0.176 & 0.144 & 0.166 & 0.133 & 0.140 & 0.169 & 0.231 \\
            & 0.5 & 0.386 & 0.305 & 0.201 & 0.180 & 0.198 & 0.216 & 0.231 & 0.287 \\
            \hline
            $\text{TCM}_n$
            & 0.0   & 0.052 & 0.066 & 0.063 & 0.057 & 0.053 & 0.065 & 0.062 & 0.045 \\
            & 0.1   & 0.060 & 0.063 & 0.070  & 0.044 & 0.080  & 0.057 & 0.049 & 0.044 \\
            & 0.2   & 0.076 & 0.066 & 0.070  & 0.063 & 0.065 & 0.057 & 0.058 & 0.066 \\
            & 0.3   & 0.078 & 0.071 & 0.069 & 0.060  & 0.070  & 0.055 & 0.058 & 0.051 \\
            & 0.4   & 0.086 & 0.093 & 0.074 & 0.085 & 0.072 & 0.064 & 0.053 & 0.068 \\
            & 0.5   & 0.105 & 0.069 & 0.074 & 0.082 & 0.069 & 0.049 & 0.058 & 0.050  \\
            \hline
			$ \text{ICM}_n $
			&0.0       &0.047 &0.060 &0.037 &0.004 &0.000 &0.000 &0.000 &0.000\\
			&0.1       &0.062 &0.080 &0.034 &0.006 &0.000 &0.000 &0.000 &0.000\\
			&0.2       &0.084 &0.093 &0.043 &0.008 &0.002 &0.000 &0.000 &0.000\\
			&0.3       &0.129 &0.097 &0.064 &0.008 &0.001 &0.000 &0.000 &0.000\\
			&0.4       &0.227 &0.178 &0.091 &0.021 &0.002 &0.000 &0.000 &0.000\\
			&0.5       &0.346 &0.238 &0.131 &0.004 &0.002 &0.000 &0.000 &0.000\\
			\hline
			$\text{PCvM}_n$
	       &0.0       &0.040 &0.060 &0.056 &0.080 &0.068 &0.057 &0.057 &0.063\\
           &0.1       &0.045 &0.056 &0.056 &0.059 &0.063 &0.062 &0.075 &0.058\\
           &0.2       &0.070 &0.048 &0.060 &0.067 &0.052 &0.066 &0.053 &0.051\\
           &0.3       &0.075 &0.051 &0.049 &0.063 &0.054 &0.057 &0.058 &0.047\\
           &0.4       &0.067 &0.047 &0.049 &0.059 &0.060 &0.062 &0.065 &0.046\\
           &0.5       &0.081 &0.058 &0.062 &0.057 &0.077 &0.053 &0.042 &0.051\\
			\hline
	\end{tabular}}
\end{table}

\begin{table}[ht!]\caption{Empirical sizes and powers of $\text{WICM}_n^{(1)}$, $\text{WICM}_n^{(2)}$, $\text{CM}_n$, $\text{TCM}_n$, $\text{ICM}_n$, and $\text{PCvM}_n$ for model $H_{2}$ and covariance structure $\Sigma_2$.}
\label{table:h2_2}
	\centering
	{\small\tiny\hspace{5cm}
		\renewcommand{\arraystretch}{1}\tabcolsep 0.2cm
		\begin{tabular}{cccccccccccc}
			\hline
			&\multicolumn{1}{c}{a} &\multicolumn{1}{c}{n=100} &\multicolumn{1}{c}{n=100} &\multicolumn{1}{c}{n=100}&\multicolumn{1}{c}{n=100} &\multicolumn{1}{c}{n=100}  &\multicolumn{1}{c}{n=200} &\multicolumn{1}{c}{n=400}
			&\multicolumn{1}{c}{n=600}  \\
			&&\multicolumn{1}{c}{p=2} &\multicolumn{1}{c}{p=4}
			&\multicolumn{1}{c}{p=6}
			&\multicolumn{1}{c}{p=8} &\multicolumn{1}{c}{p=10} &\multicolumn{1}{c}{p=14} &\multicolumn{1}{c}{p=19}  &\multicolumn{1}{c}{p=22}\\
			\hline	
			$\text{WICM}_n^{(1)}$
			&0.0    &0.052 &0.054 &0.050 &0.058 &0.073  &0.065 &0.061 &0.056 \\
			&0.1    &0.356 &0.489 &0.547 &0.608 &0.633  &0.902 &0.998 &1.000 \\
			&0.2    &0.805 &0.933 &0.975 &0.979  &0.981  &1.000 &1.000 &1.000\\
			&0.3    &0.977 &0.997 &1.000 &1.000 &1.000  &1.000 &1.000 &1.000\\
			&0.4    &0.998 &1.000 &1.000 &1.000 &1.000  &1.000 &1.000 &1.000\\
			&0.5    &0.999 &1.000 &1.000 &1.000 &1.000  &1.000 &1.000 &1.000\\	
			\hline
			$\text{WICM}_n^{(2)}$
            & 0.0    & 0.049 & 0.047 & 0.041 & 0.056 & 0.060 & 0.052 & 0.047 & 0.057 \\
            & 0.1    & 0.119 & 0.237 & 0.247 & 0.311 & 0.333 & 0.669 & 0.943 & 0.998 \\
            & 0.2    & 0.419 & 0.664 & 0.716 & 0.760 & 0.756 & 0.980 & 1.000 & 1.000 \\
            & 0.3    & 0.731 & 0.884 & 0.900 & 0.913 & 0.904 & 0.996 & 1.000 & 1.000 \\
            & 0.4    & 0.882 & 0.957 & 0.938 & 0.912 & 0.911 & 0.998 & 1.000 & 1.000 \\
            & 0.5    & 0.943 & 0.959 & 0.939 & 0.930 & 0.906 & 0.996 & 1.000 & 1.000  \\
            \hline
            $\text{CM}_n$
            & 0.0    & 0.058 & 0.041 & 0.062 & 0.061 & 0.050 & 0.064 & 0.065 & 0.078 \\
            & 0.1    & 0.129 & 0.247 & 0.296 & 0.327 & 0.355 & 0.674 & 0.920 & 0.995 \\
            & 0.2    & 0.429 & 0.643 & 0.706 & 0.737 & 0.745 & 0.967 & 0.999 & 1.000  \\
            & 0.3    & 0.719 & 0.860  & 0.871 & 0.878 & 0.858 & 0.981 & 0.999 & 1.000  \\
            & 0.4    & 0.879 & 0.935 & 0.921 & 0.893 & 0.869 & 0.978 & 1.000 & 1.000  \\
            & 0.5    & 0.93  & 0.950  & 0.921 & 0.875 & 0.854 & 0.983 & 0.999 & 1.000  \\
             \hline
            $\text{TCM}_n$
            & 0.0   & 0.052 & 0.066 & 0.050  & 0.055 & 0.055 & 0.065 & 0.062 & 0.045 \\
            & 0.1   & 0.111 & 0.124 & 0.134 & 0.171 & 0.153 & 0.196 & 0.269 & 0.381 \\
            & 0.2   & 0.215 & 0.242 & 0.245 & 0.273 & 0.235 & 0.310  & 0.391 & 0.575 \\
            & 0.3   & 0.302 & 0.338 & 0.328 & 0.311 & 0.256 & 0.374 & 0.464 & 0.623 \\
            & 0.4   & 0.363 & 0.397 & 0.354 & 0.346 & 0.322 & 0.395 & 0.472 & 0.645 \\
            & 0.5   & 0.388 & 0.464 & 0.408 & 0.352 & 0.339 & 0.433 & 0.522 & 0.663 \\
            \hline
			$ \text{ICM}_n $
			&0.0       &0.047 &0.041 &0.026 &0.011 &0.001 &0.000 &0.000 &0.000\\
			&0.1       &0.101 &0.134 &0.123 &0.029 &0.002 &0.000 &0.000 &0.000\\
			&0.2       &0.299 &0.495 &0.427 &0.148 &0.015 &0.000 &0.000 &0.000\\
			&0.3       &0.636 &0.831 &0.761 &0.328 &0.032 &0.000 &0.000 &0.000\\
			&0.4       &0.837 &0.958 &0.868 &0.477 &0.038 &0.000 &0.000 &0.000\\
			&0.5       &0.953 &0.976 &0.919 &0.499 &0.055 &0.000 &0.000 &0.000\\
			\hline
			$\text{PCvM}_n$
			&0.0       &0.049 &0.057 &0.076 &0.065 &0.068 &0.058 &0.057 &0.059\\
			&0.1       &0.148 &0.274 &0.382 &0.447 &0.528 &0.894 &0.995 &1.000\\
			&0.2       &0.452 &0.735 &0.862 &0.925 &0.945 &1.000 &1.000 &1.000\\
			&0.3       &0.710 &0.957 &0.988 &0.995 &0.996 &1.000 &1.000 &1.000\\
			&0.4       &0.918 &0.996 &0.999 &1.000 &1.000 &1.000 &1.000 &1.000\\
			&0.5       &0.977 &1.000 &1.000 &1.000 &1.000 &1.000 &1.000 &1.000\\
			\hline
	\end{tabular}}
\end{table}

\begin{table}[ht!]\caption{Empirical sizes and powers of $\text{WICM}_n^{(1)}$, $\text{WICM}_n^{(2)}$, $\text{CM}_n$, $\text{TCM}_n$, $\text{ICM}_n$, and $\text{PCvM}_n$ for model $H_{3}$ and covariance structure $\Sigma_2$.}
\label{table:h3_2}
	\centering
	{\small\tiny\hspace{5cm}
		\renewcommand{\arraystretch}{1}\tabcolsep 0.2cm
		\begin{tabular}{cccccccccccc}
			\hline
			&\multicolumn{1}{c}{a} &\multicolumn{1}{c}{n=100} &\multicolumn{1}{c}{n=100} &\multicolumn{1}{c}{n=100}&\multicolumn{1}{c}{n=100} &\multicolumn{1}{c}{n=100}  &\multicolumn{1}{c}{n=200} &\multicolumn{1}{c}{n=400}
			&\multicolumn{1}{c}{n=600}  \\
			&&\multicolumn{1}{c}{p=2} &\multicolumn{1}{c}{p=4}
			&\multicolumn{1}{c}{p=6}
			&\multicolumn{1}{c}{p=8} &\multicolumn{1}{c}{p=10} &\multicolumn{1}{c}{p=14} &\multicolumn{1}{c}{p=19}  &\multicolumn{1}{c}{p=22}\\
			\hline	
			$\text{WICM}_n^{(1)}$
			& 0.0 & 0.053 & 0.047 & 0.051 & 0.053 & 0.061 & 0.057 & 0.064 & 0.059 \\
            & 0.05 & 0.310 & 0.405 & 0.438 & 0.482 & 0.504 & 0.775 & 0.976 & 0.999 \\
            & 0.10 & 0.692 & 0.844 & 0.913 & 0.926 & 0.937 & 0.999 & 1.000 & 1.000 \\
            & 0.15 & 0.897 & 0.986 & 0.996 & 0.997 & 0.999 & 1.000 & 1.000 & 1.000 \\
            & 0.20 & 0.976 & 1.000 & 1.000 & 1.000 & 1.000 & 1.000 & 1.000 & 1.000 \\
            & 0.25 & 0.994 & 1.000 & 1.000 & 1.000 & 1.000 & 1.000 & 1.000 & 1.000 \\	
			\hline
			$\text{WICM}_n^{(2)}$
             & 0.0      & 0.049  & 0.041  & 0.051  & 0.042  & 0.053  & 0.052  & 0.047  & 0.058  \\
             & 0.05   & 0.118  & 0.240   & 0.318  & 0.381  & 0.403  & 0.766  & 0.981  & 1.000       \\
             & 0.10    & 0.439  & 0.735  & 0.829  & 0.880   & 0.865  & 0.995  & 1.000       & 1.000       \\
             & 0.15   & 0.804  & 0.946  & 0.964  & 0.963  & 0.962  & 1.000      & 1.000       & 1.000       \\
             & 0.20    & 0.922  & 0.994  & 0.986  & 0.980   & 0.978  & 1.000   & 1.000       & 1.000       \\
             & 0.25   & 0.991  & 0.993  & 0.990   & 0.991  & 0.976  &1.000       & 1.000       & 1.000       \\
            \hline
            $\text{CM}_n$
            & 0.0    & 0.059  & 0.055  & 0.054  & 0.064  & 0.064  & 0.061  & 0.056  & 0.063  \\
            & 0.05   & 0.147  & 0.257  & 0.333  & 0.429  & 0.389  & 0.764  & 0.976  & 1.000       \\
            & 0.10   & 0.444  & 0.732  & 0.814  & 0.839  & 0.858  & 0.994  & 1.000       & 1.000       \\
            & 0.15   & 0.797  & 0.936  & 0.948  & 0.971  & 0.958  & 0.999  & 1.000       & 1.000       \\
            & 0.20   & 0.906  & 0.982  & 0.977  & 0.972  & 0.965  & 0.998  & 1.000       & 1.000       \\
            & 0.25   & 0.984  & 0.991  & 0.989  & 0.97   & 0.967  & 1.000       & 1.000       & 1.000       \\
            \hline
            $ \text{TCM}_n $
             & 0.0      & 0.068  & 0.059  & 0.053  & 0.060   & 0.073  & 0.060   & 0.066  & 0.063  \\
             & 0.05   & 0.128  & 0.169  & 0.162  & 0.183  & 0.173  & 0.291  & 0.362  & 0.528  \\
             & 0.10    & 0.205  & 0.305  & 0.313  & 0.302  & 0.298  & 0.453  & 0.613  & 0.827  \\
             & 0.15   & 0.312  & 0.421  & 0.452  & 0.411  & 0.377  & 0.523  & 0.739  & 0.920   \\
             & 0.20    & 0.400    & 0.476  & 0.480   & 0.418  & 0.443  & 0.575 & 0.784  & 0.951  \\
             & 0.25   & 0.453  & 0.532  & 0.516  & 0.494  & 0.491  & 0.605  & 0.841  & 0.959  \\
            \hline
			$ \text{ICM}_n $
            & 0.0    & 0.054  & 0.048  & 0.039  & 0.013  & 0.012  & 0.000      & 0.000      & 0.000      \\
            & 0.1    & 0.137  & 0.162  & 0.180   & 0.131  & 0.072  & 0.034  & 0.011  & 0.018  \\
            & 0.2    & 0.429  & 0.571  & 0.647  & 0.557  & 0.458  & 0.487  & 0.676  & 0.872  \\
            & 0.3    & 0.784  & 0.922  & 0.928  & 0.863  & 0.765  & 0.817  & 0.927  & 0.989  \\
            & 0.4    & 0.954  & 0.986  & 0.987  & 0.959  & 0.86   & 0.891  & 0.954  & 0.989  \\
            & 0.5    & 0.997  & 1.000      & 0.998  & 0.963  & 0.883  & 0.907  & 0.961  & 0.988  \\
			\hline
			$\text{PCvM}_n $
		& 0.0      & 0.056  & 0.061  & 0.070   & 0.068  & 0.069  & 0.052  & 0.059  & 0.059  \\
 & 0.05   & 0.195  & 0.298  & 0.405  & 0.439  & 0.484  & 0.806  & 0.983  & 1.000      \\
 & 0.10   & 0.576  & 0.773  & 0.883  & 0.904  & 0.927  & 0.999  & 1.000  &1.000        \\
 & 0.15   & 0.843  & 0.974  & 0.987  & 0.99   & 0.997  & 1.000       & 1.000       & 1.000       \\
 & 0.20   & 0.975  & 0.998  & 1.000   & 1.000       & 1.000       & 1.000      & 1.000       & 1.000       \\
 & 0.25   & 0.997  & 1.000       & 1.000       & 1.000       & 1.000       & 1.000       & 1.000       & 1.000       \\
			\hline
	\end{tabular}}
\end{table}

\begin{table}[ht!]\caption{Empirical sizes and powers of $\text{WICM}_n^{(1)}$, $\text{WICM}_n^{(2)}$, $\text{CM}_n$, $\text{TCM}_n$, $\text{ICM}_n$, and $\text{PCvM}_n$ for model $H_{4}$ and covariance structure $\Sigma_2$.}
\label{table:h4_2}
	\centering
	{\tiny\hspace{5cm}
		\renewcommand{\arraystretch}{1}\tabcolsep 0.6cm
		\begin{tabular}{cccccccccccc}
			\hline
			&\multicolumn{1}{c}{a} &\multicolumn{1}{c}{n=100}  &\multicolumn{1}{c}{n=200} &\multicolumn{1}{c}{n=400}
			&\multicolumn{1}{c}{n=600}  \\
			&&\multicolumn{1}{c}{p=10} &\multicolumn{1}{c}{p=14} &\multicolumn{1}{c}{p=19}  &\multicolumn{1}{c}{p=22}\\
			\hline
			$\text{WICM}_n^{(1)}$
           &0.0  & 0.046 & 0.052 & 0.059 & 0.053 \\ 
           &0.1  & 0.319 & 0.379 & 0.481 & 0.506 \\ 
           &0.2  & 0.490 & 0.552 & 0.631 & 0.697 \\ 
           &0.3  & 0.541 & 0.593 & 0.677 & 0.696 \\ 
           &0.4  & 0.594 & 0.659 & 0.684 & 0.717 \\ 
           &0.5  & 0.633 & 0.677 & 0.706 & 0.740 \\ 	
			\hline
			$\text{WICM}_n^{(2)}$
           & 0.0  & 0.050  & 0.051 & 0.057 & 0.049 \\
           & 0.1  & 0.074 & 0.081 & 0.074 & 0.061  \\
           & 0.2  & 0.088 & 0.128 & 0.128 & 0.193 \\
           & 0.3  & 0.180  & 0.236 & 0.373 & 0.469  \\
           & 0.4  & 0.263 & 0.363 & 0.570  & 0.719  \\
           & 0.5  & 0.341 & 0.524 & 0.756 & 0.862  \\
            \hline
            $\text{CM}_n$
          & 0.0   & 0.074 & 0.059 & 0.057 & 0.060   \\
          & 0.1   & 0.068 & 0.062 & 0.066 & 0.079 \\
          & 0.2   & 0.093 & 0.122 & 0.144 & 0.169  \\
          & 0.3   & 0.188 & 0.240 & 0.305 & 0.410  \\
          & 0.4   & 0.272 & 0.395 & 0.548 & 0.698  \\
          & 0.5   & 0.361 & 0.482 & 0.726 & 0.846 \\
            \hline
            $\text{TCM}_n$
            & 0.0    & 0.058 & 0.062 & 0.060 & 0.049 \\
            & 0.1    & 0.065 & 0.066 & 0.069 & 0.090 \\
            & 0.2    & 0.089 & 0.114 & 0.117 & 0.142 \\
            & 0.3    & 0.130  & 0.183 & 0.213 & 0.215 \\
            & 0.4    & 0.182  & 0.223 & 0.282 & 0.331 \\
            & 0.5    & 0.210  & 0.268 & 0.350  & 0.430  \\
            \hline
			$ \text{ICM}_n $
			& 0.0    & 0.006 & 0.000 & 0.000 & 0.000 \\
            & 0.1    & 0.013 & 0.000 & 0.000 & 0.000 \\
            & 0.2    & 0.023 & 0.000 & 0.000 & 0.000 \\
            & 0.3    & 0.047 & 0.000 & 0.000 & 0.000 \\
            & 0.4    & 0.054 & 0.005 & 0.000 & 0.000 \\
            & 0.5    & 0.072 & 0.000 & 0.000 & 0.000  \\
			\hline
			$\text{PCvM}_n$
            & 0.0  & 0.067 & 0.068 & 0.057 & 0.055  \\
            & 0.1  & 0.077 & 0.109 & 0.151 & 0.184  \\
            & 0.2  & 0.143 & 0.206 & 0.364 & 0.450   \\
            & 0.3  & 0.191 & 0.291 & 0.506 & 0.691   \\
            & 0.4  & 0.239 & 0.382 & 0.615 & 0.792   \\
            & 0.5  & 0.254 & 0.432 & 0.698 & 0.820    \\
			\hline			
	\end{tabular}}
\end{table}

\section{Technical Lemmas}
\label{sec:lemma}

\begin{lemma}[Proposition S1 in \cite{tan2025weighted}]
\label{lem:beta}
Under Assumptions \ref{as:model}--\ref{as:inter} in the main text. If \(p^2=o(n)\), then
\begin{equation*}
\|\widehat{\beta}-\widetilde{\beta}_0\|_2
=
O_p\!\left(\sqrt{\frac{p}{n}}\right).
\end{equation*}
Furthermore, if \(p^3\log n=o(n)\), then
\begin{equation*}
\sqrt{n}\,(\widehat{\beta}-\widetilde{\beta}_0)
=
\frac{1}{\sqrt{n}}
\sum_{i=1}^n
e_i \,\Sigma^{-1}\dot{m}(X_i,\widetilde{\beta}_0)
+
O_p\!\left(\sqrt{\frac{p^3}{n}}\right).
\end{equation*}
\end{lemma}

\bibliographystyle{apalike}
\bibliography{ref}

\end{document}

%% file: definition.tex
\usepackage{amsmath,amssymb}

\usepackage{mathrsfs}
\usepackage{bm, xcolor}
\usepackage{mathrsfs}
\usepackage{verbatim}
\usepackage{longtable}
\usepackage{threeparttable}
\usepackage{color}
\usepackage{amsmath}

\usepackage{enumerate}
\usepackage{amssymb,bbm}
\usepackage{amsbsy}
\usepackage{amsmath}
\usepackage{amsthm}
\usepackage{amsfonts}
\usepackage{latexsym}
\usepackage{xifthen} 
\usepackage{color}
\usepackage{manfnt}
\usepackage{ifthen}
\usepackage{bm}
\usepackage{caption}
\usepackage{subcaption}

\makeatletter
\def\singlespace{\def\baselinestretch{1}\@normalsize}

\newtheorem{lemma}{Lemma}
\newtheorem{theorem}{Theorem}

\newtheorem{corollary}{Corollary}
\@addtoreset{equation}{section}

\newtheorem{assumption}{Assumption}

\def\singlespace{\def\baselinestretch{1}\@normalsize}

\makeatother

\def\newpage{\vfill\eject}

\newdimen\biblioindent    \biblioindent=30pt

 at 7truept

\def\beqr{\begin{eqnarray}}
	\def\eeqr{\end{eqnarray}}
\def\beqrs{\begin{eqnarray*}}
	\def\eeqrs{\end{eqnarray*}}

\def\beq{\begin{equation}}
\def\eeq{\end{equation}}
\def\beqn{\begin{eqnarray}}
\def\eeqn{\end{eqnarray}}
\def\beqnn{\begin{eqnarray*}}
\def\eeqnn{\end{eqnarray*}}

